\documentclass[12pt]{article}
\pdfoutput=1

\usepackage[utf8]{inputenc}
\usepackage[UKenglish]{babel}
\usepackage[iso,UKenglish]{isodate}
\usepackage{csquotes}

\usepackage{enumitem} %for customizing labels in enumerate environments

\newlength{\alphabet}
\usepackage[a4paper,textwidth=2.8\alphabet]{geometry}

\usepackage{microtype}

\usepackage{mathtools}
\usepackage{amsfonts}
\usepackage{amsmath}
\usepackage{amssymb}
\usepackage{amsthm}
\usepackage{thmtools, thm-restate}
\usepackage{bbm} %for blackboard 1, use \mathbbm{1}
\usepackage{xfrac} %for nice split-level fractions, use \sfrac{numerator}{denominator}
\usepackage{faktor} %for typesetting quotient structures e.g. \faktor{X}{A}

\usepackage{braket}
\usepackage[colorlinks,linkcolor={blue},citecolor={red},urlcolor={blue}]{hyperref} % hyperlinks references; colorlinks option: use text colour rather than box to indicate hyperlink
\usepackage[capitalise]{cleveref} % use \cref and \Cref commands for auto-formatted references
\crefname{equation}{}{}
\Crefname{equation}{Equation}{Equations}

\newcommand{\dd}{\ensuremath{\mathrm{d}}}
\theoremstyle{definition}
\newtheorem{definition}{Definition}[section]

\newtheorem{openquestion}[definition]{Open Question}

\theoremstyle{plain}
\newtheorem{proposition}[definition]{Proposition}
\newtheorem{theorem}[definition]{Theorem}
\newtheorem{lemma}[definition]{Lemma}
\newtheorem{corollary}[definition]{Corollary}

\theoremstyle{remark}
\newtheorem{example}[definition]{Example}
\newtheorem{remark}[definition]{Remark}

\usepackage{authblk}

\usepackage{tikz}
\usepackage{tikz-cd}

\usepackage{cancel}

\usepackage[
	disable,
	colorinlistoftodos
	]{todonotes}

\usepackage{fancyhdr}
\fancypagestyle{titlepage}{
	\fancyhf{}
	\fancyhead[R]{\texttt{BRX-TH-6699}}

}

\newcommand{\disjoint}{\bowtie}
\newcommand{\overlap}{\not \disjoint}

\DeclareMathOperator{\OverlapMonics}{\mathnormal{\overlap}-Monics}

\usepackage[style=alphabetic,
			alldates=iso,
			giveninits=true]{biblatex}
\renewbibmacro{in:}{%
  \ifentrytype{article}{}{\printtext{\bibstring{in}\intitlepunct}}}

\DeclareCiteCommand{\citejournal}
	{\usebibmacro{prenote}}
	{\clearname{author}%
	\clearfield{title}%
	\clearfield{url}%
	\clearfield{pages}%
	\clearfield{pagetotal}%
	\clearfield{edition}%
	\clearfield{labelyear}%
	\clearfield{doi}%
	\clearfield{eprintclass}%
	\usedriver
		{\DeclareNameAlias{sortname}{default}}
		{\thefield{entrytype}}}
	{\multicitedelim}
	{\usebibmacro{postnote}}

\newcommand{\citeinabstract}[1]{\citeauthor{#1} [\citejournal{#1}]}

\title{Spacetimes categories and disjointness for algebraic quantum field theory}
\author{Alastair Grant-Stuart}
\affil{\small Martin A. Fisher School of Physics, Brandeis University, Waltham MA, USA \\
\small Email: \texttt{agrantstuart@brandeis.edu}}

\date{\today}

\begin{document}

\maketitle

% display preprint number on arXiv submission:
% \thispagestyle{titlepage}

\begin{abstract}
An algebraic quantum field theory (AQFT) may be expressed as a functor from a category of spacetimes to a category of algebras of observables.
However, a generic category $\mathsf{C}$ whose objects admit interpretation as spacetimes is not necessarily viable as the domain of an AQFT functor;
often, additional constraints on the morphisms of $\mathsf{C}$ must be imposed.
We introduce \emph{disjointness relations}, a generalisation of the orthogonality relations of
\citeinabstract{BeniniSchenkelWoike2021}.
% Benini, Schenkel and Woike (2021).
In any category $\mathsf{C}$ equipped with a disjointness relation, we identify a subcategory $\mathsf{D}_\mathsf{C}$ which is suitable as the domain of an AQFT.
We verify that when $\mathsf{C}$ is the category of all globally hyperbolic spacetimes of dimension $d+1$ and all local isometries, equipped with the disjointness relation of spacelike separation, the specified subcategory $\mathsf{D}_\mathsf{C}$ is the commonly-used domain $\mathsf{Loc}_{d+1}$ of relativistic AQFTs.
By identifying appropriate chiral disjointness relations, we construct a category $\chi\mathsf{Loc}$ suitable as domain for chiral conformal field theories (CFTs) in two dimensions.
We compare this to an established AQFT formulation of chiral CFTs, and show that any chiral CFT expressed in the established formulation induces one defined on $\chi\mathsf{Loc}$.
\end{abstract}

% ------------------------------

% CONTENT

% ------------------------------

\tableofcontents

\listoftodos

\section{Introduction}
\label{sec:introduction}

In the modern locally covariant approach, an algebraic quantum field theory (AQFT) is realised as a functor whose domain is a category of spacetimes and whose codomain is a category of algebras of observables \cite{BrunettiFredenhagenVerch2003,FewsterVerch2015}.
Such a functor must satisfy physically motivated axioms, including causality and a time-slice axiom.

For relativistic theories in $d+1$ dimensions, the standard choice of domain category is $\mathsf{Loc}_{d+1}$.
The objects of $\mathsf{Loc}_{d+1}$ are globally hyperbolic spacetimes of dimension $d+1$.
The morphisms are maps simultaneously satisfying three separate properties: they preserve metric structure (i.e. are local isometries), they are injective, and their images are causally convex in their codomains.
With this, a relativistic AQFT is a functor $\mathcal{A} : \mathsf{Loc}_{d+1} \to \mathsf{Obs}$; the codomain $\mathsf{Obs}$ is commonly taken to be a category of associative and unital complex $*$-algebras whose morphisms are unit-preserving $*$-algebra homomorphisms.

It is noteworthy that the domain $\mathsf{Loc}_{d+1}$ is a strict subcategory of the category $\mathsf{GlobHypSpTm}_{d+1}$ of all $(d+1)$-dimensional globally hyperbolic spacetimes and \emph{all} structure-preserving maps (local isometries) between them.
Compared to $\mathsf{GlobHypSpTm}_{d+1}$, the morphisms of $\mathsf{Loc}_{d+1}$ are constrained to be injective with causally convex image.
This constraint is discovered by specific, \textit{ad hoc} physical reasoning particular to relativistic QFTs in \cite{BrunettiFredenhagenVerch2003}, with motivation from \cite{Kay1996}.

To formulate AQFTs beyond the standard relativistic setting, analogues of the domain category $\mathsf{Loc}_{d+1}$ must be found; this includes identifying analogous constraints on maps to define appropriate morphisms.
Several cases of physical interest are sufficiently similar to the standard relativistic case that the constraint of injectivity with causally convex image may be adapted directly from $\mathsf{Loc}_{d+1}$.
Such examples include (non-chiral) conformal field theories (CFTs) \cite{Pinamonti2009,BeniniGiorgettiSchenkel2021}, and theories defined on globally hyperbolic spacetimes equipped with extra structure (e.g. spin structures, background gauge fields) \cite{BeniniSchenkel2017}.
However, venturing further afield from $\mathsf{Loc}_{d+1}$ can lead to difficulties; see for instance \cite{BeniniPerinSchenkel2020} where a description of smooth AQFTs using stacks is limited to one-dimensional spacetimes (i.e. time intervals) only.

In developing an homotopical description of AQFT \cite{BeniniSchenkelWoike2019,BeniniSchenkel2019,BruinsmaSchenkel2019,Carmona2021,Yau2019}, \textcite{BeniniSchenkelWoike2021} have described a structure called an \emph{orthogonality relation} on a category.
To serve as the domain of an AQFT, a category must be equipped with an orthogonality relation.
On $\mathsf{Loc}_{d+1}$, the orthogonality relation formally encodes causal disjointness (spacelike separateness) in Lorentzian manifolds.
While the notion of causal disjointness also makes sense in $\mathsf{GlobHypSpTm}_{d+1}$, it fails to satisfy the definition of an orthogonality relation there; the extra constraints on morphisms of $\mathsf{Loc}_{d+1}$ influence whether causal disjointness constitutes an orthogonality relation.

\paragraph{Key contributions:}

Given a category $\mathsf{C}$ of spacetimes, we identify a particular subcategory $\mathsf{D}_\mathsf{C}$ which we propose as a suitable domain for AQFT functors.
To define $\mathsf{D}_\mathsf{C}$, we introduce a structure on $\mathsf{C}$ called a \emph{disjointness relation} -- a generalisation of the orthogonality relations of \cite{BeniniSchenkelWoike2021} -- which relates conterminous morphisms $f_1 : M_1 \rightarrow N \leftarrow M_2 : f_2$ in $\mathsf{C}$.
Denoting the disjointness relation on $\mathsf{C}$ by symbol $\disjoint_\mathsf{C}$, we say that $f_1$ and $f_2$ are disjoint if they are related under the disjointness relation ($f_1 \disjoint_\mathsf{C} f_2$) and that $f_1$ and $f_2$ overlap if they are not related ($f_1 \overlap_\mathsf{C} f_2$).

The physical role of the disjointness relation is to specify limitations on signal propagation in spacetimes.
For instance on $\mathsf{C} = \mathsf{GlobHypSpTm}_{d+1}$, the relevant disjointness relation records that $f_1 \disjoint_\mathsf{C} f_2$ if their images $f_i(M_i)$ are causally disjoint (spacelike separated) in their shared codomain $N$.
Via the causality axiom on AQFTs, the disjointness relation specifies when the degrees of freedom supported on some spacetime regions are necessarily independent.

With a disjointness relation $\disjoint_\mathsf{C}$ on $\mathsf{C}$, the identified subcategory $\mathsf{D}_\mathsf{C}$ is the wide subcategory of \emph{overlap-monic} morphisms.
Roughly, a morphism is overlap-monic with respect to a disjointness relation if it `respects disjointness'.
Notably, the disjointness relation on $\mathsf{C}$ reduces to an orthogonality relation on the subcategory $\mathsf{D}_\mathsf{C}$ upon restriction.

In the case of $\mathsf{C} = \mathsf{GlobHypSpTm}_{d+1}$ with the described causal disjointness relation, we show that a morphism is overlap-monic if and only if it is injective with causally convex image.
We begin by characterising overlap-monics without assuming global hyperbolicity of spacetimes, i.e. by considering the causal disjointness relation on a category $\mathsf{SpTm}_{d+1}$ of all spacetimes of dimension $d+1$ and all local isometries.
We can then specialise the characterisation of overlap-monics to subcategories of spacetimes with particular causal properties such as global hyperbolicity.
A key simplification occurs when all spacetimes are at least causally simple;
causal simplicity is a slightly weaker causal property than global hyperbolicity.
Specialising further to globally hyperbolic spacetimes delivers the advertised result and thereby verifies that our proposed subcategory of overlap-monics in $\mathsf{GlobHypSpTm}_{d+1}$ coincides with the often-used domain $\mathsf{Loc}_{d+1}$ of relativistic AQFTs.

The benefits of this general characterisation of AQFT-domain subcategory $\mathsf{D}_\mathsf{C}$ in spacetimes category $\mathsf{C}$ are twofold: first, $\mathsf{C}$ often has more appealing properties as a category than $\mathsf{D}_\mathsf{C}$.
For instance, in the case $\mathsf{C} = \mathsf{GlobHypSpTm}_{d+1}$, disjoint unions of spacetimes are coproducts in $\mathsf{GlobHypSpTm}_{d+1}$ but fail to satisfy the universal property of coproducts in $\mathsf{D}_\mathsf{C} = \mathsf{Loc}_{d+1}$.
Our proposal allows us access to any such good categorical properties of $\mathsf{C}$, while still giving a clear definition of AQFTs as functors $\mathsf{D}_\mathsf{C} \to \mathsf{Obs}$.

Second, our proposal gives a well-defined mathematical procedure to replace the \textit{ad hoc} discovery of morphism-constraints analogous to those of $\mathsf{Loc}_{d+1}$ (i.e. that morphisms must be injective maps with causally convex image).
This procedure can be used to construct appropriate domains for AQFTs in contexts where the notions of `spacetime' or disjointness differ from the standard relativistic setting.
The inputs needed for this procedure -- a category $\mathsf{C}$ of `spacetimes' and a disjointness relation on it -- can be arrived at by a series of physical modelling decisions:
\begin{itemize}
	\item
	The objects of $\mathsf{C}$ are decided by the relevant notion of `spacetime' in the context of interest.
	For relativistic theories of either the generic or conformal type, spacetimes are time-oriented Lorentzian manifolds of a particular dimension, possibly with additional `niceness' properties (usually global hyperbolicity).
	Euclidean theories may use Riemannian manifolds as `spacetimes' instead.
	
% 	\item
% 	Assuming $\mathsf{C}$ is to be concrete, the morphisms of $\mathsf{C}$ are structure-preserving maps for some appropriate choice of structure existing on the objects.
% 	So, if the QFTs under consideration are sensitive to some particular structure on the `spacetimes', then this structure decides the morphisms.
% 	For instance, the choice $\mathsf{C} = \mathsf{GlobHypSpTm}_{d+1}$ describes QFTs sensitive to the full Lorentzian metric; maps which preserve metric structure are local isometries.
% 	On the other hand, CFTs are sensitive only to conformal structure; in this case the appropriate maps are merely conformal rather than locally isometric.
% 	$\mathsf{C}$ may also be chosen not to be concrete; for instance, by taking $\mathsf{C}$ to be a poset (a thin category) one can recover traditional net-based formulations of AQFT.
	
	\item
	Assuming $\mathsf{C}$ is to be concrete, the morphisms of $\mathsf{C}$ are structure-preserving maps for some appropriate choice of structure existing on the objects -- this is the structure to which QFTs under consideration are sensitive.
	For instance, the choice $\mathsf{C} = \mathsf{GlobHypSpTm}_{d+1}$ describes QFTs sensitive to the full Lorentzian metric, since morphisms of $\mathsf{GlobHypSpTm}_{d+1}$ are local isometries.
	On the other hand, CFTs are sensitive only to conformal structure; in this case the appropriate morphisms are merely conformal maps rather than local isometries.
	One may also choose a non-concrete category $\mathsf{C}$;
% 	$\mathsf{C}$ may also be chosen not to be concrete;
	for instance, by taking $\mathsf{C}$ to be a poset (a thin category) one can recover traditional net-based formulations of AQFT.\todo{refer forwards if I include this example}
	
	\item
	The disjointness relation on $\mathsf{C}$ is decided by limitations on signal propagation in the theories under consideration.
	For instance, signals in standard relativistic theories may only propagate along causal curves; therefore regions of spacetime which are spacelike separated should be recorded as disjoint by the disjointness relation.
	In chiral CFTs, signals in the right- or left-moving half of the theory propagate only along right- or left-moving null curves respectively; each of these defines a different disjointness relation on the appropriate category.
\end{itemize}

\Textcite{BeniniSchenkelWoike2021} have formalised the structures on a category $\mathsf{D}$ needed to define an AQFT as a functor $\mathsf{D} \to \mathsf{Obs}$.
The required structures enable statement of the causality and time-slice axioms.

To state the causality axiom, $\mathsf{D}$ must be equipped with an orthogonality relation.
A functor $\mathcal{A} : \mathsf{D} \to \mathsf{Obs}$ satisfies the causality axiom if for any conterminous pair $f_1 : M_1 \rightarrow N \leftarrow M_2 : f_2$ in $\mathsf{D}$ related under the orthogonality relation, $\mathcal{A}f_1 (a_1)$ and $\mathcal{A}f_2 (a_2)$ commute in algebra $\mathcal{A}N$ for any observables $a_1 \in \mathcal{A}M_1$ and $a_2 \in \mathcal{A}M_2$:
\begin{equation*}
	[\mathcal{A}f_1 (a_1), \mathcal{A}f_2 (a_2)]_{\mathcal{A}N} = 0.
\end{equation*}
In the terminology of \cite{BeniniSchenkelWoike2021}, such $\mathcal{A}$ is called $\perp$-commutative.
In the case of $\mathsf{D} = \mathsf{Loc}_{d+1}$, the relevant orthogonality relation is that of causal disjointness, so $f_1$ and $f_2$ are related if the images $f_i(M_i)$ are spacelike separated in the shared target spacetime $N$.
Here, the causality axiom expresses that observables supported on spacelike-separated regions of spacetime $N$ must commute.

To state the time-slice axiom, there must be a specified class $W$ of morphisms in $\mathsf{D}$.
A functor $\mathcal{A} : \mathsf{D} \to \mathsf{Obs}$ satisfies the time-slice axiom with respect to $W$ if $\mathcal{A}f$ is an isomorphism in $\mathsf{Obs}$ for every $f \in W$.
In the terminology of \cite{BeniniSchenkelWoike2021}, such $\mathcal{A}$ is called $W$-constant.
In the case of $\mathsf{D} = \mathsf{Loc}_{d+1}$, the specified class $W$ consists of all morphisms $f : M \to N$ such that the image $f(M)$ contains a Cauchy surface of $N$; such maps are called Cauchy maps.

With these statements of the causality and time-slice axioms in hand, our proposed procedure to construct an AQFT domain may be summarised as follows: take as input data $(\mathsf{C}, \disjoint_\mathsf{C}, W_\mathsf{C})$ consisting of a category $\mathsf{C}$, a disjointness relation $\disjoint_\mathsf{C}$ on $\mathsf{C}$, and a class of morphisms $W_\mathsf{C}$ in $\mathsf{C}$.
From this, produce data $(\mathsf{D}_\mathsf{C}, \disjoint_{\mathsf{D}_\mathsf{C}}, W_{\mathsf{D}_\mathsf{C}})$ consisting of the wide subcategory $\mathsf{D}_\mathsf{C}$ of overlap-monics in $\mathsf{C}$ with respect to $\disjoint_\mathsf{C}$, along with the restrictions $\disjoint_{\mathsf{D}_\mathsf{C}}$ and $W_{\mathsf{D}_\mathsf{C}}$ to this subcategory of $\disjoint_\mathsf{C}$ and $W_\mathsf{C}$ respectively.
The restricted disjointness relation $\disjoint_{\mathsf{D}_\mathsf{C}}$ is an orthogonality relation on $\mathsf{D}_\mathsf{C}$.
This data suffices to define an AQFT on $\mathsf{D}_\mathsf{C}$ as a functor $\mathcal{A}: \mathsf{D}_\mathsf{C} \to \mathsf{Obs}$ satisfying the causality axiom with respect to $\disjoint_{\mathsf{D}_\mathsf{C}}$, and the time-slice axiom with respect to $W_{\mathsf{D}_\mathsf{C}}$.

\paragraph{Applications:}

To illustrate our proposed procedure, we apply it to chiral CFTs.
We begin by defining appropriate right- and left-chiral disjointness relations on a category $\mathsf{CSpTm}_{1+1}^\mathrm{o,to}$ of two-dimensional (2D) oriented spacetimes and conformal maps which preserve orientation and time-orientation.
These disjointness relations reflect that in the right- and left-moving halves of a chiral CFT, signals can only propagate along right- and left-moving null curves respectively.
We then characterise overlap-monics with respect to these chiral disjointness relations; this mirrors the standard relativistic case.
In particular, to simplify the characterisations we restrict to spacetimes with good chiral properties, analogous to the causal properties of causal simplicity and global hyperbolicity used in the standard relativistic case.
Formulating the needed chiral properties requires that we first develop some technical tools, namely chiral frames and chiral flows on 2D oriented spacetimes.

The characterisation of overlap-monics with respect to the right-chiral disjointness relation simplifies when we restrict to spacetimes satisfying a right-chiral analogue of global hyperbolicity.
Specifically, in the full subcategory of $\mathsf{CSpTm}_{1+1}^\mathrm{o,to}$ on spacetimes containing an appropriate right-chiral analogue of a Cauchy surface, a morphism is overlap-monic with respect to the right-chiral disjointness relation if and only if it is injective and its image satisfies an appropriate right-chiral analogue of causal convexity in its codomain.
We denote the resulting subcategory of overlap-monics between such spacetimes as $\chi\mathsf{Loc}$, in analogy with $\mathsf{Loc}_{d+1}$.
Similar observations apply for the left-chiral disjointness relation.

Our proposal then suggests that the right-moving half of a chiral CFT is a functor $\chi\mathsf{Loc}\to \mathsf{Obs}$ obeying causality and time-slice axioms.
This differs from established AQFT formulations of chiral CFTs.
In traditional, net-based AQFT, a chiral CFT is defined on a net of open subsets of the circle $\mathbb{S}^1$.
A locally covariant version of this defines a chiral CFT as a functor $\mathsf{Emb}_1 \to \mathsf{Obs}$ where $\mathsf{Emb}_1$ is a category of one-dimensional manifolds and embeddings; see \cite{BeniniGiorgettiSchenkel2021}\footnote{In \cite{BeniniGiorgettiSchenkel2021}, the category we refer to here as $\mathsf{Emb}_1$ is denoted as $\mathsf{Man}_1$.}.
The key difference from our proposal is that the established formulations of chiral CFT build in the time-slice axiom preemptively.

To compare, we explicitly express the time-slice axiom on $\chi\mathsf{Loc}$ by describing a class $W$ of morphisms serving as an appropriate chiral analogue of Cauchy maps.
We exhibit a functor $Q : \chi\mathsf{Loc} \to \mathsf{Emb}_1$ which sends all morphisms in $W$ to isomorphisms.
Consequently any functor $\mathcal{A}:\mathsf{Emb}_1 \to \mathsf{Obs}$ representing a chiral CFT in the established formulation induces by pre-composition a functor $\mathcal{A}\circ Q: \chi\mathsf{Loc} \to \mathsf{Obs}$ representing a chiral CFT satisfying the time-slice axiom with respect to $W$ in our formulation.

We leave as an open question whether the converse is also true: does any functor $\chi\mathsf{Loc} \to \mathsf{Obs}$ representing a chiral CFT satisfying the time-slice axiom also induce a functor $\mathsf{Emb}_1 \to \mathsf{Obs}$ representing a chiral CFT in the established formulation?

Future applications of our proposal to determine domain $\mathsf{D}_\mathsf{C}$ from spacetimes category $(\mathsf{C}, \disjoint_\mathsf{C})$ may proceed beyond chiral CFTs.
A minor generalisation of the standard relativistic type of AQFT loosens the assumption of global hyperbolicity on spacetimes; in fact, this generalisation is treated already in our verification that overlap-monics in $\mathsf{GlobHypSpTm}_{d+1}$ coincide with $\mathsf{Loc}_{d+1}$.
For theories on spacetimes with extra structure (e.g. spin structure, background gauge fields) one may choose $\mathsf{C}$ to be an appropriately fibred category, similar to \cite{BeniniSchenkel2017}.
Further removed from the relativistic context, for Euclidean field theories $\mathsf{C}$ is expected to be a category of Riemannian manifolds.
More speculative possibilities may include lattice field theories (if an appropriate category $\mathsf{C}$ of lattice representations of spacetimes is identified), or smooth AQFTs as in \cite{BeniniPerinSchenkel2020}
(if disjointness relations admit an adequate generalisation to stacks of categories).

\paragraph{Outline of the paper:}

In \cref{sec:disjointness}, we present a categorical description of disjointness, including the definitions of disjointness relations (\cref{sec:disjointness:relations}) and overlap-monics (\cref{sec:disjointness:overlap-monics}).
\Cref{sec:spacetimes-cats-relativistic} verifies that $\mathsf{Loc}_{d+1}$ is the subcategory in $\mathsf{GlobHypSpTm}_{d+1}$ of overlap-monics with respect to the causal disjointness relation.
This is done by first studying overlap-monics in the category $\mathsf{SpTm}_{d+1}$ of spacetimes without assuming any causal properties, and then specialising to spacetimes with causal properties of increasing strength until global hyperbolicity is reached.
In \cref{sec:spacetimes-cats-chiral}, we apply our proposed construction to chiral CFTs by defining left- and right-chiral disjointness relations a category of 2D oriented spacetimes and conformal maps.
The result is a category $\chi\mathsf{Loc}$ which we propose as domain of functors $\chi\mathsf{Loc} \to \mathsf{Obs}$ representing a chiral CFTs.
In \cref{sec:spacetimes-cats-chiral:comparison}, we compare this proposed $\chi\mathsf{Loc}$ to established AQFT formulations of chiral CFTs by implementing the time-slice axiom.

\section{A categorical description of disjointness}
\label{sec:disjointness}
% BEGIN

In this section, we introduce disjointness relations on categories; these are structures capable of describing various intuitive notions of disjointness in a formal categorical setting.
They are a direct generalisation of the orthogonality relations introduced by \textcite{BeniniSchenkelWoike2021} to formalise causal disjointness (spacelike separateness) in spacetimes, as necessary to describe the causality axiom in algebraic QFT.

Our disjointness relations facilitate a formal description of those morphisms in a category which `respect disjointness'; this description is presented in \cref{sec:disjointness:overlap-monics}.
In particular, we may recover a category with orthogonality relation -- as needed for AQFT -- by considering the `disjointness-respecting' morphisms in any category equipped with a disjointness relation.

\subsection{Disjointness relations on categories}
\label{sec:disjointness:relations}

\begin{definition} \label{def:disjointness_rel}
	A \emph{disjointness relation} (or \emph{$\disjoint$-relation}) $\disjoint_\mathsf{C}$ on a category $\mathsf{C}$ is a binary relation on conterminous morphisms of $\mathsf{C}$, denoted as $f_1 \disjoint_\mathsf{C} f_2$ or
    \begin{equation*}
		\begin{tikzcd}[cramped, column sep=small]
								& c \\
			b_1 \arrow[ur, "f_1"{name=1, near start}]	& & b_2 \arrow[ul, "f_2"{near start, name=2,swap}]
			\arrow[phantom, from=1, to=2, "\disjoint_\mathsf{C}"{below, pos=0.58}]
		\end{tikzcd}
	\end{equation*}
	when $f_1$ and $f_2$ are related under $\disjoint_\mathsf{C}$, which satisfies the following properties: for any conterminous pair $f_1 : b_1 \rightarrow c \leftarrow b_2 : f_2$ in $\mathsf{C}$,
	\begin{enumerate}
		\item
		symmetry:
		$f_1 \disjoint_\mathsf{C} f_2$ implies $f_2 \disjoint_\mathsf{C} f_1$,

		\item \label{def:disjointness_rel:pre-comp}
		stability under pre-composition:
		$f_1 \disjoint_\mathsf{C} f_2$ implies $f_1 \circ g_1 \disjoint_\mathsf{C} f_2 \circ g_2$ for any composable morphisms $g_1$ and $g_2$. Diagramatically:
		\begin{equation*}
		\begin{tikzcd}[cramped, column sep=small]
								& c \\
			b_1 \arrow[ur, "f_1"{name=1, near start}]	& & b_2 \arrow[ul, "f_2"{near start, name=2,swap}]
			\arrow[phantom, from=1, to=2, "\disjoint_\mathsf{C}"{below, pos=0.58}]
		\end{tikzcd}
		\qquad \textrm{implies} \qquad
		\begin{tikzcd}[cramped, column sep=small]
								& c \\
			b_1 \arrow[ur, "f_1"{name=1, near start}] \arrow[rr, phantom, "\disjoint_\mathsf{C}"{pos=0.58}]	& & b_2 \arrow[ul, "f_2"{near start, name=2,swap}] \\
			a_1 \arrow[u, "g_1"]	& & a_2 \arrow[u, "g_2"{swap}]
		\end{tikzcd},
		\end{equation*}

		\item \label{def:disjointness_rel:post-comp}
		stability under post-composition by isomorphisms:
		$f_1 \disjoint_\mathsf{C} f_2$ implies that $h \circ f_1 \disjoint_\mathsf{C} h \circ f_2$ for any composable isomorphism $h$. Diagramatically:
		\begin{equation*}
		\begin{tikzcd}[cramped, column sep=small]
								& c \\
			b_1 \arrow[ur, "f_1"{name=1, near start}]	& & b_2 \arrow[ul, "f_2"{near start, name=2,swap}]
			\arrow[phantom, from=1, to=2, "\disjoint_\mathsf{C}"{below, pos=0.58}]
		\end{tikzcd}
		\qquad \textrm{implies} \qquad
		\begin{tikzcd}[cramped, column sep=small]
								& d \\
								& c \arrow[u,"h"] \arrow[u, draw=none, shift right=1.2, "\sim"{marking}] \\
			b_1 \arrow[ur, "f_1"{name=1, near start}]	& & b_2 \arrow[ul, "f_2"{near start, name=2,swap}]
			\arrow[phantom, from=1, to=2, "\disjoint_\mathsf{C}"{below, pos=0.58}]
		\end{tikzcd}.
		\end{equation*}
	\end{enumerate}
	The pair $(\mathsf{C}, \disjoint_\mathsf{C})$ consisting of a category $\mathsf{C}$ and  a disjointness relation $\disjoint_\mathsf{C}$ on $\mathsf{C}$ is called a \emph{$\disjoint$-category}.
\end{definition}

\begin{remark} \label{rm:disjointness_rel:iso}
	The property of stability under post-composition by isomorphisms is the minimal property needed to ensure that $\disjoint$-relations respect isomorphisms of the category.
	Property \labelcref{def:disjointness_rel:post-comp} of \cref{def:disjointness_rel} is equivalent to:
	\begin{itemize}
		\item[$3'.$]
		For isomorphism $h : c \to d$ and conterminous pair $f_1 : b_1 \rightarrow c \leftarrow b_2 : f_2$ in $\mathsf{C}$, $f_1 \disjoint_\mathsf{C} f_2$ if and only if $h \circ f_1 \disjoint_\mathsf{C} h \circ f_2$.
	\end{itemize}
	This is because stability under post-composition by isomorphisms gives that for isomorphism $h$, $h \circ f_1 \disjoint_\mathsf{C} h \circ f_2$ implies $h^{-1} \circ h \circ f_1 \disjoint_\mathsf{C} h^{-1} \circ h \circ f_2$.
\end{remark}

\begin{example} \label{ex:setwise_disjointness:sets}
	Consider the category $\mathsf{Set}$ of sets and functions. Define a $\disjoint$-relation $\disjoint_\mathrm{set}$ of setwise-disjointness on $\mathsf{Set}$ by:
	\begin{equation*}
		\begin{tikzcd}[cramped, column sep=small]
								& Y \\
			X_1 \arrow[ur, "f_1"{name=1, near start}]	& & X_2 \arrow[ul, "f_2"{near start, name=2,swap}]
			\arrow[phantom, from=1, to=2, "\disjoint_\mathrm{set}"{below, pos=0.63}]
		\end{tikzcd}
		\qquad \text{if} \qquad
		% https://q.uiver.app/?q=WzAsNCxbMCwwLCJcXHZhcm5vdGhpbmciXSxbMSwwLCJYXzIiXSxbMCwxLCJYXzEiXSxbMSwxLCJZIl0sWzAsMV0sWzAsMl0sWzIsMywiZl8xIiwyXSxbMSwzLCJmXzIiXSxbMCwzLCIiLDEseyJzdHlsZSI6eyJuYW1lIjoiY29ybmVyIn19XV0=
		\begin{tikzcd}
			\varnothing & {X_2} \\
			{X_1} & Y
			\arrow[from=1-1, to=1-2]
			\arrow[from=1-1, to=2-1]
			\arrow["{f_1}"', from=2-1, to=2-2]
			\arrow["{f_2}", from=1-2, to=2-2]
			\arrow["\lrcorner"{anchor=center, pos=0.125}, draw=none, from=1-1, to=2-2]
		\end{tikzcd}
		\qquad \text{ is a pullback.}
	\end{equation*}
	Equivalently, $f_1 \disjoint_\mathrm{set} f_2$ if the intersection $f_1(X_1) \cap f_2(X_2)$ of their images in $Y$ is empty.
	It is straightforward to check that this satisfies symmetry and stability under post-composition by isomorphisms as per \cref{def:disjointness_rel}.
	To show stability under pre-composition, note first that the initial object $\varnothing$ in $\mathsf{Set}$ is stable under pullback, i.e. the (apex of the) pullback of $g: W \rightarrow X_1 \leftarrow \varnothing$ for any $g$ is again $\varnothing$.
	It follows by the pasting law for pullbacks that if the rightmost square of the following diagram is a pullback, then so too is the outermost square:
	\begin{equation*}
	% https://q.uiver.app/?q=WzAsNixbMCwwLCJcXHZhcm5vdGhpbmciXSxbMSwwLCJcXHZhcm5vdGhpbmciXSxbMiwwLCJYXzIiXSxbMiwxLCJZIl0sWzEsMSwiWF8xIl0sWzAsMSwiVyJdLFswLDFdLFsxLDJdLFsyLDMsImZfMiJdLFsxLDRdLFswLDVdLFs1LDQsImciLDJdLFs0LDMsImZfMSIsMl0sWzAsNCwiIiwyLHsic3R5bGUiOnsibmFtZSI6ImNvcm5lciJ9fV1d
	\begin{tikzcd}
		\varnothing & \varnothing & {X_2} \\
		W & {X_1} & Y
		\arrow[from=1-1, to=1-2]
		\arrow[from=1-2, to=1-3]
		\arrow["{f_2}", from=1-3, to=2-3]
		\arrow[from=1-2, to=2-2]
		\arrow[from=1-1, to=2-1]
		\arrow["g"', from=2-1, to=2-2]
		\arrow["{f_1}"', from=2-2, to=2-3]
		\arrow["\lrcorner"{anchor=center, pos=0.125}, draw=none, from=1-1, to=2-2]
	\end{tikzcd}.
	\end{equation*}
	Hence if $f_1 \disjoint_\mathrm{set} f_2$ then $f_1 \circ g \disjoint_\mathrm{set} f_2$.
	
	We note that it is not true that $f_1 \disjoint_\mathrm{set} f_2$ implies $h \circ f_1 \disjoint_\mathrm{set} h \circ f_2$ for arbitrary $h$: any non-injective function $h$ sends some disjoint subsets of its domain to intersecting subsets of its codomain.

	This example generalises to any category $\mathsf{C}$ with an initial object that is stable under pullback; for instance, this includes $\mathsf{C} = \mathsf{Cat}$ the category of small categories, with initial object the empty category.
% 	
% 	
% 	
% 	Equivalently, $f_1 \disjoint_\mathrm{set} f_2$ if the intersection $f_1(X_1) \cap f_2(X_2)$ of their images in $Y$ is empty.
% 	It is straightforward to check that this satisfies \cref{def:disjointness_rel}; stability under pre-composition uses that the initial object $\varnothing$ in $\mathsf{Set}$ is stable under pullback, i.e. the (apex of the) pullback of $\varnothing \rightarrow Y \leftarrow X$ is again $\varnothing$. We note that it is not true that $f_1 \disjoint_\mathrm{set} f_2$ implies $h \circ f_1 \disjoint_\mathrm{set} h \circ f_2$ for arbitrary $h$: any non-injective function $h$ sends some disjoint subsets of its domain to intersecting subsets of its codomain.
% 
% 	
% 	This example generalises to any category $\mathsf{C}$ with pullbacks and an initial object that is stable under pullback; for instance, this includes $\mathsf{C} = \mathsf{Cat}$ the category of small categories, with initial object the empty category.
\end{example}

Disjointness relations of \cref{def:disjointness_rel} generalise the \emph{orthogonality relations} or \emph{$\perp$-relations} introduced in \cite{BeniniSchenkelWoike2021}.
Specifically, an orthogonality relation $\perp_\mathsf{C}$ on $\mathsf{C}$ satisfies \cref{def:disjointness_rel} (everywhere replacing $\disjoint_\mathsf{C}$ by $\perp_\mathsf{C}$), except with property \labelcref{def:disjointness_rel:post-comp} substituted for:
\begin{itemize}
	\item[$3_\perp.$] stability under post-composition: $f_1 \perp_\mathsf{C} f_2$ implies $h \circ f_1 \perp_\mathsf{C} h \circ f_2$ for any composable morphism $h$ (isomorphism or otherwise).
\end{itemize}
The generalisation from $\perp$-relation to $\disjoint$-relation allows us to describe notions of disjointness that would not meet the definition of an orthogonality relation, as illustrated by \cref{ex:setwise_disjointness:sets} above.

When we wish to denote that conterminous morphisms $f_1$ and $f_2$ of a $\disjoint$-category $(\mathsf{C}, \disjoint_\mathsf{C})$ are not related under $\disjoint_\mathsf{C}$, we write $f_1 \overlap_\mathsf{C} f_2$; the complement $\overlap_\mathsf{C}$ of $\disjoint_\mathsf{C}$ is also a binary relation on conterminous morphisms of $\mathsf{C}$.
$f_1 \overlap_\mathsf{C} f_2$ may be read as `$f_1$ and $f_2$ are not disjoint' or `$f_1$ and $f_2$ overlap'.

\begin{proposition} \label{prop:complement_disjointness_relation}
	A binary relation $\disjoint_\mathsf{C}$ on conterminous morphisms of category $\mathsf{C}$ is a valid $\disjoint$-relation if and only if its complement $\overlap_\mathsf{C}$ satisfies the following properties
	for any conterminous pair $f_1 : b_1 \rightarrow c \leftarrow b_2 : f_2$ in $\mathsf{C}$:
	\begin{enumerate}
		\item
		symmetry:
		$f_1 \overlap_\mathsf{C} f_2$ implies $f_2 \overlap_\mathsf{C} f_1$

		\item
		stability under pre-cancellation:
		for any morphisms $g_i : a_i \to b_i$,
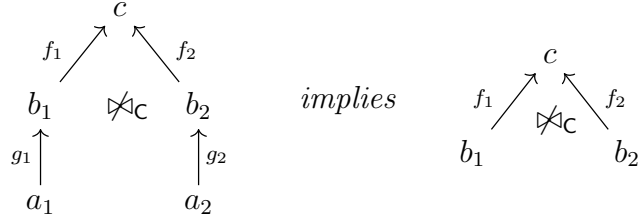
		\begin{equation*}
		\begin{tikzcd}[cramped, column sep=small]
								& c \\
			b_1 \arrow[ur, "f_1"{name=1, near start}] \arrow[rr, phantom, "\overlap_\mathsf{C}"{pos=0.58}]	& & b_2 \arrow[ul, "f_2"{near start, name=2,swap}] \\
			a_1 \arrow[u, "g_1"]	& & a_2 \arrow[u, "g_2"{swap}]
		\end{tikzcd}
		\qquad \textrm{implies} \qquad
		\begin{tikzcd}[cramped, column sep=small]
								& c \\
			b_1 \arrow[ur, "f_1"{name=1, near start}]	& & b_2 \arrow[ul, "f_2"{near start, name=2,swap}]
			\arrow[phantom, from=1, to=2, "\overlap_\mathsf{C}"{below,pos=0.58}]
		\end{tikzcd}
		\end{equation*}

		\item
		stability under post-composition by isomorphisms:
		for isomorphism $h : c \xrightarrow{\sim} d$,
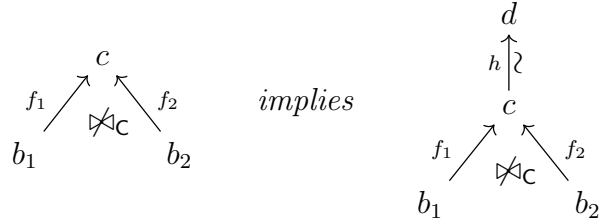
		\begin{equation*}
		\begin{tikzcd}[cramped, column sep=small]
								& c \\
			b_1 \arrow[ur, "f_1"{name=1, near start}]	& & b_2 \arrow[ul, "f_2"{near start, name=2,swap}]
			\arrow[phantom, from=1, to=2, "\overlap_\mathsf{C}"{below,pos=0.58}]
		\end{tikzcd}
		\qquad \textrm{implies} \qquad
		\begin{tikzcd}[cramped, column sep=small]
								& d \\
								& c \arrow[u,"h"] \arrow[u, draw=none, shift right=1.2, "\sim"{marking}] \\
			b_1 \arrow[ur, "f_1"{name=1, near start}]	& & b_2 \arrow[ul, "f_2"{near start, name=2,swap}]
			\arrow[phantom, from=1, to=2, "\overlap_\mathsf{C}"{below,pos=0.58}]
		\end{tikzcd}
		\end{equation*}
	\end{enumerate}
\end{proposition}
\begin{proof}
	As in \cref{rm:disjointness_rel:iso}, stability of $\overlap_\mathsf{C}$ under post-composition by isomorphisms is equivalent to the condition that $f_1 \overlap_\mathsf{C} f_2$ if and only if $h \circ f_1 \overlap_\mathsf{C}  h \circ f_2$ for any isomorphism $h$.
	By contraposition, this gives that $f_1 \disjoint_\mathsf{C} f_2$ if and only if $h \circ f_1 \disjoint_\mathsf{C} h \circ f_2$ for isomorphism $h$.
	Symmetry and stability under pre-cancellation of $\overlap_\mathsf{C}$ are respective contrapositives of symmetry and stability under pre-composition of $\disjoint_\mathsf{C}$.
\end{proof}

\begin{example} \label{ex:binrel_disjointness}
	Consider the category $\mathsf{sBin}$ of sets equipped with symmetric binary relations defined as follows.
	Objects $(X, R_X)$ of $\mathsf{sBin}$ consist of a set $X$ and a symmetric homogeneous binary relation $R_X \subseteq X \times X$; recall that $R_X$ is symmetric if $(a,b) \in R_X$ implies $(b,a) \in R_X$ for any $a,b \in X$.
	Morphisms $f:(X, R_X) \to (Y,R_Y)$ of $\mathsf{sBin}$ are relation-preserving maps, i.e. maps $f : X \to Y$ such that $(a,b) \in R_X$ implies $(f(a), f(b)) \in R_Y$ for any $a,b \in X$.
	
	Define $\disjoint$-relation $\disjoint_\mathrm{bin}$ on $\mathsf{sBin}$ by:
	
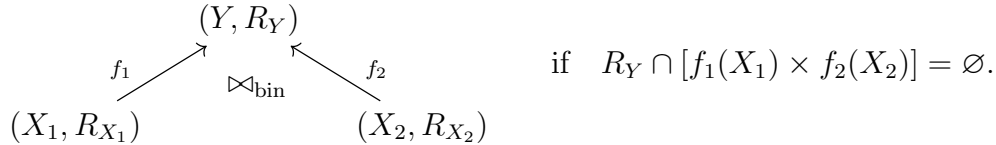
\begin{equation*}
		\begin{tikzcd}[cramped, column sep=small]
								& (Y,R_Y) \\
			(X_1,R_{X_1}) \arrow[ur, "f_1"{name=1, near start}]	& & (X_2,R_{X_2}) \arrow[ul, "f_2"{near start, name=2,swap}]
			\arrow[phantom, from=1, to=2, "\disjoint_\mathrm{bin}"{below, pos=0.54}]
		\end{tikzcd}
		\qquad \text{if} \quad
		R_Y \cap \left[f_1(X_1) \times f_2(X_2)\right] = \varnothing.
	\end{equation*}
	Equivalently, $f_1 \disjoint_\mathrm{bin} f_2$ if there does not exist $a_1 \in X_1$ and $a_2 \in X_2$ such that $(f_1(a_1) , f_2(a_2)) \in R_Y$.
	
	That this $\disjoint_\mathrm{bin}$ is a valid $\disjoint$-relation can be shown using \cref{prop:complement_disjointness_relation}.
	Symmetry of $\overlap_\mathrm{bin}$ follows from symmetry of the binary relations $R_X$ in objects $(X, R_X)$ of $\mathsf{sBin}$.
	Stability of $\overlap_\mathrm{bin}$ under pre-cancellation and post-composition are shown straightforwardly; in particular, $\overlap_\mathrm{bin}$ is stable under post-composition by arbitrary morphisms rather than only isomorphisms.
\end{example}

Functors may respect $\disjoint$-relations on categories in different ways:

\begin{definition}
	Let $(\mathsf{C}, \disjoint_\mathsf{C})$ and $(\mathsf{D}, \disjoint_\mathsf{D})$ be $\disjoint$-categories, and let $F :\mathsf{C} \to \mathsf{D}$ be a functor. We say:
	\begin{itemize}
		\item
		$F$ \emph{preserves $\disjoint$-relations}, or $F : (\mathsf{C}, \disjoint_\mathsf{C}) \to (\mathsf{D}, \disjoint_\mathsf{D})$ is \emph{$\disjoint$-preserving},  if
		\begin{equation*}
		\begin{tikzcd}[cramped, column sep=small]
								& c \\
			b_1 \arrow[ur, "f_1"{name=1, near start}]	& & b_2 \arrow[ul, "f_2"{near start, name=2,swap}]
			\arrow[phantom, from=1, to=2, "\disjoint_\mathsf{C}"{below, pos=0.6}]
		\end{tikzcd}
		\qquad \textrm{implies} \qquad
		\begin{tikzcd}[cramped, column sep=small]
								& Fc \\
			Fb_1 \arrow[ur, "Ff_1"{name=1, near start}]	& & Fb_2 \arrow[ul, "Ff_2"{near start, name=2,swap}]
			\arrow[phantom, from=1, to=2, "\disjoint_\mathsf{D}"{below, pos=0.55}]
		\end{tikzcd},
		\end{equation*}
		for any conterminous pair $b_1 \xrightarrow{f_1} c \xleftarrow{f_2} b_2$ in $\mathsf{C}$.

		\item
		$F$ \emph{reflects $\disjoint$-relations}, or $F : (\mathsf{C}, \disjoint_\mathsf{C}) \to (\mathsf{D}, \disjoint_\mathsf{D})$ is \emph{$\disjoint$-reflecting}, if
		\begin{equation*}
		\begin{tikzcd}[cramped, column sep=small]
								& Fc \\
			Fb_1 \arrow[ur, "Ff_1"{name=1, near start}]	& & Fb_2 \arrow[ul, "Ff_2"{near start, name=2,swap}]
			\arrow[phantom, from=1, to=2, "\disjoint_\mathsf{D}"{below, pos=0.55}]
		\end{tikzcd}
		\qquad \textrm{implies} \qquad
		\begin{tikzcd}[cramped, column sep=small]
								& c \\
			b_1 \arrow[ur, "f_1"{name=1, near start}]	& & b_2 \arrow[ul, "f_2"{near start, name=2,swap}]
			\arrow[phantom, from=1, to=2, "\disjoint_\mathsf{C}"{below, pos=0.6}]
		\end{tikzcd},
		\end{equation*}
		for any conterminous pair $b_1 \xrightarrow{f_1} c \xleftarrow{f_2} b_2$ in $\mathsf{C}$.
	\end{itemize}
\end{definition}

In this work, many functors of interest will both preserve and reflect $\disjoint$-relations.
One source of such functors is as follows:
\begin{definition}
	Given a $\disjoint$-category $(\mathsf{D}, \disjoint_\mathsf{D})$ and a functor $F : \mathsf{C} \to \mathsf{D}$, the \emph{pullback of $\disjoint_\mathsf{D}$ along $F$} is the $\disjoint$-relation $\disjoint_F$ on $\mathsf{C}$ defined by $f_1 \disjoint_F f_2$ if $Ff_1 \disjoint_\mathsf{D} Ff_2$, for conterminous pair $f_1 : b_1 \rightarrow c \leftarrow b_2 : f_2$ in $\mathsf{C}$.
	
	With respect to $\disjoint_F$, the functor $F : (\mathsf{C}, \disjoint_F) \to (\mathsf{D}, \disjoint_\mathsf{D})$ both preserves and reflects $\disjoint$-relations.
\end{definition}
\todo[color=cyan]{better name than pullback? My guess: $\disjoint_F$ is the smallest $\disjoint$-relation such that $F$ is $\disjoint$-reflecting, and largest such that $F$ is $\disjoint$-preserving. Then maybe `final' or `cofinal' or `initial' $\disjoint$-relation?}

\begin{example} \label{ex:setwise_disjointness:concrete}
	For any concrete category $\mathsf{C}$ with forgetful functor $U : \mathsf{C} \to \mathsf{Set}$, we may equip $\mathsf{C}$ with the pullback $\disjoint_U$ of the setwise-disjointness relation $\disjoint_\mathrm{set}$ of \cref{ex:setwise_disjointness:sets}.
	
	For clarity where the forgetful functor is not explicitly named, we may also denote the pullback relation $\disjoint_U$ on $\mathsf{C}$ as $\disjoint_\mathrm{set}$.
\end{example}

\begin{example} \label{ex:setwise_disjointness:diagonal}
	Consider the functor $\Delta : \mathsf{Set} \to \mathsf{sBin}$ which sends sets $X$ to $(X, \Delta_X)$, for diagonal relation $\Delta_X := \left\{(x,x)\in X \times X\right\}$.
	By definition, any map $f : X \to Y$ preserves the diagonal relations.
	
	Then $\disjoint_\mathrm{set}$ on $\mathsf{Set}$ coincides with the pullback $\disjoint_\Delta$ of $\disjoint_\mathrm{bin}$ along $\Delta$: for any conterminous pair $f_1 : X_1 \rightarrow Y \leftarrow X_2 : f_2$, we have $f_1 \disjoint_\mathrm{set} f_2$ if and only if $f_1(X_1) \cap f_2(X_2) = \varnothing$ if and only if $\Delta_Y \cap \left[f_1(X_1) \times f_2(X_2)\right] = \varnothing$ if and only if $f_1 \disjoint_\Delta f_2$.
	See also \cref{ex:setwise_disjointness:sets,ex:binrel_disjointness}.
\end{example}

\begin{example} \label{ex:top_disjointness}
	Consider the functor $C : \mathsf{Top} \to \mathsf{sBin}$ which sends any topological space $X$ to its underlying set equipped with the equivalence relation $C_X$ defined by $(a, b) \in C_X$ if $a,b \in X$ lie in the same connected component of $X$.
	Continuous maps preserve the relations $C$ because the image of a connected set under a continuous map is connected.
	
	The pullback $\disjoint_C$ of $\disjoint_\mathrm{bin}$ along $C$ gives a $\disjoint$-relation on $\mathsf{Top}$.
	Explicitly, conterminous pair $f_1 : X_1 \rightarrow Y \leftarrow X_2 : f_2$ in $\mathsf{Top}$ has $f_1 \disjoint_C f_2$ if no points $a_1 \in X_1$, $a_2 \in X_2$ have $f_i(a_i)$ lying in the same connected component of $Y$, i.e. $f_1 \disjoint_C f_2$ if either of $X_i$ is empty or there exists a separation $\left\{U_1,U_2\right\}$ of $Y$ with $f_i(X_i) \subseteq U_i$.
	
	Similarly, consider the functor $P : \mathsf{Top} \to \mathsf{sBin}$ which sends any topological space $X$ to its underlying set equipped with the equivalence relation $P_X$ defined by $(a,b) \in P_X$ if $a,b \in X$ lie in the same path-component of $X$.
	
	The pullback $\disjoint_P$ of $\disjoint_\mathrm{bin}$ gives another $\disjoint$-relation on $\mathsf{Top}$: conterminous pair $f_1 : X_1 \rightarrow Y \leftarrow X_2 : f_2$ in $\mathsf{Top}$ has $f_1 \disjoint_P f_2$ if no points $a_1 \in X_1$, $a_2 \in X_2$ have $f_i(a_i)$ lying in the same path-component of $Y$, i.e. $f_1 \disjoint_P f_2$ if there is no continuous path in $Y$ beginning in image $f_1(X_1)$ and ending in $f_2(X_2)$.
	
	Because path-components are always contained in connected components but the converse does not hold, the identity functor $\mathrm{id}_\mathsf{Top} : (\mathsf{Top}, \disjoint_C) \to (\mathsf{Top}, \disjoint_P)$ preserves but does not reflect $\disjoint$-relations.
\end{example}

\todo[color=cyan]{maybe there's a statement like invertible functor $F$ preserves $\disjoint$-relations if and only if $F^{-1}$ reflects $\disjoint$-relations. Prove if useful...}

\begin{remark}
	Any $\disjoint$-relation on a category can induce a notion of disjointness of subobjects in that category.
	In many familiar categories, the reverse holds: familiar notions of subobject-disjointness produce $\disjoint$-relations (but not necessarily orthogonality relations).
	Besides the minimal selection of examples presented above for illustration, and those of physical interest in \cref{sec:spacetimes-cats-relativistic,sec:spacetimes-cats-chiral}, other $\disjoint$-relations addressed in \cite[Chapter 2]{Grant-Stuart2022a} include: linear independence of subspaces in categories of vector spaces; orthogonality (in the sense of inner products) of subspaces in categories of Hilbert spaces; mutual commutativity of submonoids in the category of monoids in any symmetric monoidal category.
\end{remark}

Given a $\disjoint$-category $(\mathsf{C}, \disjoint_\mathsf{C})$ with subcategory inclusion $i : \mathsf{B} \hookrightarrow \mathsf{C}$, we refer to $(\mathsf{B}, \disjoint_i)$ as a \emph{$\disjoint$-subcategory} of $(\mathsf{C}, \disjoint_\mathsf{C})$.
Where the inclusion functor is not explicitly named, we may also denote the pullback of $\disjoint_\mathsf{C}$ to $\mathsf{B}$ along the inclusion as $\disjoint_\mathsf{C}$ again\footnote{Compare to \cref{ex:setwise_disjointness:concrete} of concrete categories and $\disjoint_\mathrm{set}$.}.

For brevity, we may refer to a $\disjoint$-category $(\mathsf{C}, \disjoint_\mathsf{C})$ merely by its underlying category $\mathsf{C}$ when the $\disjoint$-relation $\disjoint_\mathsf{C}$ is unambiguously implied.
Similarly, we may refer to a $\disjoint$-subcategory $(\mathsf{B}, \disjoint_\mathsf{C})$ of $(\mathsf{C}, \disjoint_\mathsf{C})$ as simply a $\disjoint$-subcategory $\mathsf{B}$ of $\mathsf{C}$.

\subsection{Morphisms which respect disjointness} \label{sec:disjointness:overlap-monics}

In any $\disjoint$-category, there is a special class of morphisms which `respect disjointness':

\begin{definition} \label{def:overlap-monic}
	Let $(\mathsf{C}, \disjoint_\mathsf{C})$ be a $\disjoint$-category, and let $h : c \to d$ be a morphism in $\mathsf{C}$. We say $h$ is \emph{overlap-monic} (or \emph{$\overlap$-monic}) if for any conterminous pair
	$f_1 : b_1 \rightarrow c \leftarrow b_2 : f_2$
% 	$b_1 \xrightarrow{f_1} c \xleftarrow{f_2} b_2$
	in $\mathsf{C}$,
	\begin{equation*}
		f_1 \disjoint_\mathsf{C} f_2
		\qquad \text{implies} \qquad
		h \circ f_1 \disjoint_\mathsf{C} h \circ f_2,
	\end{equation*}
	or equivalently,
	\begin{equation} \label{eqn:overlap-monics}
		h \circ f_1 \overlap_\mathsf{C} h \circ f_2
		\qquad \text{implies} \qquad
		f_1 \overlap_\mathsf{C} f_2.
	\end{equation}
\end{definition}
We avoid calling such morphisms `disjointness-preserving' to prevent conflation with $\disjoint$-preserving functors.
Instead, we name them $\overlap$-monic after the following analogy: in standard terminology, a morphism is monic if it satisfies a property similar to \cref{eqn:overlap-monics}, but with the relation $\overlap_\mathsf{C}$ replaced by the relation of equality of morphisms.

\begin{remark}
	We may rephrase property \labelcref{def:disjointness_rel:post-comp} of the definition of a disjointness relation $\disjoint_\mathsf{C}$ on category $\mathsf{C}$ as:
	\begin{itemize}
		\item[$3''.$] All isomorphisms in $\mathsf{C}$ are $\overlap$-monic.
	\end{itemize}
\end{remark}

\begin{example} \label{ex:setwise_disjointness:overlap_monics}
	In $(\mathsf{Set}, \disjoint_\mathrm{set})$ of \cref{ex:setwise_disjointness:sets}, a morphism $h : Y \to Z$ is $\overlap$-monic if and only if it is injective (equivalently, monic in $\mathsf{Set}$):
	
	Say $h$ is injective, and let maps $f_i : X_i \to Y$ have $f_1(X_1) \cap f_2(X_2) = \varnothing$ so that $f_1 \disjoint_\mathrm{set} f_2$.
	If either $X_i$ is the empty set, then trivially  $h \circ f_1 (X_1) \cap h \circ f_2 (X_2) = \varnothing$.
	Otherwise take any $x_i \in X_i$; since $f_1 (x_1) \neq f_2(x_2)$ and $h$ is injective, we have $h \circ f_1 (x_1) \neq h \circ f_2(x_2)$.
	Hence again $h \circ f_1 (X_1) \cap h \circ f_2 (X_2) = \varnothing$ in $Z$, i.e. $h \circ f_1 \disjoint_\mathrm{set} h \circ f_2$, so $h$ is $\overlap$-monic.
	
	On the other hand, say $h$ is not injective, so that there exist distinct $y_1, y_2 \in Y$ with $h(y_1) = h(y_2)$.
	Take $f_i : \left\{*\right\} \to Y$ the constant maps to $y_i$.
	Then $f_1 \disjoint_\mathrm{set} f_2$ since $\{y_1\} \cap \{ y_2\} = \varnothing$, but $h \circ f_1 \overlap_\mathrm{set} h \circ f_2$ since $\{ h(y_1) \} \cap \{ h (y_2) \} \neq \varnothing$; hence $h$ is not $\overlap$-monic.
\end{example}

\begin{example} \label{ex:binrel_disjointness:overlap-monics}
	In $(\mathsf{sBin}, \disjoint_\mathrm{bin})$ of \cref{ex:binrel_disjointness}, a morphism $h : (Y,R_Y) \to (Z, R_Z)$ is $\overlap$-monic if and only if it reflects relations, i.e. $(h(y_1), h(y_2)) \in R_Z$ implies $(y_1, y_2) \in R_Y$ for any $y_1,y_2 \in Y$:
	
	Say $h$ reflects relations; take relation-preserving maps $f_i : (X_i, R_{X_i}) \to (Y, R_Y)$ with
% 	$R_Y \cap \left[f_1(X_1) \times f_2(X_2)\right] = \varnothing$ so that 
	$h \circ f_1 \overlap_\mathrm{bin} h \circ f_2$, so there exists $x_i \in X_i$ with $(h\circ f_1(x_1), h \circ f_2(x_2)) \in R_Z$.
	Since $h$ reflects relations, this gives $(f_1(x_1), f_2(x_2)) \in R_Y$ so $f_1 \overlap_\mathrm{bin} f_2$; hence $h$ is $\overlap$-monic.
	
	Conversely, say $h$ is $\overlap$-monic; consider $y_1,y_2 \in Y$ such that $(h(y_1),h(y_2)) \in R_Z$.
	Take $f_i : (\{*\}, \varnothing) \to (Y, R_Y)$ the constant maps to $y_i$; these trivially preserve relations since the domain $\{*\}$ is equipped with the empty relation $\varnothing$.
	Then $h \circ f_1 \overlap_\mathrm{bin} h \circ f_2$ and thus $f_1 \overlap_\mathrm{bin} f_2$ since $h$ is $\overlap$-monic.
	This means $(y_1, y_2) \in R_Y$, so $h$ reflects relations.
\end{example}

\begin{proposition} \label{prop:functors-respect-overlap-monics}
	Say $F : (\mathsf{C}, \disjoint_\mathsf{C}) \to (\mathsf{D}, \disjoint_\mathsf{D})$ preserves and reflects $\disjoint$-relations.
	Then $F$ reflects $\overlap$-monics,
	i.e. for any morphism $h : c \to d$ in $\mathsf{C}$,	if $Fh$ is $\overlap$-monic in $\mathsf{D}$ then $h$ is $\overlap$-monic in $\mathsf{C}$.
	
	If $F$ is moreover full and essentially surjective on objects then it also preserves $\overlap$-monics,
	i.e. if $h$ is $\overlap$-monic, then $Fh$ is $\overlap$-monic.
\end{proposition}
\begin{proof}
	Say $h : c \to d$ is a morphism in $\mathsf{C}$ such that $Fh$ is $\overlap$-monic in $\mathsf{D}$, and consider any morphisms $f_i : b_i \rightarrow c$ in $\mathsf{C}$.
	If $f_1 \disjoint_\mathsf{C} f_2$ then $Ff_1 \disjoint_\mathsf{D} Ff_2$ since $F$ is $\disjoint$-preserving.
	Then
% 	\todo[color=red]{should I avoid inline $=$ next to $\disjoint$?}
% 	\begin{equation*}
% 		F( h\circ f_1) = 
% 		Fh \circ Ff_1 \disjoint_\mathsf{D} Fh \circ Ff_2
% 		= F (h \circ f_2),
% 	\end{equation*}
	$Fh \circ Ff_1 \disjoint_\mathsf{D} Fh \circ Ff_2$
	since $Fh$ is $\overlap$-monic in $\mathsf{D}$.
	But then $h \circ f_1 \disjoint_\mathsf{C} h \circ f_2$ in $\mathsf{C}$ since $F$ is $\disjoint$-reflecting; hence $h$ is $\overlap$-monic in $\mathsf{C}$.
	
	Say that $F$ is not only $\disjoint$-preserving and $\disjoint$-reflecting, but also full and essentially surjective on objects.
	Let $h : c \to d$ be $\overlap$-monic in $\mathsf{C}$, and take any morphisms $f_i : a_i \to Fc$ in $\mathsf{D}$ with $f_1 \disjoint_\mathsf{D} f_2$.
	Since $F$ is essentially surjective, there exist objects $b_i$ of $\mathsf{C}$ and isomorphisms $g_i : Fb_i \to a_i$ in $\mathsf{D}$.
	Since $F$ is full, there exist morphisms $k_i : b_i \to c$ in $\mathsf{C}$ such that $F k_i$ coincide in $\mathsf{D}$ with the composites
	\begin{equation*}
		Fb_i \xrightarrow{g_i} a_i \xrightarrow{f_i} Fc.
	\end{equation*}
	By stability under pre-composition, $f_1 \disjoint_\mathsf{D} f_2$ implies $f_1 \circ g_1 \disjoint_\mathsf{D} f_2 \circ g_2$, so $Fk_1 \disjoint_\mathsf{D} Fk_2$.
% 	\begin{equation*}
% 		F k_1 = f_1 \circ g_1 \disjoint_\mathsf{D} f_2 \circ g_2 = Fk_2.
% 	\end{equation*}
	Then $k_1 \disjoint_\mathsf{C} k_2$ since $F$ reflects $\disjoint$-relations, so $h \circ k_1 \disjoint_\mathsf{C} h \circ k_2$ since $h$ is $\overlap$-monic in $\mathsf{C}$.
	Because $F$ preserves $\disjoint$-relations, this gives
	\begin{equation*}
		Fh \circ f_1 \circ g_1
% 		=
% 		Fh \circ Fk_1
		\disjoint_\mathsf{D}
% 		Fh \circ Fk_2
% 		=
		Fh \circ f_2 \circ g_2.
	\end{equation*}
	Then we may pre-compose with $g_1^{-1}$ on the left and $g_2^{-1}$ on the right to find $Fh \circ f_1 \disjoint_\mathsf{D} Fh \circ f_2$, showing that $Fh$ is $\overlap$-monic in $\mathsf{D}$.
\end{proof}

\begin{example} \label{ex:disjoint_subcat:overlap-monics}
	For any $F : (\mathsf{B},\disjoint_F) \to (\mathsf{C}, \disjoint_\mathsf{C})$ where $\disjoint_F$ is the pullback of $\disjoint_\mathsf{C}$ along $F$, a morphism $h : c \to d$ in $\mathsf{B}$ is $\overlap$-monic in $\mathsf{B}$ if $Fh$ is $\overlap$-monic in $\mathsf{C}$ by \cref{prop:functors-respect-overlap-monics}.
	
	In particular, morphisms in a $\disjoint$-subcategory $\mathsf{B}$ of $\mathsf{C}$ are necessarily $\overlap$-monic in $\mathsf{B}$ if they are $\overlap$-monic in $\mathsf{C}$.
	Similarly, in any concrete category equipped with $\disjoint_\mathrm{set}$ as in \cref{ex:setwise_disjointness:concrete}, if a morphism $f$ is an injective map (and so $\overlap$-monic in $\mathsf{Set}$ by \cref{ex:setwise_disjointness:overlap_monics}) then it is $\overlap$-monic.
\end{example}

\begin{example} \label{ex:top_disjointness:overlap-monics}
	Consider $(\mathsf{Top}, \disjoint_P)$ as in \cref{ex:top_disjointness}.
	By definition of the pullback relation $\disjoint_P$, \cref{ex:binrel_disjointness:overlap-monics,prop:functors-respect-overlap-monics} apply to give that a continuous map $f : X \to Y$ is $\overlap$-monic with respect to $\disjoint_P$ if $f$ reflects the relations $P$, i.e. for any $x, x' \in X$ if there exists a path in $Y$ from $f(x)$ to $f(x')$ then there exists a path in $X$ from $x$ to $x'$.
	Equivalently, $f : X \to Y$ is $\overlap$-monic with respect to $\disjoint_P$ if the induced map $\pi_0(f) : \pi_0(X) \to \pi_0(Y)$ is injective.

	In this example, the converse is also true: if $\pi_0(f)$ is injective, then $f$ is $\overlap$-monic with respect to $\disjoint_P$.
	For, consider $g_i : W_i \to X$ such that $f \circ g_1 \overlap_P f \circ g_2$; then there exist $w_i \in W_i$ such that there is a path in $Y$ from $f\circ g_1(w_1)$ to $f\circ g_2(w_2)$.
	Since $\pi_0(f)$ is injective, there is a path in $X$ from $g_1(w_1)$ to $g_2(w_2)$; thus $g_1 \overlap_P g_2$.
\end{example}

The generalisation from $\perp$-relations to $\disjoint$-relations on categories allows the relations to describe a variety of intuitive notions of disjointness on familiar categories.
Nonetheless, for AQFT it remains necessary to work with categories of spacetimes equipped with $\perp$-relations \cite{BeniniSchenkelWoike2021}.

Observe that $\perp$-relations are precisely $\disjoint$-relations with respect to which all morphisms are $\overlap$-monic.
Using this, we can produce from any $\disjoint$-category a wide $\disjoint$-subcategory on which the $\disjoint$-relation is moreover a $\perp$-relation:

\begin{definition} \label{def:overlap-monic-subcat}
	Let $(\mathsf{C}, \disjoint_\mathsf{C})$ be a $\disjoint$-category.
	Define $\OverlapMonics \left(\mathsf{C}, \disjoint_\mathsf{C}\right)$ to be the $\disjoint$-subcategory consisting of all objects of $\mathsf{C}$, and only those morphisms in $\mathsf{C}$ which are $\overlap$-monic with respect to $\disjoint_\mathsf{C}$.
\end{definition}
All isomorphisms, and in particular all identities, in $\mathsf{C}$ are $\overlap$-monic.
It follows from \cref{eqn:overlap-monics} that a composition of $\overlap$-monics is also $\overlap$-monic.
Hence $\OverlapMonics(\mathsf{C}, \disjoint_\mathsf{C})$ is indeed a subcategory of $\mathsf{C}$.
As per \cref{ex:disjoint_subcat:overlap-monics}, all morphisms in $\OverlapMonics \left(\mathsf{C}, \disjoint_\mathsf{C}\right)$ are $\overlap$-monic because they are $\overlap$-monic in the ambient $\disjoint$-category $(\mathsf{C}, \disjoint_\mathsf{C})$ by definition; so, $\OverlapMonics \left(\mathsf{C}, \disjoint_\mathsf{C}\right)$ is an orthogonal category.

\begin{example}
% 	From the characterisation of $\overlap$-monics in $(\mathsf{Set},\disjoint_\mathrm{set})$ in 
	From \cref{ex:setwise_disjointness:overlap_monics} of $(\mathsf{Set},\disjoint_\mathrm{set})$, we have orthogonal category $\OverlapMonics(\mathsf{Set}, \disjoint_\mathrm{set})$ consisting of sets and injective maps.

% 	From the characterisation of $\overlap$-monics in $(\mathsf{sBin}, \disjoint_\mathrm{bin})$ in 
	From \cref{ex:binrel_disjointness:overlap-monics}, we have orthogonal category $\OverlapMonics(\mathsf{sBin}, \disjoint_\mathrm{bin})$ consisting of sets equipped with symmetric binary relations, and morphisms that preserve and reflect the relations.

% 	From the characterisation of $\overlap$-monics in $(\mathsf{Top},\disjoint_P)$ in 
	From \cref{ex:top_disjointness:overlap-monics}, we have orthogonal category $\OverlapMonics(\mathsf{Top}, \disjoint_P)$ consisting of topological spaces $X$ and continuous maps $f:X \to Y$ such that the induced map $\pi_0(f) : \pi_0(X) \to \pi_0(Y)$ is injective.
\end{example}

% END

\section{Causal disjointness and spacetimes categories for relativistic QFT} \label{sec:spacetimes-cats-relativistic}
% BEGIN

We now turn our toolkit of disjointness relations to the study of categories of spacetimes suitable for algebraic quantum field theory, in the categorical formulation sometimes called locally covariant QFT \parencite{BrunettiFredenhagenVerch2003,FewsterVerch2015}.

By \emph{spacetime} $(M,g)$ we mean a smooth, time-oriented, Lorentzian manifold  $M$ (Hausdorff, paracompact) of dimension at least 2, where $g$ is the metric tensor of signature $(-+\ldots+)$.
% $(1, \dim M -1)$; we use the mostly-plus signature convention.
Where explicit notation for the metric is not needed, we refer to spacetime $(M,g)$ merely by $M$.
We do not assume $M$ to be connected, in contrast to several standard references \cite{BeemEhrlichEasley1996,HawkingEllis1973,ONeill1983}.

A map $f : M \to N$ between spacetimes $(M,g)$ and $(N,h)$ is \emph{conformal} if it is smooth and there is some $\omega_f \in C^\infty(M)$ such that $f^*h = e^{\omega_f} g$.
% The everywhere-strictly-positive function $e^{\omega_f}$ is called the \emph{conformal factor} of $f$.
Further, $f$ is a \emph{local isometry} if it is conformal with $\omega_f = 0$.
% i.e. unit conformal factor.
Any conformal map $f : M \to N$ is a smooth immersion, i.e. its tangent map $df_p : T_p M \to T_{f(p)} N$ is injective for all points $p \in M$.
For, if $v \in T_p M$ has $df_p (v) = 0$ then for any other $u \in T_p M$ we have $e^{\omega_f(p)} g_p(u,v) = (f^*h)_p (u,v) = h_{f(p)} (df_p (u), df_p(v)) = 0$; then $v = 0$ since $g$ is non-degenerate.
If $f : M \to N$ is both a conformal map and a diffeomorphism, then its inverse $f^{-1} : N \to M$ is also conformal with $\omega_{f^{-1}} = - \omega_f \circ f^{-1}$.

\begin{definition}
	For any integer $d \geq 1$, denote by $\mathsf{SpTm}_{d+1}$ the category whose objects are spacetimes of dimension $d+1$, and whose morphisms are local isometries, composing as functions.
\end{definition}

\begin{remark} \label{rm:sptm-mor-local-diffeo}
	The morphisms of $\mathsf{SpTm}_{d+1}$ are smooth immersions.
	Since all objects of $\mathsf{SpTm_{d+1}}$ share the same fixed dimension $d+1$, smooth immersions between them are local diffeomorphisms; see for example \cite[Proposition 4.8]{Lee2013}.
	Hence all morphisms of $\mathsf{SpTm}_{d+1}$ are local diffeomorphisms, and in particular open maps.
	
	If a morphism $g : M \to N$ of $\mathsf{SpTm}_{d+1}$ is injective, it is therefore a smooth embedding \parencite[Proposition 4.22]{Lee2013} and so a diffeomorphism onto its image $g(M) \subseteq N$.
\end{remark}

The category of spacetimes most often used as domain for relativistic AQFTs is not $\mathsf{SpTm}_{d+1}$, but rather a subcategory of it:
\begin{definition} \label{def:loc}
	Denote by $\mathsf{Loc}_{d+1}$ the category whose objects are globally hyperbolic spacetimes of dimension $d+1$, and whose morphisms $f:M \to N$  are injective local isometries such that $f(M)$ is a causally convex subset of $N$.
\end{definition}
We recall the definitions of causal convexity and global hyperbolicity in \cref{sec:spacetimes-cats-relativistic:causality,sec:spacetimes-cats-relativistic:causal_properties} below.
Since it is unlikely to cause ambiguity, we leave the dimension $d+1$ implicit and write only $\mathsf{SpTm}$, $\mathsf{Loc}$ etc.
In this work we only consider categories of spacetimes wherein all spacetimes share the same fixed dimension.

In comparison to $\mathsf{SpTm}$, the morphisms of $\mathsf{Loc}$ are notably constrained to be injective and have causally convex image.
From a purely categorical viewpoint, these constraints seem unappealing: they go further than merely asking that morphisms `preserve structure'.
As a consequence, some basic constructions do not possess their usual universal properties in $\mathsf{Loc}$;
for instance, the disjoint union of spacetimes is a coproduct in $\mathsf{SpTm}$ but not in $\mathsf{Loc}$.

Originally in \cite{BrunettiFredenhagenVerch2003}, these constraints on the morphisms of $\mathsf{Loc}$ were arrived at by \textit{ad hoc} physical reasoning specific to relativistic QFTs.
In this section, we relate $\mathsf{Loc}$ to (a subcategory of) $\mathsf{SpTm}$ systematically via the machinery of disjointness relations described in \cref{sec:disjointness}.
In particular, we show that the constraints imposed on morphisms in $\mathsf{Loc}$ amount precisely to the requirement that they be $\overlap$-monic, as per \cref{def:overlap-monic}.

\begin{remark}
	In fact, the category of spacetimes more usually used for relativistic AQFT is the subcategory $\mathsf{Loc}^{\mathrm{o,to}}$ of $\mathsf{Loc}$ whose objects are oriented and whose morphisms are orientation- and time-orientation preserving.
	The relationship we show below between $\mathsf{Loc}$ and $\mathsf{SpTm}$ has an analogue between $\mathsf{Loc}^\mathrm{o,to}$ and the subcategory $\mathsf{SpTm}^\mathrm{o,to}$ of $\mathsf{SpTm}$ similarly restricted to oriented spacetimes and orientation- and time-orientation-preserving morphisms.
	For generality, we work with $\mathsf{Loc}$ and $\mathsf{SpTm}$.
\end{remark}

\subsection{Causal relations on a spacetime} \label{sec:spacetimes-cats-relativistic:causality}

Before defining on $\mathsf{SpTm}$ a $\disjoint$-relation to describe causal disjointness, we recall some standard notions in the causality theory of spacetimes to fix notation.

% % --------old version--------------
% Causal properties of a spacetime are frequently encoded in binary relations on that spacetime.
% For any homogeneous binary relation $R \subseteq X \times X$ on a set $X$, we may always construct its \emph{symmetric closure}, which we denote $s(R)$ or $sR$, as the smallest symmetric relation which contains $R$.
% Explicitly, denoting the transpose of relation $R$ as $R^T := \Set{(y,x) \in X \times X | (x,y) \in R }$,
% \todo{check consistent set-builder notation throughout document}
% it is a standard result that $s(R) = R \cup R^T$.
% 
% When $X$ is a topological space, for any binary relation $R \subseteq X \times X$ we may consider also its \emph{topological closure} $\overline{R}$, namely the smallest closed subset of topological space $X \times X$ (equipped with the product topology) which contains $R$.
% As a consequence of the definition of the product topology, it is straightforward to observe that $(x,y) \in \overline{R}$ if and only if for any open neighbourhoods $U$ of $x$ and $V$ of $y$ in $X$, there exist $x' \in U$ and $y' \in V$ with $(x',y') \in R$.
% The operations of taking symmetric and topological closures of a relation commute:
% % ---------------------------------

Causal properties of a spacetime are frequently encoded in binary relations on that spacetime.
For any homogeneous binary relation $R \subseteq X \times X$ on a set $X$, we may always take its \emph{symmetric closure} -- the smallest symmetric relation on $X$ containing $R$ -- which we denote as $s(R)$ or $sR$.
Denoting the transpose of relation $R$ as $R^T := \Set{(y,x) \in X \times X | (x,y) \in R }$, it is a standard result that $s(R) = R \cup R^T$.

If $X$ is a topological space then, for any binary relation $R \subseteq X \times X$, we may consider also its \emph{topological closure} $\overline{R}$ -- namely, the smallest closed subset of topological space $X \times X$ (equipped with the product topology) containing $R$.

\begin{lemma} \label{lem:top_closure_relation}
	Let $R$ be a binary relation on a topological space $X$, and let $x, y \in X$.
	Then $(x,y) \in \overline{R}$ if and only if for any open neighbourhoods $U$ of $x$ and $V$ of $y$ in $X$, there exist $x' \in U$ and $y' \in V$ with $(x',y') \in R$.
\end{lemma}
\begin{proof}
	The product topology on $X \times X$ has a basis consisting of sets $U \times V$ for $U,V$ open sets in $X$.
	It is a standard characterisation of topological closures that $(x,y) \in \overline{R}$ if and only if every basis set $U \times V$ containing $(x,y)$ intersects $R$.
\end{proof}

The operations of taking symmetric and topological closures of a relation commute:

\begin{lemma} \label{lem:sym_top_closures_commute}
	$\overline{s(R)} = s\left(\overline{R}\right)$ for any binary relation $R$ on topological space $X$.
\end{lemma}
\begin{proof}
\begin{equation*}
	\overline{s(R)}
	= \overline{R \cup R^T}
	= \overline{R} \cup \overline{R^T}
	= \overline{R} \cup \overline{R}^T
	= s \left(\overline{R}\right),
\end{equation*}
where the third equality follows from the fact that $R^T$ is the image of $R$ under the homeomorphism $X \times X \to X \times X$ which rearranges $(x,y) \mapsto (y,x)$.
\end{proof}

If sets $X$ and $Y$ are equipped with binary relations $R_X$ and $R_Y$ respectively, recall that a map $f : X \to Y$ \emph{preserves} the relations if $(x,y) \in R_X$ implies $(f(x), f(y)) \in R_Y$ for all $x,y \in X$.
In cases where the relations $R_X$ and $R_Y$ are reflexive, we say $f : X \to Y$ \emph{strictly preserves} the relations if $f$ preserves the relations and $(x,y) \in R_X$ has $f(x) = f(y)$ only if $x = y$.
The map $f$ \emph{reflects} relations if  $(f(x), f(y)) \in R_Y$ implies $(x,y) \in R_X$.
A map preserves (reflects) given relations if and only if it preserves (reflects) their transposes.
From $s(R) = R \cup R^T$ it then follows that any map which preserves (reflects) some given relations also preserves (reflects) their symmetric closures.

The binary relations on spacetimes that are relevant for our purposes are defined in terms of causal curves.
Recall that for $(M,g)$ a spacetime and $I \subseteq \mathbb{R}$ an open interval, a smooth curve $\gamma : I \to M$ is \emph{timelike}, \emph{causal} or \emph{null} if at every $t \in I$ its tangent vector $\dot{\gamma}_t$ is respectively timelike, causal or null, i.e. has $\dot{\gamma}_t \neq 0$ with respectively $g(\dot{\gamma}_t, \dot{\gamma}_t) < 0$, $g(\dot{\gamma}_t, \dot{\gamma}_t) \leq 0$ or $g(\dot{\gamma}_t, \dot{\gamma}_t) = 0$.
Timelike and null curves are both special cases of causal curves.
Causal curve $\gamma$ is \emph{future-directed} if $\dot{\gamma}_t$ lies in the future lightcone, and \emph{past-directed} if $\dot{\gamma}_t$ lies in the past lightcone.
Smooth causal curves are regular by our definition; no smooth causal curve is future-directed for some values of its parameter $t$ and past-directed for others.
For $J \subseteq \mathbb{R}$ a not-necessarily-open interval, a continuous and piecewise-smooth curve $\gamma : J \to M$ is causal if either each smooth piece of $\gamma$ defined on an open sub-interval of $J$ is future-directed causal, or each such smooth piece is past-directed causal.
In the former case, $\gamma$ is future-directed causal; in the latter case, $\gamma$ is past-directed causal.
We adopt the convention that if a curve $\gamma$ is merely called (future- or past-directed) causal then this means it is piecewise-smooth and (respectively future- or past-directed) causal; when smoothness of $\gamma$ is required this is stated explicitly.

Let $\gamma : I \to M$ be a future-directed causal curve.
If it exists, the limit from above $\lim_{t \to \inf I^+} \gamma(t)$ is the \emph{past-endpoint} of $\gamma$. (We take $\inf I := -\infty$ if $I \subseteq \mathbb{R}$ is unbounded below; similarly $\sup I := + \infty$ if $I$ is unbounded above.)
If $\gamma$ has no past-endpoint, then we say $\gamma$ is \emph{past-inextendable}.
Similarly, the \emph{future-endpoint} of $\gamma$ is the limit from below $\lim_{t \to \sup I^-} \gamma(t)$ if it exists; if no future-endpoint exists then we say $\gamma$ is \emph{future-inextendable}.
$\gamma$ is called \emph{inextendable} if it is both past-inextendable and future-inextendable.

Any spacetime $(M,g)$ has a \emph{causal relation}:
\begin{equation*}
	J_M := \Set{(p,q) \in M \times M |
	\begin{array}{c}
		p = q, \text{ or there exists a future-directed causal} \\
		\text{curve } \gamma : [0,1] \to M \text{ with } \gamma(0)=p \text{ and } \gamma(1)=q.
	\end{array} }.
\end{equation*}
Where unambiguous, we may write only $J$ and leave the spacetime $M$ implicit.
The relation $J$ is reflexive by definition.
$J$ is also transitive by concatenation of causal curves, so that $J$ is a preorder on $M$.

Standard notations for the \emph{causal future and past of $p \in M$} are respectively
\begin{equation*}
	J^+(p) := \Set{ q \in M | (p,q) \in J}
	\qquad \text{and} \qquad
	J^-(p) := \Set{ q \in M | (q,p) \in J}.
\end{equation*}
\emph{Causal diamonds} in $M$ are sets of the form $J^+(p) \cap J^-(q)$ for some $p,q \in M$.

The symmetric closure $sJ = J \cup J^T$ of $J$ may be explicitly characterised as:
\begin{equation*}
	sJ_M := \Set{(p,q) \in M \times M | \begin{array}{c}
		p = q, \text{ or there exists any causal curve} \\
		\gamma : [0,1] \to M \text{ with } \gamma(0)=p \text{ and } \gamma(1)=q.
	\end{array} }.
\end{equation*}

\begin{remark}
	It is possible to generalise from piecewise-smooth to merely continuous causal curves; see \cite[pg. 184]{HawkingEllis1973} and \cite[Proposition 3.16]{MinguzziSanchez2008} for equivalent definitions.
	However, this generalisation leaves the causal relations $J$ unchanged: as explained in \cite[Remark 2.14]{Minguzzi2019}, if there is a future-directed continuous causal curve from point $p$ to point $q$ in a spacetime $(M,g)$, then there is also a future-directed piecewise-smooth causal curve from $p$ to $q$ in $(M,g)$.
	For this reason, all causal curves in this work are taken to be at least piecewise-smooth.
\end{remark}

For any open subset $U \subseteq M$ understood as a submanifold with the induced metric and time-orientation, it is automatically true that
\begin{equation*}
	J_U \subseteq J_M \cap [U \times U].
\end{equation*}
The converse $J_M \cap [U \times U] \subseteq J_U$ is not true for general $U$.

A subset $U \subseteq M$ of spacetime $M$ is called \emph{causally convex} if any causal curve $\gamma : [0,1] \to M$ with endpoints $\gamma(0), \gamma(1)$ lying in $U$ has $\gamma([0,1]) \subseteq U$, i.e. $\gamma$ remains always in $U$.
If $U \subseteq M$ is open and causally convex, then $J_U = J_M \cap [U \times U]$.
We may interpret this as saying that the `internal' causal relation $J_U$ on $U$ coincides with the causal relation $J_M \cap [U \times U]$ `induced' by restriction of $J_M$.
We note that open, causally convex subsets are not the only subsets that possess this property: take for instance the open but not causally convex strip $U := \Set{(t,x) \in \mathbb{R}^{1,1} | -1 < x < 1 }$ in two-dimensional Minkowski space $\mathbb{R}^{1,1}$ with its usual coordinates.

Given a conformal map $f : (M,g) \to (N,h)$ between spacetimes, any causal curve $\gamma : I \to M$ maps to a causal curve $f \circ \gamma : I \to N$, since
\begin{equation*}
	h \left(df_{\gamma(t)} (\dot{\gamma}_t), df_{\gamma(t)} (\dot{\gamma}_t)\right) = (f^* h) (\dot{\gamma}_t, \dot{\gamma}_t) = e^{\omega_f} g (\dot{\gamma}_t, \dot{\gamma}_t)
\end{equation*}
has the same sign as $g (\dot{\gamma}_t, \dot{\gamma}_t)$ for all $t\in I$.
For $\gamma$ future-directed, the curve $f \circ \gamma$ will be future- or past-directed if $f$ preserves or reverses time-orientation respectively on the path-component of $M$ which contains $\gamma$.
Consequently, any conformal map $f : M \to N$ \emph{preserves $J$ up time-orientation reversal}: for any $(p,q) \in J_M$, it follows that $\left(f(p), f(q)\right) \in J_N$ when $f$ preserves time-orientation on the path-component of $M$ containing $p$ and $q$, and $\left(f(q), f(p)\right) \in J_N$ when $f$ reverses time-orientation on that path-component.
Regardless of time-orientation reversals, any conformal map also preserves the symmetric closures $sJ$ of causal relations.

A conformal map $f : M \to N$ \emph{strictly preserves $J$ up to time-orientation reversal} if it preserves $J$ up to time-orientation reversal and for any $(p,q) \in J_M$, we have $f(p) = f(q)$ only if $p = q$.
Not all conformal maps strictly preserve $J$ up to time-orientation reversal.
Consider, for instance, the quotient map from two-dimensional Minkowski spacetime $\mathbb{R}^{1,1}$ which identifies $(t,x) \sim (t + 1, x)$ in standard coordinates.
The points $p = (0,0)$ and $q = (1,0)$ in $\mathbb{R}^{1,1}$ have $(p,q) \in J_{\mathbb{R}^{1,1}}$ as exhibited by the future-directed causal curve $\gamma : [0,1] \to \mathbb{R}^{1,1}, \gamma(\tau) = (\tau , 0)$.
While $p$ and $q$ are distinct in $\mathbb{R}^{1,1}$, their images under the quotient coincide; note that $\gamma$ becomes a closed causal curve in the quotient.

While all conformal maps preserve $J$ up to time-orientation reversal, some may also reflect it. Again, we must account for time-orientation reversal:
\begin{definition}
	For spacetimes $M$ and $N$, a conformal map $f : M \to N$  \emph{reflects $J$ up to time-orientation reversal} if for any $p,q \in M$, when $(f(p), f(q)) \in J_N$ it follows that $p$ and $q$ lie in the same path-component $\widetilde{M}$ of $M$, and:
	\begin{itemize}
		\item $(p,q) \in J_M$ if $f$ preserves time-orientation on $\widetilde{M}$, and
		\item $(q,p) \in J_M$ if $f$ reverses time-orientation on $\widetilde{M}$.
	\end{itemize}
\end{definition}

If $f : M \to N$ is a conformal embedding, i.e. a conformal map which is also a smooth embedding, then any future-directed causal curve $\gamma : I \to N$ with image $\gamma(I)$ contained in $f(M)$ has unique curve $\delta : I \to M$ such that $\gamma = f \circ \delta$.
Since $f$ (and hence the inverse of its codomain-restriction) is conformal, this $\delta$ is timelike, causal or null when $\gamma$ is timelike, causal or null respectively; it is future-directed if $f$ preserves time-orientation on the path-component of $M$ containing $\delta$, and past-directed if $f$ reverses time-orientation on that path-component.
Using this, we can immediately identify examples of maps which reflect $J$ up to time-orientation reversal:

\begin{proposition} \label{prop:causally_convex_conf_emb_reflect_J}
	Let $M, N$ be spacetimes and $f : M \to N$ a conformal map.
	If $f$ is a smooth embedding whose image $f(M)$ is a causally convex subset of $N$, then $f$ reflects $J$ up to time-orientation reversal.
\end{proposition}
\begin{proof}
	Take $p,q \in M$ such that $\left(f(p), f(q)\right) \in J_N$.
	Since $f$ is injective, if $f(p) = f(q)$ then $p=q$ so that both $(p,q)$ and $(q,p)$ lie in $J_M$.
	So assume $f(p) \neq f(q)$; then there is a future-directed causal curve $\gamma : [0,1] \to N$ with $\gamma(0) = f(p)$ and $\gamma(1) = f(q)$.
	By causal convexity of $f(M)$ in $N$, this implies that $\gamma([0,1]) \in f(M)$; so, there is a unique causal curve $\delta : [0,1] \to M$ with $\gamma = f \circ \delta$ and in particular $\delta(0) = p, \delta(1) = q$.
	If $f$ preserves time-orientation on the path-component of $M$ containing $\delta$, then $\delta$ is future-directed, exhibiting $(p,q) \in J_M$.
	Likewise if $f$ reverses time-orientation on that component then $\delta$ is past-directed and so exhibits $(q,p) \in J_M$.
\end{proof}

\begin{remark}
	The preceding proposition shows that all morphisms of $\mathsf{Loc}$ reflect $J$ up to time-orientation reversal.
	Indeed, the original motivation\footnote{In \cite{BrunettiFredenhagenVerch2003}, the category now commonly called $\mathsf{Loc}$ is denoted $\mathfrak{Man}$.}
	stated in \cite{BrunettiFredenhagenVerch2003} for choosing to impose that any morphism $f : M \to N$ of $\mathsf{Loc}$ have causally convex image was that the `intrinsic' and `induced' causal structures on $f(M)$ should coincide.
	In the language of causal relations, this stipulates that $J_{f(M)} = J_N \cap \left[f(M) \times f(M)\right]$, which for locally isometric embedding $f$ is equivalent to $f$ reflecting $J$ up to time-orientation reversal.
	
	It is not true that any conformal $f : M \to N$ which reflects $J$ up to time-orientation reversal is necessarily an embedding with causally convex image.
	For instance, the inclusion of the strip $U := \Set{ (t,x) \in \mathbb{R}^{1,1} | -1 < x < 1}$ into $\mathbb{R}^{1,1}$ has image which is not causally convex in $\mathbb{R}^{1,1}$; nonetheless, because $J_U = J_{\mathbb{R}^{1,1}} \cap [U \times U]$ the inclusion reflects $J$.
	
	However, we show in \cref{prop:reflect_J_causally_convex_conf_emb} below
	that a converse of \cref{prop:causally_convex_conf_emb_reflect_J} does hold when $f : M \to N$ is conformal map between globally hyperbolic spacetimes.
\end{remark}

For any subsets $U$ and $V$ of a topological space $X$, and continuous curve $\gamma : I \to X$, we say that \emph{$\gamma$ connects $U$ and $V$} if there are $a,b \in I$ such that $\gamma(a) \in U$ and $\gamma(b) \in V$.
For points $p,q \in X$, curve $\gamma : I \to X$ connects $p$ and $q$ if $\gamma$ connects $\left\{p\right\}$ and $\left\{q\right\}$.
If a causal curve $\gamma$ connects subsets $U$ and $V$ of a spacetime $M$, then $\gamma$ exhibits that $sJ \cap [U \times V]$ is non-empty.

Note that this terminology agnostic as to whether $\gamma$ is future- or past-directed.
When we need to specify direction, we instead say that \emph{$\gamma$ is future-directed causal from $U$ to $V$} or \emph{$\gamma$ is past-directed causal from $U$ to $V$}.
The former case exhibits that $J \cap [U \times V]$ is non-empty, and the latter that $J \cap [V \times U]$ is non-empty.
If $\gamma : [0,1] \to M$ is future-directed causal from $U$ to $V$ then the reparameterised curve $\gamma' : [0,1] \to M$ with $\gamma'(t) := \gamma(1-t)$ is past-directed causal from $V$ to $U$.

For open subsets $U$ and $V$ in spacetime $M$, we have the following equivalent characterisations of the notion that $U$ and $V$ are spacelike-separated:

\begin{proposition} \label{prop:caus_discon_equiv}
	For open subsets $U$ and $V$ of spacetime $M$, the following are equivalent:
	\begin{enumerate}[label=\textup{(\roman*)}]
		\item \label{prop:caus_discon_equiv:sJ}
		$sJ_M$ does not intersect $U \times V$,
		
		\item \label{prop:caus_discon_equiv:J}
		$J_M$ intersects neither $U \times V$ nor $V \times U$,
		
		\item \label{prop:caus_discon_equiv:spacelikesep}
		There exists no causal curve in $M$ connecting $U$ and $V$.
	\end{enumerate}
\end{proposition}
\begin{proof}
	\labelcref{prop:caus_discon_equiv:sJ,prop:caus_discon_equiv:J} are equivalent since $sJ = J \cup J^T$.
	\labelcref{prop:caus_discon_equiv:J} implies \labelcref{prop:caus_discon_equiv:spacelikesep} trivially.
	Say \labelcref{prop:caus_discon_equiv:sJ} does not hold, so there is some $p \in U$ and $q \in V$ with $(p,q) \in sJ_M$.
	Then either there is a causal curve in $M$ connecting $p$ and $q$, showing immediately that \labelcref{prop:caus_discon_equiv:spacelikesep} does not hold, or $p = q$.
	In the latter case, since $U$ and $V$ are open, their intersection $U \cap V$ is also open and hence a subspacetime of $M$.
	So, there exists a causal curve in $U \cap V$ through the point $p = q$ (take for instance any timelike geodesic in $U \cap V$ through $p=q$), exhibiting also that \labelcref{prop:caus_discon_equiv:spacelikesep} does not hold.
\end{proof}

\subsection{Causal disjointness on the category of spacetimes}
\label{sec:spacetimes-cats-relativistic:disjointness}

From the causal relations $J$ available on any spacetime, we can produce a disjointness relation\footnote{Compare to \cref{ex:top_disjointness}.} on the category $\mathsf{SpTm}$ to encode the notion of causal disjointness (i.e. spacelike-separateness) of subspacetimes:

\begin{definition} \label{def:causal_discon}
	Define a $\disjoint$-relation $\disjoint_J$ on $\mathsf{SpTm}$ as follows: for any conterminous pair $f_1 : M_1 \rightarrow N \leftarrow M_2 : f_2$, say that $f_1 \disjoint_J f_2$ if
% 	$sJ_N \cap [f_1(M_1) \times f_2(M_2)] = \varnothing$.
	$sJ_N$ does not intersect $f_1(M_1) \times f_2(M_2)$.

	We call $\disjoint_J$ the \emph{causal disjointness relation} on $\mathsf{SpTm}$.
\end{definition}

By \cref{prop:caus_discon_equiv} and since the morphisms of $\mathsf{SpTm}$ are open maps, we can equivalently say that $f_1 \disjoint_J f_2$ if there is no causal curve in $N$ connecting $f_1(M_1)$ and $f_2(M_2)$.
That $\disjoint_J$ is a valid $\disjoint$-relation follows since it is the pullback of the $\disjoint$-relation $\disjoint_\mathrm{bin}$ of \cref{ex:binrel_disjointness} along the functor $\mathsf{SpTm} \to \mathsf{sBin}$ which sends spacetimes $M$ to their underlying sets equipped with binary relation $sJ_M$; morphisms $f : M \to N$ preserve $sJ$ because they are conformal maps.

From \cref{ex:binrel_disjointness:overlap-monics,prop:functors-respect-overlap-monics}, it follows that if a morphism of $\mathsf{SpTm}$ reflects $sJ$ then it is $\overlap$-monic.
However, a full characterisation of $\overlap$-monics in $\mathsf{SpTm}$ is as follows:

\begin{theorem} \label{thm:overlap-monics-SpTm}
	A morphism $f : M \to N$ in $\left(\mathsf{SpTm}, \disjoint_J\right)$ is $\overlap$-monic if and only if it reflects $\overline{sJ}$,
	i.e. for any $p,q \in M$:
	\begin{equation*}
		\left(f(p), f(q)\right) \in \overline{sJ_N}
		\qquad \text{implies} \qquad
		(p,q) \in \overline{sJ_M}.
	\end{equation*}
\end{theorem}

\begin{proof}
$(\Leftarrow)$
	Say $f : M \to N$ reflects $\overline{sJ}$, and consider any conterminous pair $g_1 : O_1 \rightarrow M \leftarrow O_2 : g_2$ with $f \circ g_1 \overlap_J f \circ g_2$.
	Then there exist $p \in O_1$ and $q \in O_2$ with
	\begin{equation*}
		\left(f \circ g_1(p), f \circ g_2(q)\right) \in sJ_N \subseteq \overline{sJ_N}.
	\end{equation*}
	Since $f$ reflects $\overline{sJ}$, it follows that $\left(g_1(p), g_2(q)\right) \in \overline{sJ_M}$.
	Recalling that morphisms of $\mathsf{SpTm}$ are necessarily open maps as per \cref{rm:sptm-mor-local-diffeo}, we have open set $g_1(O_1) \times g_2(O_2)$ in $M \times M$ which intersects $\overline{sJ_M}$.
	By definition of topological closure, this implies that $g_1(O_1) \times g_2(O_2)$ intersects $sJ_M$ and hence $g_1 \overlap_J g_2$.
	
$(\Rightarrow)$
	Say that $f : M \to N$ is $\overlap$-monic;
	take any $p,q \in M$ with $\left(f(p), f(q)\right) \in \overline{sJ_N}$.
	Choose arbitrary open neighbourhoods $U$ of $p$ and $V$ of $q$ in $M$, and denote the inclusion maps $\iota_U : U \hookrightarrow M$ and $\iota_V : V \hookrightarrow M$.
	Then $U$ and $V$ are subspacetimes of $M$ equipped with the induced metrics and time-orientations, so the inclusions $\iota_U$ and $\iota_V$ are valid $\mathsf{SpTm}$ morphisms.
	Using that $f$ is an open map, $f(U) \times f(V)$ is an open set in $N \times N$ intersecting $\overline{sJ_N}$ at least at $\left(f(p), f(q)\right)$.
	Thus also $f(U) \times f(V)$ intersects $sJ_N$ by definition of topological closure, which demonstrates that $f \circ \iota_U \overlap_J f \circ \iota_V$.
	Since $f$ is $\overlap$-monic, this implies $\iota_U \overlap_J \iota_V$, and hence that $U \times V$ intersects $sJ_M$, i.e. there exist $p' \in U$ and $q' \in V$ such that $(p',q') \in sJ_M$.
	Because this holds for arbitrary open neighbourhoods $U$ of $p$ and $V$ of $q$, it shows that $(p,q) \in \overline{sJ_M}$.
\end{proof}

In general, the property that a map $f$ reflects $\overline{sJ}$ does not straightforwardly simplify:

\begin{example} \label{ex:minkowski_minus_point}
	We exhibit a conformal map which reflects $\overline{J}$ and hence $\overline{sJ}$, but not $J$ or $sJ$.
	Let $\mathbb{R}^{1,1}$ be two-dimensional Minkowski spacetime, and define open subset $U := \mathbb{R}^{1,1} \setminus \left\{r\right\}$ for some point $r = (t_0,x_0) \in \mathbb{R}^{1,1}$.
	$U$ becomes a spacetime when equipped with the induced metric and time-orientation from $\mathbb{R}^{1,1}$; then the inclusion $U \hookrightarrow \mathbb{R}^{1,1}$ is a local isometry and so a conformal map.

	Consider the points $p := (t_0 - \Delta, x_0 - \Delta)$ and $q := (t_0 + \Delta, x_0 + \Delta)$, for any $\Delta > 0$, as illustrated in \cref{fig:minkowski_minus_point}.
	Then $(p,q) \in J_{\mathbb{R}^{1,1}}$ as exhibited by the null curve in $\mathbb{R}^{1,1}$ given by $\tau \mapsto (t_0 + \tau, x_0 + \tau)$ for $\tau \in [-\Delta, \Delta]$.
	
	However, due to the removal of $r$, there is no causal curve in $U$ connecting $p$ and $q$.
	Hence $(p,q) \not \in sJ_U$, so the inclusion $U \hookrightarrow \mathbb{R}^{1,1}$ does not reflect $J$ or $sJ$.

	Nonetheless, it follows from \cref{lem:top_closure_relation} that $(p,q) \in \overline{J_U}$.
	For, consider any open neighbourhoods $V_p, V_q \subseteq U$ of $p$ and $q$ respectively; then there exists some $\epsilon > 0$ such that $p' := (t_0 + \epsilon - \Delta, x_0 - \Delta)$ is in $V_p$ and $q' := (t_0 + \epsilon + \Delta, x_0 + \Delta)$ is in $V_q$.
	It holds that $(p',q') \in J_U$, as exhibited by the null curve in $U$ given by $\tau \mapsto (t_0 + \epsilon + \tau, x_0 + \tau)$ for $\tau \in [-\Delta, \Delta]$; notice that this curve avoids the point $r = (t_0, x_0)$ which is missing from $U$.

% 	The points $p := (t_0 - \Delta, x_0 - \Delta)$ and $q := (t_0 + \Delta, x_0 + \Delta)$ illustrated in \cref{fig:minkowski_minus_point} have $(p,q) \in J_{\mathbb{R}^{1,1}}
% % 	\subseteq \overline{sJ_{\mathbb{R}^{1,1}}}
% 	$ for any $\Delta > 0$.
% 	Clearly, $(p,q) \not \in sJ_U$ because of the removal of $r$ from $U$; so the inclusion $U \hookrightarrow \mathbb{R}^{1,1}$ does not reflect $J$ or $sJ$.
% 	However, it is true that $(p,q) \in \overline{J_U}$ since any arbitrary neighbourhood of $p$ can be connected to any arbitrary neighbourhood of $q$ by a future-directed causal curve.

	Similar considerations for points $\tilde{p} := (t_0 - \Delta, x_0 + \Delta)$ and $\tilde{q} := (t_0 + \Delta, x_0 - \Delta)$ show that the inclusion $U \hookrightarrow \mathbb{R}^{1,1}$ does reflect $\overline{J}$.

	\begin{figure}[bt]
		\centering
% 		% makebox command allows figure to be wider than text and remain centred.
% 		\makebox[\linewidth][c]{
		\begin{tikzpicture}[scale=0.8]
			%axes:
			\draw[->] (-1,0) -- (4,0) node [anchor=north]{$x$};
			\draw[->] (0,-1) -- (0,4) node [anchor=east]{$t$};
			
			\draw[dashed]	(1,1) node[anchor=east]{$p$}
						--	(2,2) node[anchor=north west]{$r$}
						--	(3,3) node[anchor=west]{$q$};
			\filldraw[fill=white] (2,2) circle (1.5pt);
			\filldraw (3,3) circle (1pt);
			\filldraw (1,1) circle (1pt);
		\end{tikzpicture}
		\caption{Minkowski spacetime $\mathbb{R}^{1,1}$ with a point $r$ removed. Points $p$ and $q$ have $(p,q) \in \overline{J}$ but cannot be connected by a causal curve.}
		\label{fig:minkowski_minus_point}
	\end{figure}
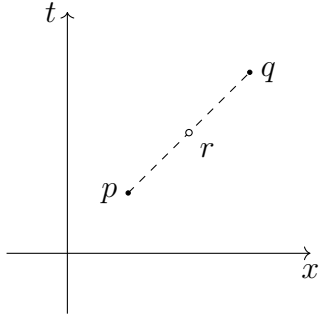
\end{example}

Examples such as this make the characterisation of $\overlap$-monics in $\mathsf{SpTm}$ given by \cref{thm:overlap-monics-SpTm} difficult to handle, both conceptually and technically.
However, the characterisation simplifies greatly when we restrict to spacetimes with suitable causal properties.

% Subspacetimes which intersect as sets are necessarily not causally disjoint.
% Formally, the category $\mathsf{SpTm}$ is concrete, so we may consider on it the setwise-disjointness relation $\disjoint_\mathrm{set}$ as in \cref{ex:setwise_disjointness:concrete}:
% \begin{proposition} \label{prop:causal_implies_setwise_disjointness}
% 	The identity functor $(\mathsf{SpTm}, \disjoint_J) \to (\mathsf{SpTm}, \disjoint_\mathrm{set})$ preserves but does not reflect $\disjoint$-relations.
% \end{proposition}
% \begin{proof}
% 	Say $M_1 \xrightarrow{f_1} N \xleftarrow{f_2} M_2$ has $f_1 \disjoint_J f_2$.
% 	Since $J_N$ is reflexive, so too is $sJ_N$, i.e. $sJ_N$ contains the diagonal $\Delta_N := \left\{(p,p) \in N \times N\right\}$.
% 	Then
% 	\begin{equation*}
% 		\Delta_N \cap \left[f_1(M_1) \times f_2(M_2)\right] \subseteq sJ_N \cap \left[f_1(M_1) \times f_2(M_2)\right] = \varnothing,
% 	\end{equation*}
% 	so that subsets $f_1(M_1)$ and $f_2(M_2)$ of $N$ do not intersect, showing $f_1 \disjoint_\mathrm{set} f_2$.
% 	
% 	The converse that $f_1 \disjoint_\mathrm{set} f_2$ implies $f_1 \disjoint_J f_2$ does not hold: for counter-example, take any two open subsets of a spacetime which are non-intersecting but not spacelike separated.
% \end{proof}

Before demonstrating this, we observe that
unions of subspacetimes respect causal disjointness:
\todo[color=cyan]{this looks like a coverage, making $\mathsf{SpTm}$ a site - check pullback stability}
\todo{better to move the following lemma closer to where I use it?}
\begin{lemma} \label{lem:unions_preserve_causal_disjointness}
	Let $M_1 \xrightarrow{f_1} N \xleftarrow{f_2} M_2$ be any conterminous pair in $\mathsf{SpTm}$,
	and let $\left\{i_\alpha : U_\alpha \hookrightarrow M_1 \right\}_{\alpha \in \mathcal{I}}$ be a family (not necessarily countable) of subspacetime inclusions in $\mathsf{SpTm}$ such that $M_1 = \bigcup_{\alpha \in \mathcal{I}} U_\alpha$.
	\begin{equation*}
		\text{If} \qquad
		\begin{tikzcd}[cramped, column sep=small]
								& N \\
			M_1 \arrow[ur, "f_1"{name=1, near start}] \arrow[rr, phantom, "\disjoint_J"]	& & M_2 \arrow[ul, "f_2"{near start, name=2,swap}] \\
			U_\alpha \arrow[u, hook, "i_\alpha"]
		\end{tikzcd}
		\qquad \text{for each } \alpha \in \mathcal{I}, \text{ then} \qquad
		\begin{tikzcd}[cramped, column sep=small]
								& N \\
			M_1 \arrow[ur, "f_1"{name=1, near start}]	& & M_2 \arrow[ul, "f_2"{near start, name=2,swap}]
			\arrow[phantom, from=1, to=2, "\disjoint_J"{below}]
		\end{tikzcd}.
	\end{equation*}
\end{lemma}
\begin{proof}
	Assume $f_1 \overlap_J f_2$, so there exist $p \in M_1$, $q \in M_2$ with $(f_1(p), f_2(q)) \in sJ_N$.
	Since $M_1 = \bigcup_{\alpha \in \mathcal{I}} U_\alpha$, there is an $\alpha \in \mathcal{I}$ with $p \in U_\alpha$.
	Then $(f_1 \circ i_\alpha(p), f_2(q)) \in sJ_N$, showing that $f_1 \circ i_\alpha \overlap_J f_2$.
\end{proof}
The same property passes to any $\disjoint$-subcategory of $\mathsf{SpTm}$:
\begin{corollary} \label{cor:unions_preserve_causal_disjointness}
	Let $\mathsf{C}$ be any $\disjoint$-subcategory of $(\mathsf{SpTm}, \disjoint_J)$. Then \cref{lem:unions_preserve_causal_disjointness} applies with $\mathsf{C}$ in place of $\mathsf{SpTm}$.
\end{corollary}
\begin{proof}
	Say $\left\{ i_\alpha : U_\alpha \hookrightarrow M_1\right\}_{\alpha \in \mathcal{I}}$ are the given subspacetime inclusions in $\mathsf{C}$,
	and say for each $\alpha \in \mathcal{I}$ that $f_1 \circ i_\alpha \disjoint f_2$ in $\mathsf{C}$.
	Since the $\disjoint$-subcategory inclusion $\mathsf{C} \hookrightarrow \mathsf{SpTm}$ preserves $\disjoint$-relations, this means that $f_1 \circ i_\alpha \disjoint_J f_2$ in $\mathsf{SpTm}$.
	Then $f_1 \disjoint_J f_2$ in $\mathsf{SpTm}$ by \cref{lem:unions_preserve_causal_disjointness}, so $f_1 \disjoint_J f_2$ in $\mathsf{C}$ since $\mathsf{C} \hookrightarrow \mathsf{SpTm}$ reflects $\disjoint$-relations.
\end{proof}
Another immediate corollary of \cref{lem:unions_preserve_causal_disjointness} is that, since coproducts in $\mathsf{SpTm}$ are disjoint unions, coproducts preserve the causal disjointness relation $\disjoint_J$.

\subsection{Spacetimes with good causal properties}
\label{sec:spacetimes-cats-relativistic:causal_properties}

In the definition of $\mathsf{Loc}$, objects are restricted to be only those spacetimes which are globally hyperbolic.
Global hyperbolicity is the strongest of a hierarchy of causal properties that a spacetime may possess.
As we restrict from $\mathsf{SpTm}$ to full subcategories of spacetimes obeying specified causal properties, the characterisation in \cref{thm:overlap-monics-SpTm} of $\overlap$-monics with respect to $\disjoint_J$ simplifies.

The most significant simplification occurs when we restrict our spacetimes to a level of the causal hierarchy called causal simplicity.
After demonstrating such simplification, we relate $\overlap$-monics to the definition of $\mathsf{Loc}$.

We begin by recalling those causal properties in the hierarchy that will be relevant; we refer the reader to \textcite{Minguzzi2019} for an extensive modern review of the hierarchy.
Along the way, we also note some consequences for conformal maps when their domains or codomains have given causal properties.

A spacetime $(M,g)$ is \emph{causal} if it contains no closed causal curves. Equivalently, the causal relation $J$ is anti-symmetric: if $(p,q) \in J$ and $(q,p) \in J$ then $p = q$.
This property lies near the bottom of the hierarchy, but already leads to several useful facts about conformal maps that take causal spacetimes as domain or codomain.

\begin{proposition} \label{prop:conf_map_to_causal_strictly_preserves_J}
	Let $M, N$ be spacetimes with $N$ causal, and let $f : M \to N$ be a conformal map.
	Then $f$ strictly preserves $J$ up to time-orientation reversal.
\end{proposition}
\begin{proof}
	Since it is conformal, $f$ preserves $J$ up to time-orientation reversal.
	Take $(p,q) \in J_M$ such that $f(p) = f(q)$, and assume that there is a future-directed causal curve $\gamma : [0,1] \to M$ from $p = \gamma(0)$ to $q = \gamma(1)$.
	Then $f \circ \gamma : [0,1] \to N$ has $f \circ \gamma (0) = f(p) = f(q) = f \circ \gamma (1)$, and so is a closed causal curve.
	(Since $f$ is a smooth immersion, $f \circ \gamma$ is not constant.)
	This is in contradiction with $N$ a causal spacetime, so we must have $p = q$.
\end{proof}

\begin{corollary} \label{cor:codomain_caus_reflect_J_iff_reflect_sJ}
	Let $M, N$ be spacetimes with $N$ causal, and let $f : M \to N$ be a conformal map.
	Then $f$ reflects $J$ up to time-orientation reversal if and only if $f$ reflects the symmetric closure $sJ$.
\end{corollary}
\begin{proof}
$(\Rightarrow)$:
	Trivial using $sJ = J \cup J^T$.
	
$(\Leftarrow)$:
	Take any $p,q \in M$ such that $\left(f(p), f(q)\right) \in J_N$.
	Since $f$ reflects $sJ$ and $J_N \subseteq sJ_N$, it follows that $(p,q) \in sJ_M$ so
% 	at least one of the following cases holds:
	either $(p,q) \in J_M$ or $(q,p) \in J_M$.
	In either case, $p$ and $q$ lie in the same path-component $\widetilde{M}$ of $M$.
	
	Say $f$ preserves time-orientation on $\widetilde{M}$; we show that the first case $(p,q) \in J_M$ necessarily holds by showing that the second case $(q,p) \in J_M$ implies $p = q$.
	Assume that $(q,p) \in J_M$; conformal $f$ preserves $J$ on $\widetilde{M}$, from which it results that both $(f(q), f(p))$ and $(f(p), f(q))$ lie in $J_N$.
	Hence $f(p) = f(q)$ because $N$ is causal.
% 	($J_N$ is anti-symmetric).
	But also because $N$ is causal, $f$ \emph{strictly} preserves $J$ by the preceding proposition;
% 	\cref{prop:conf_map_to_causal_strictly_preserves_J};
	therefore $p=q$.
	
	Similarly if $f$ reverses time-orientation on $\widetilde{M}$, then the first case $(p,q) \in J_M$ implies that $p=q$, so that the second case $(q,p) \in J_M$ necessarily holds.
% 	So $f$ reflects $J$ up to time-orientation reversal.
\end{proof}

% In case the domain of a conformal map is a causal spacetime, we have the following 

\begin{proposition} \label{prop:domain_caus_reflect_J_injective}
	Let $M, N$ be spacetimes with $M$ causal, and let $f : M \to N$ be a conformal map.
	If $f$ reflects $J$ up to time-orientation reversal then $f$ is injective.
\end{proposition}
\begin{proof}
	Take $p,q \in M$ with $f(p) = f(q)$.
	Then both $\left(f(p), f(q)\right)$ and $\left(f(q), f(p)\right)$ lie in $J_N$, so that both $(p,q)$ and $(q,p)$ lie in $J_M$ for $f$ which reflects $J$ up to time orientation reversal.
	But then $p=q$ since $J_M$ is antisymmetric.
\end{proof}

Proceeding further up the hierarchy of causal properties,
a spacetime $(M,g)$ is \emph{non-imprisoning} (also called non-totally-imprisoning) if no future-inextendable causal curve is contained in a compact set.
Equivalently, $(M,g)$ is non-imprisoning if no past-inextendable causal curve is contained in a compact set.
\todo{Minguzzi states this, but the proof goes via a lot of extra technology.}
Any non-imprisoning spacetime is causal, since any closed causal curve has compact image and may be made inextendable by winding over its image.

A spacetime $(M,g)$ is \emph{causally simple} if it is causal and its causal relation $J$ is topologically closed ($J = \overline{J}$).
Topological closedness of $J$ has several equivalent characterisations \cite[Theorem 4.12]{Minguzzi2019}:
\begin{proposition} \label{prop:characterisation_J_closed}
	For any spacetime $(M,g)$, the following are equivalent:
	\begin{enumerate}[label=\textup{(\roman*)}]
		\item the causal relation $J$ is a closed subset of $M \times M$,
		\item causal futures $J^+(p)$ and pasts $J^-(p)$ are closed subsets of $M$ for all $p \in M$,
		\item causal diamonds $J^+(p) \cap J^-(q)$ are closed subsets of $M$ for all $p,q \in M$.
	\end{enumerate}
\end{proposition}
\todo[color=cyan]{Can I find a concise proof that causally simple implies non-imprisoning? Try to use that $J^+(p)$ is closed, so $J^+(p) \cap K$ is compact for any compact $K$...}
If spacetime $(M,g)$ is causally simple, then it is non-imprisoning; this is shown via several intermediate levels of the causal hierarchy in \cite[Section 4]{Minguzzi2019}.

A spacetime $(M,g)$ is \emph{globally hyperbolic} if it is causal and its causal diamonds $J^+(p) \cap J^-(q)$ are compact for all $p,q \in M$.
Compactness of causal diamonds implies their closedness, so any globally hyperbolic spacetime is causally simple.
There are several equivalent characterisations of global hyperbolicity; one important characterisation is that the spacetime $(M,g)$ have a Cauchy surface.

\begin{remark} \label{rm:cauchy_surfaces_smooth}
	Cauchy surfaces may defined merely to be subsets $S \subseteq M$ met exactly once by every inextendable timelike curve \cite{BernalSanchez2003,ONeill1983}, from which it follows that they are achronal, topologically embedded continuous submanifolds of codimension 1, and met by any inextendable causal curve (possibly at more than one point, in case of null segments of causal curves).
	
	However, \textcite{BernalSanchez2003} show that in any globally hyperbolic spacetime there exists a Cauchy surface which is not merely topologically embedded and achronal but also smoothly embedded and spacelike (so acausal).
	So without loss of generality, we take Cauchy surfaces to be smoothly embedded and spacelike.
\end{remark}

\begin{remark}
	Classical definitions of global hyperbolicity and causal simplicity appear more restrictive than those we have given. For instance, \textcite{HawkingEllis1973} use strong causality in place of causality in the definition of global hyperbolicity. \textcite{BernalSanchez2007} show that the definitions given above are equivalent to the classical ones.
\end{remark}

We note here that arbitrarily small globally hyperbolic neighbourhoods can be found around any point $p$ in an any spacetime $M$: there is a neighbourhood base at $p$ of $M$ consisting of open, globally hyperbolic sets.
See \cite[Corollary 2]{Minguzzi2014} and \cite[Theorem 1.35]{Minguzzi2019}, though these discuss such neighbourhood bases with many more good topological and pseudo-Riemannian properties than are needed here.

When a conformal map has domain a globally hyperbolic spacetime and codomain any spacetime of the same dimension, we obtain a converse to \cref{prop:causally_convex_conf_emb_reflect_J}:

\begin{proposition} \label{prop:reflect_J_causally_convex_conf_emb}
	Let $M, N$ be spacetimes of the same dimension $\dim M = \dim N$ with $M$ globally hyperbolic, and let $f : M \to N$ be a conformal map.
	If $f$ reflects $J$ up to time-orientation then $f$ is a smooth embedding with image $f(M)$ causally convex in $N$.
\end{proposition}
\begin{proof}
	Since $f$ is conformal and $\dim M = \dim N$, it follows that $f$ is a local diffeomorphism as in \cref{rm:sptm-mor-local-diffeo}.
	By \cref{prop:domain_caus_reflect_J_injective}, $f$ is injective because $M$ is globally hyperbolic and hence causal.
	Thus $f : M \to N$ is a smooth embedding.
	
	Pick any distinct $p,q \in M$ with future-directed causal curve $\gamma : [0,1] \to N$ from $\gamma(0) = f(p)$ to $\gamma(1) = f(q)$.
	To show that $f(M)$ is causally convex in $N$, assume for the sake of contradiction that $\gamma$ leaves $f(M)$.
	Then $\gamma^{-1}\left(N \setminus f(M) \right)$ is a non-empty subset of $(0,1)$; define
	\begin{equation*}
		t_0 := \inf \gamma^{-1} \left( N \setminus f(M) \right).
	\end{equation*}
	Since $f$ is a local diffeomorphism, it is an open map; hence $N \setminus f(M)$ is closed and $\gamma^{-1}(N \setminus f(M))$ is closed. So, $t_0$ lies in $\gamma^{-1} \left( N \setminus f(M) \right)$.

	Consider the restriction $\gamma|_{[0,t_0)} : [0, t_0) \to f(M)$.
	Because $f : M \to N$ is a conformal embedding, there is a unique causal curve $\delta : [0, t_0) \to M$ with $\gamma|_{[0,t_0)} = f \circ \delta$.
	Since $f$ reflects $J$ up to time-orientation, $p$, $q$ and curve $\delta$ lie in the same path-component $\widetilde{M}$ of $M$.
	If $f$ preserves time-orientation on $\widetilde{M}$ then $\delta$ is future-directed. Moreover, $\delta$ is future-inextendable: for, say $\lim_{t \to t_0^-} \delta(t) \in M $ exists; then
	\begin{equation*}
		\gamma(t_0) 
		= \lim_{t \to t_0^-} \gamma(t)
		= \lim_{t \to t_0^-} f \circ \delta (t)
		= f \Big( \lim_{t \to t_0^-} \delta(t) \Big)
		\in f(M),
	\end{equation*}
	in contradiction with $t_0 \in \gamma^{-1} \left( N \setminus f(M) \right)$.
	Likewise if $f$ reverses time-orientation on $\widetilde{M}$ then $\delta$ is past-directed and past-inextendable.

	For any $t \in [0,t_0)$, the future-directed causal curve $\gamma$ exhibits that $(f(p), f \circ \delta(t))$ and $(f \circ \delta (t), f(q))$ lie in $J_N$.
	Say $f$ preserves time-orientation on $\widetilde{M}$.
	Then because $f$ reflects $J$ up to time-orientation reversal, it follows that $(p, \delta(t))$ and $(\delta(t), q)$ lie in $J_M$ for all $t \in [0,t_0)$, so that the curve $\delta$ is contained in the causal diamond $J_M^+(p) \cap J_M^-(q)$.
	Since $M$ is globally hyperbolic, its causal diamonds are compact.
	But then we have a future-inextendable causal curve $\delta$ in $M$ contained in a compact set, which contradicts the non-imprisoning property of $M$.
	Similarly, when $f$ reverses time-orientation on $\widetilde{M}$, we find that the past-inextendable causal curve $\delta$ is contained in compact causal diamond $J_M^+(q) \cap J_M^-(p)$ which again contradicts the non-imprisoning property of $M$.
\end{proof}

Formally, we may realise the hierarchy of causal properties as a nested sequence of full $\disjoint$-subcategories of $\mathsf{SpTm}$:
\begin{definition}
	Let $\kappa$ denote any causal property in the hierarchy; we define $\kappa \mathsf{SpTm}$ to be the full $\disjoint$-subcategory of $(\mathsf{SpTm}, \disjoint_J)$ consisting of those spacetimes which satisfy $\kappa$.
\end{definition}
If $\kappa$ and $\lambda$ are any two causal properties in the hierarchy with $\kappa$ stronger than $\lambda$, then the inclusion $\kappa \mathsf{SpTm} \hookrightarrow \lambda \mathsf{SpTm}$ is also a full $\disjoint$-subcategory inclusion.

This notation also works if we regard the null causal property (i.e. having no specified causal property) as the lowest level of the hierarchy;
then $\mathsf{SpTm}$ is $\kappa \mathsf{SpTm}$ for $\kappa$ the null property.
In particular, we have full $\disjoint$-subcategory inclusions
\begin{equation*}
	\mathsf{GlobHypSpTm} \hookrightarrow \mathsf{CausSimSpTm} \hookrightarrow \mathsf{SpTm},
\end{equation*}
for prefixes $\mathsf{GlobHyp}$ and $\mathsf{CausSim}$ denoting globally hyperbolic and causally simple, respectively.

\subsection{Overlap-monics in categories of spacetimes with good causal properties}
\label{sec:spacetimes-cats-relativistic:overlap-monics}

To make use of the general characterisation of $\overlap$-monics in $\mathsf{SpTm}$ given in \cref{thm:overlap-monics-SpTm} to characterise $\overlap$-monics in a subcategory $\kappa \mathsf{SpTm}$, we must first ensure that moving between such subcategories does not change which morphisms are $\overlap$-monic.

Since $\kappa \mathsf{SpTm} \hookrightarrow \lambda \mathsf{SpTm}$ are $\disjoint$-subcategory inclusions, they reflect $\overlap$-monics by \cref{prop:functors-respect-overlap-monics}.
It is also true that they preserve $\overlap$-monics, essentially because we can cover any spacetime with open globally hyperbolic neighbourhoods and the $\disjoint$-relation $\disjoint_J$ respects such coverings as per \cref{lem:unions_preserve_causal_disjointness}:

\begin{proposition} \label{prop:spacetime_inclusions_preserve_overlap-monics}
	The inclusion functors $(\kappa \mathsf{SpTm}, \disjoint_J) \hookrightarrow (\lambda \mathsf{SpTm}, \disjoint_J)$ preserve $\overlap$-monics.
\end{proposition}
\begin{proof}
	Let $N \xrightarrow{g} P$ be $\overlap$-monic in $\kappa \mathsf{SpTm}$, and consider any conterminous pair $f_1 : M_1 \rightarrow N \leftarrow M_2 : f_2$ in $\lambda \mathsf{SpTm}$.
	Since $\lambda$ is a weaker causal property than $\kappa$, $M_1$ and $M_2$ are not necessarily objects of $\kappa \mathsf{SpTm}$.
	
	However, each point $p \in M_1$ has an open globally hyperbolic neighbourhood $U_p \subseteq M_1$; denote by $i_p : U_p \hookrightarrow M_1$ the inclusion.
	Then since global hyperbolicity is the strongest causal property, the neighbourhood $U_p$ is an object of $\kappa \mathsf{SpTm}$, and the inclusion $i_p$ is a morphism of $\lambda \mathsf{SpTm}$.
	Similarly, for each $q \in M_2$ there is an open globally hyperbolic neighbourhood $V_q$ in $\kappa \mathsf{SpTm}$ with inclusion $j_q : V_q \hookrightarrow M_2$ in $\lambda \mathsf{SpTm}$.

	Say $f_1 \disjoint_J f_2$ in $\lambda \mathsf{SpTm}$. Then by stability of $\disjoint_J$ under pre-composition, 
	\begin{equation*}
		\begin{tikzcd}[cramped, column sep=small]
								& N \\
			M_1 \arrow[ur, "f_1"{name=1, near start}] \arrow[rr, phantom, "\disjoint_J"{pos=0.55}]	& & M_2 \arrow[ul, "f_2"{near start, name=2,swap}] \\
			U_p \arrow[u, hook, "i_p"]	& & V_q \arrow[u, hook', "j_q"{swap}]
		\end{tikzcd}
		\qquad \text{in $\lambda \mathsf{SpTm}$ for all $p \in M_1$ and $q \in M_2$.}
	\end{equation*}
	Since $\kappa \mathsf{SpTm} \hookrightarrow \lambda \mathsf{SpTm}$ is a full subcategory, $f_1 \circ i_p : U_p \to N$ and $f_2 \circ j_q : V_q \to N$ are morphisms in $\kappa \mathsf{SpTm}$. The inclusion functor also reflects $\disjoint$-relations, so
	\begin{equation*}
		\begin{tikzcd}[cramped, column sep=small]
								& N \\
			U_p \arrow[ur, "f_1 \circ i_p"{name=1, near start}]	& & V_q \arrow[ul, "f_2 \circ j_q"{near start, name=2,swap}]
			\arrow[phantom, from=1, to=2, "\disjoint_J"{below,pos=0.55}]
		\end{tikzcd}
		\qquad \text{in $\kappa \mathsf{SpTm}$ for all $p \in M_1$ and $q \in M_2$.}
	\end{equation*}
	Since $g : N \to P$ is $\overlap$-monic in $\kappa \mathsf{SpTm}$, this gives $g \circ f_1 \circ i_p \disjoint_J g \circ f_2 \circ j_q$ in $\kappa \mathsf{SpTm}$. But since the inclusion functor preserves $\disjoint$-relations,
	\begin{equation*}
		\begin{tikzcd}[cramped, column sep=small]
								& P \\
								& N \arrow[u, "g"]\\
			M_1 \arrow[ur, "f_1"{name=1, near start}] \arrow[rr, phantom, "\disjoint_J"]	& & M_2 \arrow[ul, "f_2"{near start, name=2,swap}] \\
			U_p \arrow[u, hook, "i_p"]	& & V_q \arrow[u, hook', "j_q"{swap}]
		\end{tikzcd}
		\qquad \text{in $\lambda \mathsf{SpTm}$ for all $p \in M_1$ and $q \in M_2$.}
	\end{equation*}
	By \cref{cor:unions_preserve_causal_disjointness}, and since clearly $\bigcup_{p\in M_1} U_p = M_1$ and $\bigcup_{q \in M_2} V_q = M_2$, we conclude that $g \circ f_1 \disjoint_J g \circ f_2$ in $\lambda \mathsf{SpTm}$ so $g$ is $\overlap$-monic in $\lambda \mathsf{SpTm}$.
\end{proof}

Each inclusion $\kappa \mathsf{SpTm} \hookrightarrow \lambda \mathsf{SpTm}$ both preserves and reflects $\overlap$-monics: for any morphism $f$ in $\kappa \mathsf{SpTm}$, $f$ is $\overlap$-monic in $\kappa \mathsf{SpTm}$ if and only if it is $\overlap$-monic in $\lambda \mathsf{SpTm}$ for any weaker causal property $\lambda$.
In particular, $f$ is $\overlap$-monic in $\kappa \mathsf{SpTm}$ if and only if it satisfies the general characterisation given by \cref{thm:overlap-monics-SpTm} of $\overlap$-monics in $\mathsf{SpTm}$.
Using this, we can find greatly simplified characterisations of $\overlap$-monics when we restrict to spacetimes with strong causal properties:

\begin{theorem} \label{thm:overlap-monics-CausSimSpTm}
	A morphism $f : M \to N$ in $(\mathsf{CausSimSpTm}, \disjoint_J)$ is $\overlap$-monic if and only if it reflects causal relations $J$ up to time-orientation reversal.
\end{theorem}
\begin{proof}
	$f : M \to N$ is $\overlap$-monic in $\mathsf{CausSimSpTm}$ if and only if it is $\overlap$-monic in $\mathsf{SpTm}$, hence if and only if it reflects $\overline{sJ}$ by \cref{thm:overlap-monics-SpTm}.
	Each object of $\mathsf{CausSimSpTm}$ has $\overline{sJ} = s\left(\overline{J}\right) = sJ$, where the first equality holds because taking topological and symmetric closures of relations commutes (\cref{lem:sym_top_closures_commute}).
	Hence $f$ is $\overlap$-monic if and only if it reflects $sJ$.
	But since $N$ is also causal, $f$ reflects $sJ$ if and only if it reflects $J$ up to time-orientation reversal by \cref{cor:codomain_caus_reflect_J_iff_reflect_sJ}.
\end{proof}

\begin{theorem} \label{thm:overlap-monics-GlobHypSpTm}
	A morphism $f : M \to N$ in $(\mathsf{GlobHypSpTm}, \disjoint_J)$ is $\overlap$-monic if and only if the following equivalent conditions hold:
	\begin{enumerate}[label=\textup{(\roman*)}]
		\item
		$f$ reflects $J$ up to time-orientation reversal,
		
		\item
		$f$ is injective and has causally convex image $f(M)$ in $N$.
	\end{enumerate}
\end{theorem}
\begin{proof}
	$f : M \to N$ is $\overlap$-monic in $\mathsf{GlobHypSpTm}$ if and only if it is $\overlap$-monic in $\mathsf{CausSimSpTm}$, hence if and only if it reflects $J$ up to time-orientation reversal by \cref{thm:overlap-monics-CausSimSpTm}.
	
	By \cref{prop:reflect_J_causally_convex_conf_emb}, if $f$ reflects $J$ up to time-orientation reversal then $f$ is injective with causally convex image.
	Conversely, if $f$ is injective then (since $f$ is a local diffeomorphism as per \cref{rm:sptm-mor-local-diffeo}) it is an embedding; so if $f$ is injective with causally convex image then $f$ reflects $J$ up to time-orientation reversal by \cref{prop:causally_convex_conf_emb_reflect_J}.
\end{proof}

This characterisation of $\overlap$-monics in $\mathsf{GlobHypSpTm}$ exactly matches with the extra constraints imposed on morphisms in the definition of $\mathsf{Loc}$.
Moreover, the orthogonality relation $\perp$ of spacelike separateness on $\mathsf{Loc}$ \cite[Example 3.21]{BeniniSchenkelWoike2021} coincides with the restriction of $\disjoint_J$:
\begin{equation*}
	(\mathsf{Loc}, \perp) = \OverlapMonics(\mathsf{GlobHypSpTm}, \disjoint_J).
\end{equation*}

We conclude that the category $\mathsf{Loc}$ often used as domain for relativistic AQFT functors may be arrived at systematically by the procedure proposed in \cref{sec:introduction}:
first formulate the category $\mathsf{GlobHypSpTm}$ of globally hyperbolic spacetimes and all structure-preserving maps between them.
Also formulate an appropriate disjointness relation $\disjoint_J$ on $\mathsf{GlobHypSpTm}$ to encode physical limitations on signal propagation to be imposed by the causality axiom.
Finally, restrict to the wide subcategory of $\overlap$-monics with respect to $\disjoint_J$; this is $\mathsf{Loc}$.
The disjointness relation becomes an orthogonality relation on the subcategory, so we may formulate an AQFT as a functor,
\begin{equation*}
	\mathcal{A} : \OverlapMonics(\mathsf{GlobHypSpTm}, \disjoint_J) \to \mathsf{Obs},
\end{equation*}
obeying the causality axiom.

The time-slice axiom seems to play no significant role in this procedure to construct $\mathsf{Loc}$.
We can define an appropriate class $\widetilde{W}$ of Cauchy maps (morphisms $f : M \to N$ where $f(M)$ contains a Cauchy surface of $N$) in the original spacetimes category $\mathsf{GlobHypSpTm}$.
Then the time-slice axiom is simply implemented with respect to the intersection $W := \widetilde{W} \cap \operatorname{mor}\mathsf{Loc}$ of $\widetilde{W}$ with the class of morphisms in the subcategory $\mathsf{Loc}$: the functor $\mathcal{A}$ above satisfies the time-slice axiom if $\mathcal{A}f$ is an isomorphism for every $f \in W$.

\todo[color=cyan]{remark on using chronal relation $I$ instead of causal $J$?}

\begin{remark} \label{rm:CLoc}
	The same procedure applies when the structure to be preserved by morphisms is not the full Lorentzian metric structure, but merely the conformal structure.
	Formulating a category $\mathsf{CSpTm}$ of spacetimes and conformal maps in place of the category $\mathsf{SpTm}$ of spacetimes and local isometries, the characterisations presented above still apply.
	Namely, after equipping $\mathsf{CSpTm}$ with $\disjoint$-relation $\disjoint_J$ defined similarly to \cref{def:causal_discon} and restricting to globally hyperbolic spacetimes, we obtain full $\disjoint$-subcategory $(\mathsf{GlobHypCSpTm}, \disjoint_J)$.
	Then the orthogonal subcategory of $\overlap$-monics is
	\begin{equation*}
		(\mathsf{CLoc}, \perp) = \OverlapMonics(\mathsf{GlobHypCSpTm}, \disjoint_J),
	\end{equation*}
	consisting of globally hyperbolic spacetimes and conformal maps which are injective with causally convex image; $\mathsf{CLoc}$ is the appropriate category of spacetimes for defining (non-chiral) CFTs \cite{Pinamonti2009}.
\end{remark}

% END

\section{Chiral disjointness and spacetimes categories for chiral CFT} \label{sec:spacetimes-cats-chiral}
% BEGIN

As shown in the previous section, the category $\mathsf{Loc}$ may be constructed as the subcategory of $\overlap$-monics in a category of globally hyperbolic spacetimes and all appropriate structure-preserving maps.
In this section, we apply the same construction in the context of chiral CFTs on two-dimensional (2D) spacetimes.
After developing some useful technical tools in \cref{sec:spacetimes-cats-chiral:frames} and using them to produce a hierarchy of chiral properties of 2D oriented spacetimes in \cref{sec:spacetimes-cats-chiral:chiral-properties},
the construction leads to a proposal in \cref{sec:spacetimes-cats-chiral:overlap-monics} of a category $\chi\mathsf{Loc}$ to be used as the domain of chiral CFT functors.
In \cref{sec:spacetimes-cats-chiral:comparison}, we compare this proposal to established approaches to chiral CFTs in the AQFT literature.

Let us begin by recalling an elementary example of a chiral CFT on a fixed spacetime:
\begin{example} \label{ex:free_fermion}
	A free fermion field theory on 2D Minkowski space $\mathbb{R}^{1,1}$ is described by the Lagrangian
	\begin{equation*}
		\mathcal{L} = i \overline{\Psi} \gamma^\mu \partial_\mu \Psi,
	\end{equation*}
	where $\Psi$ is a two-component spinor field, $\overline{\Psi} = \Psi^\dagger \gamma^0$ is its Dirac conjugate, and we use the following representation of the 2D Dirac algebra:
	\begin{equation*}
		\gamma^0
		=
		\begin{pmatrix}
			 0 & 1 \\
			-1 & 0
		\end{pmatrix}
		\qquad \text{and} \qquad
		\gamma^1
		=
		\begin{pmatrix}
			0 & 1 \\
			1 & 0
		\end{pmatrix}.
	\end{equation*}
	With $(t,x)$ the standard coordinates on $\mathbb{R}^{1,1}$, the Dirac operator is then:
	\begin{equation*}
		i \gamma^\mu \partial_\mu
		=
		\begin{pmatrix}
			0 & i(\partial_t + \partial_x) \\
			- i(\partial_t - \partial_x) & 0
		\end{pmatrix}.
	\end{equation*}
	Denoting components of $\Psi$ by $\Psi^t = \left(\psi, \widetilde{\psi}\right)$,
% 	\begin{equation*}
% 		\Psi
% 		=
% 		\begin{pmatrix}
% 			\psi \\
% 			\widetilde{\psi} \quad 
% 		\end{pmatrix},
% 	\end{equation*}
	the Lagrangian may be written as
	\begin{equation*}
		\mathcal{L} = -i \left[ \psi^*(\partial_t - \partial_x) \psi + \widetilde{\psi}^* (\partial_t + \partial_x) \widetilde{\psi} \right].
	\end{equation*}
	The classical equations of motion are then simply first-order wave equations,
	\begin{equation*}
		(\partial_t - \partial_x) \psi = 0
		\qquad \text{and} \qquad
		(\partial_t + \partial_x) \widetilde{\psi} = 0,
	\end{equation*}
	so that $\psi$ is a function $\psi(t,x) = f_1(t+x)$ of the combination $t+x$ only, and $\widetilde{\psi}$ is a function $\widetilde{\psi}(t,x) = f_2(t-x)$ of $t-x$ only.
	
% 	Hence the component $\psi$ is left-moving, and component $\widetilde{\psi}$ is right-moving.
	Because there is no coupling between $\psi$ and $\widetilde{\psi}$ in the Lagrangian, the theory splits into two independent halves: one describing the left-moving degree of freedom $\psi$, and the other describing the right-moving degree of freedom $\widetilde{\psi}$.
	This split into left- and right-moving halves means that the free fermion theory is chiral.
	
	The left-moving component $\psi$ is constant on curves of constant $t+x$ (the characteristics of the equation of motion).
	In $\mathbb{R}^{1,1}$, these are the \emph{left-moving null} curves;
% 	Such lines are left-moving null curves.
	that $\psi$ is constant on them means that, in a particularly trivial way, $\psi$ \enquote{propagates signals} along the left-moving null curves.
% 	constancy of $\psi$ is a particularly trivial form of \enquote{signal propagation} along these curves.
	
	On the other hand, the equations of motion specify no influence of the value $\psi(t,x) = f_1(t+x)$ on $\psi(t',x') = f_1(t'+x')$ for points $(t,x)$ and $(t',x')$ on distinct such curves $t+x \neq t'+x'$: the function $f_1$ is arbitrary (determined only by initial conditions).
	Hence there is no \enquote{signal propagation} between distinct left-moving null curves in the left-moving half of the free fermion theory.
	
	Similar observations hold for the right-moving half of the theory: the component $\widetilde{\psi}$ is constant on curves of constant $t-x$, which are the \emph{right-moving null} curves of $\mathbb{R}^{1,1}$, and there is no \enquote{signal propagation} between distinct such curves.
% 	
% 	
% 	\todo[inline]{1st order wave equations, so we can definitely write solutions as $\psi(t+x)$ and $\tilde{\psi}(t-x)$}
% 	The component $\psi$ is left-moving: since it is independent of the coordinate combination $t - x$, we may specify it as a function $\psi(t+x)$ of the combination $t+x$ alone.
% 	Similarly, the component $\widetilde{\psi}$ may be specified as a function only of $t-x$, and hence is right-moving.
% 	
% 	\todo[inline]{I should refresh my memory about characteristic curves of PDEs, and relate this to the `signal propagation' language I've been using...}
% 	
% 	Because there is no coupling between $\psi$ and $\widetilde{\psi}$ in the Lagrangian, the theory splits into two independent halves: one describing the left-moving degree of freedom $\psi$, and the other describing the right-moving degree of freedom $\widetilde{\psi}$.
\end{example}

% \begin{remark}
% 	\todo[inline]{Move to after I've appropriately defined left/right-moving null curves?}
% 	The left-moving component $\psi$ of \cref{ex:free_fermion} is constant on curves of constant $x+t$ (the characteristics of the equation of motion for $\psi$).
% 	In $\mathbb{R}^{1,1}$, these are the left-moving null curves;
% % 	Such lines are left-moving null curves.
% 	constancy of $\psi$ is a particularly trivial form of \enquote{signal propagation} along these curves.
% 	
% 	Note that for points $(t,x)$ and $(t',x')$ on distinct such curves ($x+t \neq x'+t'$) the equations of motion specify no influence of the value $\psi(t,x) = f_1(x+t)$ on $\psi(t',x') = f_1(x'+t')$: the function $f_1$ is arbitrary (determined only by initial conditions).
% 	Hence there is no \enquote{signal propagation} between distinct curves of constant $x+t$ in the left-moving half of the free fermion theory.
% 	
% 	Similar observations hold for the right-moving half of the theory, with component $\widetilde{\psi}$ and curves of constant $x-t$.
% \end{remark}

\begin{remark}
	Chiral CFTs are commonly studied in Euclidean rather than Lorentzian signature.
	Roughly, the Euclidean signature version of the preceding example replaces time $t$ with an imaginary counterpart $i\tau$.
	Denote $z := x + i\tau$ and $\bar{z} := x - i\tau$.
	Then in place of left- and right-moving halves in the Lorentzian signature of \cref{ex:free_fermion}, one finds holomorphic and anti-holomorphic halves in the Euclidean signature: the equations of motion reduce to $\partial_{\bar{z}} \psi = 0$ and $\partial_z \widetilde{\psi} = 0$ respectively.
	
	Generally, Euclidean CFTs which split into holomorphic and anti-holomorphic parts
	(free bosons,
	$bc$-ghost systems, % GSW does this
	WZW models, etc.) % see e.g. arXiv:2205.11190
	have Lorentzian counterparts which split analogously into left- and right-moving parts.
	In this work, we restrict our attention to Lorentzian signature.
\end{remark}

As illustrated by the free fermion of \cref{ex:free_fermion}, a chiral CFT splits into a left-moving half and a right-moving half.
Signals propagate only along left-moving null curves in the left-moving half, and dually, only along right-moving null curves in the right-moving half.
This is in contrast to standard relativistic QFTs as defined on $\mathsf{Loc}$ or CFTs as defined on $\mathsf{CLoc}$\footnote{See \cref{rm:CLoc}.}, wherein signals may propagate from a point to anywhere in its causal future.

To formulate a description of chiral CFTs on general 2D spacetimes, as per locally covariant AQFT, we begin with an appropriate category of 2D spacetimes and structure-preserving maps.
Since we are to describe conformal field theories, the maps are chosen to preserve only conformal structure on spacetimes rather than the full Lorentzian metric:
\begin{definition}
	Let $\mathsf{CSpTm}_{1+1}^{\mathrm{o,to}}$ be the category whose objects are 2D oriented spacetimes, and whose morphisms are orientation- and time-orientation-preserving smooth conformal maps, composing as functions.
\end{definition}
We restrict to oriented 2D spacetimes in order to facilitate a consistent definition of left-moving and right-moving null curves.
As in $\mathsf{SpTm}_{d+1}$, the morphisms of $\mathsf{CSpTm}_{1+1}^\mathrm{o,to}$ are necessarily local diffeomorphisms since they are smooth immersions between manifolds of the same dimension; see also \cref{rm:sptm-mor-local-diffeo}.
% We have restricted to oriented spacetimes and orientation- and time-orientation-preserving maps in $\mathsf{CSpTm}_{1+1}^\mathrm{o,to}$ to facilitate consistent definition of left-moving and right-moving}null curves.

For each half (left- or right-moving) of a chiral theory, we equip $\mathsf{CSpTm}_{1+1}^{\mathrm{o,to}}$ with a disjointness relation describing the relevant limitation on signal propagation.
First, we recall coordinate-independent definitions of left- and right-moving null curves on oriented 2D spacetimes.

% ---------

Take $(M,g)$ a 2D oriented spacetime.
Let $\Omega$ be a smooth nowhere-vanishing 2-form on $M$ representing the orientation.
Similarly, let $T$ be a smooth everywhere-timelike vector field on $M$ representing the time-orientation.
At any point $p \in M$, any non-zero, non-spacelike tangent vector $v \in T_p M$ lies in one of two familiar classes: either $v$ is future-directed in case $g_p(T,v) < 0$, or $v$ is past-directed in case $g_p(T,v) > 0$.
Moreover, for two different timelike vector fields $T$ and $T'$ representing the same time-orientation on $M$ so $g(T,T') < 0$, the signs of $g_p(T, v)$ and $g_p(T', v)$ coincide for non-spacelike $v$.
Thus the future and past lightcones are determined by the time-orientation on $M$ rather than merely its representative $T$.

Similarly, when $M$ is two-dimensional and oriented, any non-zero non-timelike tangent vector $v \in T_p M$ lies in one of two classes: either $v$ is \emph{right-pointing} in case $\Omega_p(T,v) > 0$, or $v$ is \emph{left-pointing} in case $\Omega_p(T,v) <0$.
These classes are independent of the representatives $\Omega$ and $T$ of the orientation and time-orientation on $M$:
\begin{proposition}
	Let $M$ be a 2D oriented spacetime, and $v \in T_p M$ a non-zero non-timelike tangent vector at $p \in M$.
	Say $\Omega$ and $\Omega'$ are both nowhere-vanishing 2-forms representing the orientation on $M$, and say $T$ and $T'$ are both everywhere-timelike vector fields representing the time-orientation on $M$.
	Then $\Omega_p(T,v)$ and $\Omega'_p(T',v)$ have the same sign.
\end{proposition}
\begin{proof}
	Since $\Omega$ and $\Omega'$ represent the same orientation, there is some $\omega \in C^\infty(M)$ such that $\Omega' = e^{\omega} \Omega$.
	So, we need only show that $\Omega_p(T,v)$ and $\Omega_p(T',v)$ have the same sign.	
	For the sake of contradiction, and without loss of generality, assume
	\begin{equation*}
		\lambda := \Omega_p(T,v) > 0
		\qquad \text{and} \qquad
		\lambda' := \Omega_p(T',v) < 0.
	\end{equation*}
	Since the (interior of the) future lightcone is a convex cone,
	\begin{equation*}
		T'' := \frac{1}{\lambda} T + \left(-\frac{1}{\lambda'}\right) T'
	\end{equation*}
	is also a future-directed timelike vector field.
	But $\Omega_p \left(T'', v\right) = 0$, which implies that $v$ and $T''$ are collinear since $\Omega_p$ is a non-zero skew-symmetric bilinear form on two-dimensional vector space $T_p M$.
	This is a contradiction because $T''$ is timelike and $v$ is non-timelike.
\end{proof}

We shall call a curve $\gamma : I \to M$ \emph{right-chiral} if it is smooth and (its restriction to the interior of $I$ is) causal with $\dot{\gamma}_t$ either everywhere null, future-directed and right-pointing, or everywhere null, past-directed and left-pointing.
Similarly, $\gamma$ is \emph{left-chiral} if it is smooth and causal with $\dot{\gamma}_t$ either everywhere null, future-directed and left-pointing, or everywhere null, past-directed and right-pointing.
A future-directed right- or left-chiral curve is called \emph{right-} or \emph{left-moving} respectively.
Since a chiral curve $\gamma : I \to M$ has non-zero derivative $\dot{\gamma}_t$ everywhere, it is either future-directed everywhere or past-directed everywhere.
In particular, any right-chiral $\gamma : [0,1] \to M$ is either right-moving or has $\gamma' : [0,1] \to M$ defined by $\gamma'(t) = \gamma(1-t)$ right-moving; similar for left-chiral curves.

% ---------
% CHIRAL_EXAMPLE_EDITS
% ---------
% Inserted for example:

% \todo[inline,color=purple]{Move some parts of $\mathbb{R}^{1,1}$ example up to here: get as far as lightcone coordinates, and that (regular) curves of constant $x^\pm$ are left/right-chiral curves.
% 
% Then can do an example of a chiral CFT: just use basic free boson CFT, and show that equation of motion breaks down into left/right-chiral halves constant in $x^\pm$, so constant on left/right-chiral curves.
% 
% Remark 1: since EOM give no more than this, in particular they do not relate (pass signals between) different left/right-chiral curves.
% 
% Remark 2: this is the Lorentzian signature version of the separation of Euclidean chiral CFTs into holomorphic and anti-holomorphism halves (corresponding to left- and right-chiral halves). In particular, an analogous story holds for the Lorentzian analogue of any Euclidean CFT that separates into holomorphic and anti-holomorphic halves.
% }

\begin{example} \label{ex:Minkowski_lightcone_coordinates}
	Consider 2D Minkowski spacetime $\mathbb{R}^{1,1}$ with metric $g = -\dd t^2 + \dd x^2$ in its standard coordinates $(t,x)$.
	The standard time-orientation can be represented by future-directed timelike vector field $\partial_t$, and the standard orientation by 2-form $\dd t \wedge \dd x$.
	
	Define lightcone coordinates
% 	$x^\pm := t \pm x$.
	$x^+ := t+x$ and $x^- := t-x$.
	The coordinate vector field $\partial_+ = \frac{1}{2} \left(\partial_t + \partial_x\right)$ is null, future-directed and right-pointing;
	similarly, $\partial_- = \frac{1}{2} \left(\partial_t - \partial_x\right)$ is null, future-directed and left-pointing.
	
	Consequently, any regular curve in $\mathbb{R}^{1,1}$ at constant $x^- = t - x$ is right-chiral, and any regular curve at constant $x^+ = t + x$ is left-chiral.
	With appropriate parameterisation, these are the right- and left-moving null curves of \cref{ex:free_fermion}.
\end{example}

% ---------

Analogous to the causal relation $J_M$ on an arbitrary spacetime, we define the \emph{right-chiral relation} on any 2D oriented spacetime by:
\begin{equation*}
	\chi^+_M := \Set{(p,q) \in M \times M | \begin{array}{c}
		p = q, \text{ or there exists a right-moving null curve} \\
		\gamma : [0,1] \to M \text{ with } \gamma(0)=p \text{ and } \gamma(1)=q.
	\end{array} }.
\end{equation*}
By definition, $\chi^+$ is reflexive.
Since we have restricted our definition of chiral curves to be smooth for convenience elsewhere, it is not immediately clear that $\chi^+$ is transitive -- generally the concatenation of smooth curves is not smooth.
After developing the necessary tools, we show in \cref{sec:spacetimes-cats-chiral:frames} that indeed $\chi^+$ is transitive.
Its symmetric closure is
\begin{equation*}
	s\chi^+_M = \Set{(p,q) \in M \times M | \begin{array}{c}
		p = q, \text{ or there exists a right-chiral curve} \\
		\gamma : [0,1] \to M \text{ with } \gamma(0)=p \text{ and } \gamma(1)=q.
	\end{array} }.
\end{equation*}
The \emph{left-chiral relation} $\chi^-_M$ on $M$ is defined dually, with right-moving curves replaced by left-moving curves.
All definitions and statements below involving the right-chiral relation $\chi^+$ have duals with the left-chiral relation $\chi^-$ in its place.

For 2D oriented spacetimes $M$ and $N$, any conformal map $f : M \to N$ sends null curves to null curves; if, moreover, $f$ preserves time-orientation and orientation, then $f$ sends right-moving null curves to right-moving null curves and left-moving null curves to left-moving null curves.
Consequently, conformal maps $f$ which preserve orientation and time-orientation necessarily preserve both chiral relations $\chi^+$ and $\chi^-$.

For open subsets $U$ and $V$ in spacetime $M$, the right-chiral relation $\chi^+_M$ describes when $U$ and $V$ cannot be connected by right-chiral curves in $M$ as follows:
\begin{proposition} \label{prop:chiral_discon_equiv}
	For open subsets $U$ and $V$ of 2D oriented spacetime $M$, the following are equivalent:
	\begin{enumerate}[label=\textup{(\roman*)}]
		\item
		$s\chi^+_M$ does not intersect $U \times V$,
		
		\item
		$\chi^+_M$ intersects neither $U \times V$ nor $V \times U$,
		
		\item
		There exists no right-chiral curve in $M$ connecting $U$ and $V$.
	\end{enumerate}
\end{proposition}
\begin{proof}
	Exactly as \cref{prop:caus_discon_equiv}.
\end{proof}

As done with $\mathsf{SpTm}_{d+1}$ and causal relations $J$, we may produce a $\disjoint$-relation on $\mathsf{CSpTm}_{1+1}^\mathrm{o,to}$ using right-chiral relations $\chi^+$:

\begin{definition} \label{def:chiral_discon}
	Define a $\disjoint$-relation $\disjoint_{\chi^+}$ on $\mathsf{CSpTm}_{1+1}^{\mathrm{o,to}}$ as follows: for any conterminous pair $f_1: M_1 \rightarrow N \leftarrow M_2 : f_2$, say that $f_1 \disjoint_{\chi^+} f_2$ if
	$s\chi^+_N$ does not intersect $f_1(M_1) \times f_2(M_2)$.
% 	$s\chi^+_N \cap [f_1(M_1) \times f_2(M_2)] = \varnothing$.
	
	We call $\disjoint_{\chi^+}$ the \emph{right-chiral disjointness relation} on $\mathsf{CSpTm}_{1+1}^\mathrm{o,to}$.
\end{definition}
Equivalently, $f_1 \disjoint_{\chi^+} f_2$ if there is no right-chiral curve in $N$ connecting $f_1(M_1)$ and $f_2(M_2)$.
Since all morphisms of $\mathsf{CSpTm}_{1+1}^{\mathrm{o,to}}$ are conformal and preserve orientations and time-orientations, they preserve $\chi^+$ and hence also $s\chi^+$.
Then $\disjoint_{\chi^+}$ is easily seen to be the pullback of relation $\disjoint_\mathrm{bin}$ of \cref{ex:binrel_disjointness} along the functor $\mathsf{CSpTm}_{1+1}^{\mathrm{o,to}} \to \mathsf{sBin}$ which sends each spacetime $M$ to its underlying set equipped with binary relation $s\chi^+_M$.

A similar \emph{left-chiral disjointness relation} $\disjoint_{\chi^-}$ on $\mathsf{CSpTm}_{1+1}^\mathrm{o,to}$ may be defined using the left-chiral relation $\chi^-$.
In the context of chiral disjointness relations $\disjoint_{\chi^+}$ or $\disjoint_{\chi^-}$, the relevant spacetimes are always two-dimensional and oriented so that $\chi^\pm$ can be defined, and the relevant maps between spacetimes always preserve both orientation and time-orientation so that they preserve $s\chi^\pm$.
Consequently, we may unambiguously omit the super- and subscripts from $\mathsf{CSpTm}_{1+1}^\mathrm{o,to}$ and write only $\mathsf{CSpTm}$ when discussing chiral disjointness.

At the level of the $\disjoint$-relation, $\left(\mathsf{SpTm}_{d+1}, \disjoint_J \right)$ and $\left(\mathsf{CSpTm}, \disjoint_{\chi^+}\right)$ are very similar. Indeed, \cref{thm:overlap-monics-SpTm,lem:unions_preserve_causal_disjointness,cor:unions_preserve_causal_disjointness} and their proofs all carry over verbatim, with $\mathsf{SpTm}_{d+1}$ and $J$ replaced by $\mathsf{CSpTm}$ and $\chi^+$ respectively:

\begin{theorem} \label{thm:overlap-monics-CSpTm}
	A morphism $f : M \to N$ in $\left(\mathsf{CSpTm}, \disjoint_{\chi^+}\right)$ is $\overlap$-monic if and only if it reflects $\overline{s\chi^+}$,
	i.e. for any $p,q \in M$:
	\begin{equation*}
		\left(f(p), f(q)\right) \in \overline{s\chi^+_N}
		\qquad \text{implies} \qquad
		(p,q) \in \overline{s\chi^+_M}.
	\end{equation*}
\end{theorem}

\begin{lemma} \label{lem:unions_preserve_chiral_disjointness}
	Let $M_1 \xrightarrow{f_1} N \xleftarrow{f_2} M_2$ be any conterminous pair in $\mathsf{CSpTm}$,
	and let $\left\{ i_\alpha : U_\alpha \hookrightarrow M_1 \right\}_{\alpha \in \mathcal{I}}$ be a family (not necessarily countable) of subspacetime inclusions in $\mathsf{CSpTm}$ such that $M_1 = \bigcup_{\alpha \in \mathcal{I}} U_\alpha$.
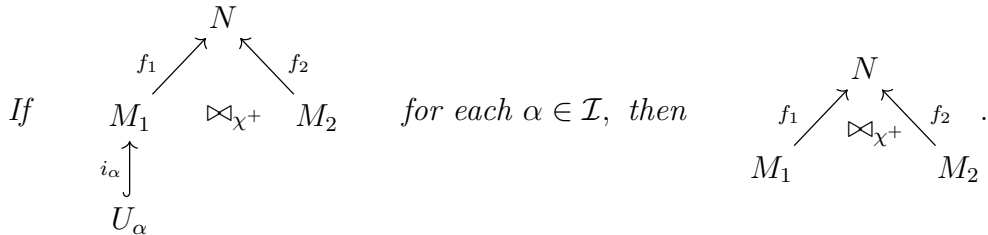
	\begin{equation*}
		\text{If} \qquad
		\begin{tikzcd}[cramped, column sep=small]
								& N \\
			M_1 \arrow[ur, "f_1"{name=1, near start}] \arrow[rr, phantom, "\disjoint_{\chi^+}"{pos=0.6}]	& & M_2 \arrow[ul, "f_2"{near start, name=2,swap}] \\
			U_\alpha \arrow[u, hook, "i_\alpha"]
		\end{tikzcd}
		\qquad \text{for each } \alpha \in \mathcal{I}, \text{ then} \qquad
		\begin{tikzcd}[cramped, column sep=small]
								& N \\
			M_1 \arrow[ur, "f_1"{name=1, near start}]	& & M_2 \arrow[ul, "f_2"{near start, name=2,swap}]
			\arrow[phantom, from=1, to=2, "\disjoint_{\chi^+}"{below,pos=0.6}]
		\end{tikzcd}.
	\end{equation*}
\end{lemma}

\begin{corollary} \label{cor:unions_preserve_chiral_disjointness}
	Let $\mathsf{C}$ be any $\disjoint$-subcategory of $(\mathsf{CSpTm}, \disjoint_{\chi^+})$. Then \cref{lem:unions_preserve_chiral_disjointness} applies with $\mathsf{C}$ in place of $\mathsf{CSpTm}$.
\end{corollary}

Just as for the case of $(\mathsf{SpTm}_{d+1}, \disjoint_J)$, the general characterisation of $\overlap$-monics in $\left(\mathsf{CSpTm}, \disjoint_{\chi^+}\right)$ is technically and conceptually difficult to handle.
To simplify it, we can again restrict to spacetimes with good `chiral properties' analogous to the causal properties of \cref{sec:spacetimes-cats-relativistic:causal_properties}.

\subsection{Chiral frames and chiral flows} \label{sec:spacetimes-cats-chiral:frames}

Before developing chiral analogues of the causal properties in \cref{sec:spacetimes-cats-relativistic:causal_properties}, we introduce some technical tools.

\begin{definition} \label{def:chiral_vectors}
	Let $(M,g)$ be a 2D oriented spacetime.
	
	A \emph{right-chiral vector field} $n_+$ on $M$ is a smooth, nowhere-vanishing vector field such that $n_+(p)$ is null, future-directed and right-pointing for all $p \in M$.
	Dually, a \emph{left-chiral vector field} $n_-$ on $M$ is a smooth, nowhere-vanishing vector field that is everywhere null, future-directed and left-pointing.
	
	A \emph{chiral frame} $(n_-, n_+)$ on $M$ is a smooth global frame on $M$, i.e. a frame for the tangent bundle $TM \to M$, where $n_-$ and $n_+$ are respectively left- and right-chiral vector fields.
	
	A \emph{right-chiral 1-form} $\eta^+$ on $M$ is a 1-form such that the vector field $(\eta^+)^\sharp$ is right-chiral.
	Similarly, a \emph{left-chiral 1-form} $\eta^-$ is one such that $(\eta^-)^\sharp$ is a left-chiral vector field.
	The musical isomorphism is with respect to the metric $g$ on $M$, so $(\eta^\pm)^\sharp$ is defined such that $\eta^\pm = g\left((\eta^\pm)^\sharp, \cdot \right)$.
	
	A \emph{chiral coframe} is a smooth global coframe $(-\eta^+, -\eta^-)$ on $M$, i.e. a frame for the cotangent bundle $T^*M \to M$, where $\eta^-$ and $\eta^+$ are respectively left- and right-chiral 1-forms.
\end{definition}

The right-chiral vector fields on $M$ form a convex cone in the $C^\infty(M)$-module $\Gamma(TM)$ of vector fields on $M$.
In other words, if $n_+, n_+' \in \Gamma(TM)$ are both right-chiral vector fields, and $f_1, f_2 \in C^\infty(M)$ are both smooth functions with $f_i(p) > 0$ for all $p \in M$, then $f_1 n_+ + f_2 n_+'$ is also a right-chiral vector field.
The same is true of left-chiral vector fields, and similar is true for right- and left-chiral 1-forms in $\Gamma(T^*M)$.

The arbitrary choice of signs and ordering in \cref{def:chiral_vectors} of a chiral frames $(n_-,n_+)$ and chiral coframes $(-\eta^+, -\eta^-)$ is arranged such that the following holds:

\begin{lemma}
	Let $(M, g)$ be a 2D oriented spacetime, and let $(E_1,E_2)$ be smooth global frame on $M$ with dual coframe $(\varepsilon^1, \varepsilon^2)$.
	Then $(E_1, E_2)$ is a chiral frame if and only if $(\varepsilon^1, \varepsilon^2)$ is a chiral coframe.
\end{lemma}
\begin{proof}
$(\Rightarrow)$:
	The dual coframe $(\varepsilon^1, \varepsilon^2)$ is defined such that $\varepsilon^i(E_j) = \delta^i_j$.
	Using this, it is straightforward to show that if $E_1$ and $E_2$ are respectively left- and right-chiral vector fields then
	\begin{equation*}
		(- \varepsilon^1)^\sharp = \frac{E_2}{-g(E_1,E_2)}
		\qquad \text{and} \qquad
		(- \varepsilon^2)^\sharp = \frac{E_1}{-g(E_1,E_2)}.
	\end{equation*}
	This shows that $(-\varepsilon^1)^\sharp$ and $(-\varepsilon^2)^\sharp$ are respectively right- and left-chiral vector fields,
	since $E_1$ and $E_2$ are both future-directed causal (and not collinear) so that $g(E_1,E_2) < 0$.
	Hence $(\varepsilon^1, \varepsilon^2)$ is a chiral coframe.
	
$(\Leftarrow)$:
	Similarly, if $(\varepsilon^1, \varepsilon^2)$ is a chiral coframe then its dual frame may be written as
	\begin{equation*}
		E_1 = \frac{(-\varepsilon^2)^\sharp}{- g \left((-\varepsilon^1)^\sharp, (-\varepsilon^2)^\sharp\right)}
		\qquad \text{and} \qquad
		E_2 = \frac{(-\varepsilon^1)^\sharp}{- g \left((-\varepsilon^1)^\sharp, (-\varepsilon^2)^\sharp\right)}.
	\end{equation*}
	So $E_1$ and $E_2$ are respectively left- and right-chiral, using $g \left((-\varepsilon^1)^\sharp, (-\varepsilon^2)^\sharp\right) < 0$.
	Hence $(E_1, E_2)$ is a chiral frame.
\end{proof}

% ---------
% CHIRAL_EXAMPLE_EDITS
% ---------
% % original:
% 
% 
% \begin{example} \label{ex:Minkowski_lightcone_coordinates}
% 	Consider 2D Minkowski spacetime $\mathbb{R}^{1,1}$ with metric $g = -\dd t^2 + \dd x^2$ in its standard coordinates $(t,x)$.
% 	The standard time-orientation can be represented by future-directed timelike vector field $\partial_t$, and the standard orientation by 2-form $\dd t \wedge \dd x$.
% 	
% 	Lightcone coordinates $x^\pm := t \pm x$ give coordinate vector fields
% 	$\partial_+ = \frac{1}{2} \left(\partial_t + \partial_x\right)$ a right-chiral vector field and $\partial_- = \frac{1}{2} \left(\partial_t - \partial_x\right)$ a left-chiral vector field.
% 	Hence $(\partial_-, \partial_+)$ is a chiral frame for $\mathbb{R}^{1,1}$.
% 	
% 	In terms of lightcone coordinates, the metric is $g = -\frac{1}{2} (\dd x^- \otimes \dd x^+ + \dd x^+ \otimes \dd x^-)$.
% 	Using this, it is easy to show that the dual coframe $(\dd x^-, \dd x^+)$ to $(\partial_-, \partial_+)$ has  $(\dd x^\pm)^\sharp = -2 \partial_\mp$ so that $- \dd x^+$ is a left-chiral 1-form and $- \dd x^-$ is a right-chiral 1-form.
% \end{example}
% 
% 
% ---------
% % with example:

% \todo[color=purple]{Trim down, since I'm moving initial intro lightcone coordinates on $\mathbb{R}^{1,1}$ higher. Also check referencing still works right.}

\begin{example} \label{ex:Minkowski_chiral_frame}
	Consider 2D Minkowski spacetime $\mathbb{R}^{1,1}$, with lightcone coordinates $(x^-,x^+)$ as per \cref{ex:Minkowski_lightcone_coordinates}.
	Since the coordinate vector field $\partial_+$ is everywhere null, future-directed and right-pointing, it is a right-chiral vector field on $\mathbb{R}^{1,1}$.
	Similarly, the coordinate vector field $\partial_-$ is a left-chiral vector field, so that $(\partial_-, \partial_+)$ is a chiral frame on $\mathbb{R}^{1,1}$.
	
	Written in lightcone coordinates, the metric is $g = -\frac{1}{2} (\dd x^- \otimes \dd x^+ + \dd x^+ \otimes \dd x^-)$.
	Using this, it is easy to show that the dual coframe $(\dd x^-, \dd x^+)$ to $(\partial_-, \partial_+)$ has  $(\dd x^\pm)^\sharp = -2 \partial_\mp$ so that $- \dd x^+$ is a left-chiral 1-form and $- \dd x^-$ is a right-chiral 1-form.
	Thus $(\dd x^-, \dd x^+)$ is a chiral coframe on $\mathbb{R}^{1,1}$ as per \cref{def:chiral_vectors}.
\end{example}

% ---------

\begin{proposition} \label{prop:chiral_frames_exist}
	Let $M$ be any 2D oriented spacetime.
	There exists a chiral frame $(n_-, n_+)$ on $M$.
\end{proposition}
\begin{proof}
	Locally-defined chiral vector fields may be constructed directly in coordinate charts.
	Let $U \subseteq M$ be an open subset with chart $\phi : U \to \mathbb{R}^2$, $p \mapsto (t(p),x(p))$ that is compatible with orientation and with time-orientation.
	Compatibility of the chart with orientation means that $\dd t \wedge \dd x$ represents the orientation on $U \subseteq M$, while compatibility with time-orientation means that the coordinate vector field $\partial_t$ is future-directed timelike.
	Then by direct computation, the smooth local vector fields
	\begin{equation*}
		n_\pm^U := \left[ \sqrt{- \det g} \pm g_{tx}\right] \partial_t \mp g_{tt} \partial_x
	\end{equation*}
	are null and future-directed, with $n_+^U$ right-pointing and $n_-^U$ left-pointing.
	Here $g_{ij} = g(\partial_i, \partial_j)$ are the components of the metric and $\det g$ is the determinant of the matrix of metric components.
	
	Given an atlas $\{ U_\alpha, \phi_\alpha : U_\alpha \to \mathbb{R}^2 \}$ of such charts, take a smooth partition of unity $\{ f_\alpha : M \to [0,1] \}$ subordinate to it.
	Set
	\begin{equation*}
		n_\pm := \sum_\alpha f_\alpha n_\pm^{U_\alpha}.
	\end{equation*}
	Since the collection of right-chiral vector fields is a convex cone in $\Gamma(TM)$, it follows that $n_+$ is right-chiral.
	Similarly $n_-$ is left-chiral.
\end{proof}

\begin{remark}
	Since \cref{prop:chiral_frames_exist} demonstrates the existence of a smooth global frame on any 2D oriented spacetime, it shows in particular that every 2D oriented spacetime is parallelisable.
	
	Alternatively, we see that any 2D oriented spacetime $(M,g)$ is parallelisable because $(T, \Omega(T, \cdot)^\sharp)$ is a smooth global frame, where $T$ is a future-directed timelike vector field representing the time-orientation on $M$, $\Omega$ is a nowhere-vanishing 2-form representing the orientation on $M$, and the musical isomorphism is taken with respect to metric $g$.
\end{remark}

Chiral frames on a two-dimensional Lorentzian manifold encode essentially the same information as orientations and time-orientations:

\begin{proposition} \label{prop:chiral_frame_orientation}
	Let $(M, g)$ be any Lorentzian manifold with $\dim M = 2$,
	and let $(n_1, n_2)$ be a smooth global frame on $M$ where the $n_i$ are null $g(n_i,n_i) = 0$ and have $g(n_1, n_2) < 0$ everywhere.
	
	Then there is a unique choice of orientation and time-orientation on $M$ such that $(n_1, n_2)$ is a chiral frame.
\end{proposition}

\begin{proof}
	From the null frame $(n_1, n_2)$ and its dual coframe $(\epsilon_1, \epsilon_2)$,
	construct vector field $T := n_1 + n_2$ and 2-form $\Omega := \epsilon_1 \wedge \epsilon_2$ on $M$.
	$T$ is everywhere timelike since
	\begin{equation*}
		g(T, T) = g(n_1 + n_2, n_1 + n_2) = 2 g(n_1, n_2) < 0,
	\end{equation*}
	and $\Omega$ is nowhere-vanishing by linear independence of $(n_1, n_2)$.
	
	The time-orientation in which $T$ is future-directed is the unique time-orientation on $M$ for which $n_1$ and $n_2$ are everywhere future-directed, since
	\begin{equation*}
		g(T, n_i) = g(n_1, n_2) < 0.
	\end{equation*}
	With this time-orientation fixed, the orientation on $M$ specified by volume form $\Omega$ is the unique orientation for which $n_1$ is left-pointing and $n_2$ is right-pointing:
% 	\begin{equation*}
		$\Omega(T,n_1) = \epsilon_1 \wedge \epsilon_2 (n_2, n_1) = -1 < 0$, and
% 		\quad \text{and} \quad
		$\Omega(T,n_2) = \epsilon_1 \wedge \epsilon_2 (n_1, n_2) = 1 > 0$.
% 	\end{equation*}
\end{proof}

\begin{example}
	Lightcone coordinates $x^\pm = t \pm x$ on two-dimensional Minkowski spacetime $\mathbb{R}^{1,1}$ have coordinate frame $(\partial_-, \partial_+)$ with both $\partial_+$ and $\partial_-$ everywhere null; see also
% ---------
% CHIRAL_EXAMPLE_EDITS
% ---------
% % original:
% 	\cref{ex:Minkowski_lightcone_coordinates}.
% ---------
% % with example:
	\cref{ex:Minkowski_lightcone_coordinates,ex:Minkowski_chiral_frame}.
% ---------
	
	From the dual coframe $(\dd x^-, \dd x^+)$, the orientation on $\mathbb{R}^{1,1}$ specified as per the proposition above is given by 2-form $ \dd x^- \wedge \dd x^+ = (\dd t - \dd x) \wedge (\dd t + \dd x) = 2 \dd t \wedge \dd x$.
	This coincides with the standard orientation on $\mathbb{R}^{1,1}$ given by $\dd t \wedge \dd x$.
	
	Similarly, the time-orientation specified above is given by timelike vector field $\partial_- + \partial_+ = \frac{1}{2} \left(\partial_t - \partial_x \right) + \frac{1}{2} \left(\partial_t + \partial_x \right) = \partial_t$; this coincides with the standard time-orientation on $\mathbb{R}^{1,1}$.
\end{example}

Any null frame $(n_1, n_2)$ has, on any given path-component, either $g(n_1, n_2) < 0$ or $g(n_1, n_2) > 0$ everywhere by linear independence.
So, up to swapping $n_1 \mapsto -n_1$ on some path-components of $M$, the hypothesis $g(n_1, n_2) < 0$ of \cref{prop:chiral_frame_orientation} may be satisfied by any null frame:
\begin{corollary}
	Let $(M, g)$ be any Lorentzian manifold with $\dim M = 2$,
	and let $(n_1, n_2)$ be a smooth global frame on $M$ where the $n_i$ are null everywhere.
	Then $(M,g)$ is orientable and time-orientable.
\end{corollary}

Any (right- or left-) chiral vector field is nowhere-vanishing, and remains a chiral vector field after scaling by an everywhere-strictly-positive function.
Consequently, we may always rescale a chiral vector field to a complete chiral vector field, i.e. such that the domain of any maximal integral curve is the whole of $\mathbb{R}$ \cite{mathsx378789}\todo{can I find earlier or textbook citation?}:

\begin{lemma} \label{lem:vector_field_rescale_to_complete}
	Let $X$ be a smooth manifold, and $V$ be a smooth, nowhere-vanishing vector field on $X$.
	There exists a smooth function $f \in C^\infty(X)$ which is everywhere-strictly-positive $f > 0$ such that $fV$ is a complete vector field.
\end{lemma}
\begin{proof}
	Choose any complete Riemannian metric $h$ on $X$ -- existence of such is proved in \cite{NomizuOzeki1961}.
	Set
	\begin{equation*}
		f = \frac{1}{\sqrt{h(V,V)}},
	\end{equation*}
	so that $fV$ is normalised to unit $h$-length everywhere; this is possible since $V$ is nowhere-vanishing.
	
	Let $\gamma : (a,b) \to M$ be the maximal integral curve of $fV$ from $\gamma(0) = p \in X$.
	By the choice of normalisation of $fV$, $\gamma$ is parameterised by $h$-arc-length.
	Assume for the sake of contradiction that $b < \infty$; then by the Escape Lemma \parencite[Lemma 9.19]{Lee2013}, $\gamma([0,b))$ is not contained in any compact subset of $M$.
	In particular, $\gamma$ leaves the closed $h$-ball $B_h(p,r)$ of radius $r > b$ centred at $p$, which is compact by the Hopf-Rinow theorem since $h$ is complete.
	So there exists $t \in (0,b)$ with $\gamma(t) \not \in B_h(p,r)$;
	but then $\gamma : [0,t] \to M$ is a path from $p$ of $h$-length $t < r$ which ends outside the ball $B_h(p,r)$ -- a contradiction.
	
	A similar contradiction can be shown for $a > -\infty$, so $\gamma$ has domain $\mathbb{R}$.
\end{proof}

\begin{definition} \label{def:chiral_flow}
	A \emph{right-chiral flow}  $\theta^+ : \mathbb{R} \times M \to M$, $(\tau, p) \mapsto \theta^+_\tau(p)$ on a 2D oriented spacetime $M$ is the global flow of a complete right-chiral vector field $n_+$ on $M$.
	Similarly, a \emph{left-chiral flow} $\theta^- : \mathbb{R} \times M \to M$ is the global flow of a complete left-chiral vector field $n_-$.
\end{definition}
Chiral flows $\theta^\pm$ are actions of the Lie group $\mathbb{R}$ on the spacetime.
Since any 2D oriented spacetime has a chiral frame, and any chiral vector field may be rescaled to a complete chiral vector field, it follows that left- and right-chiral flows exist on any 2D oriented spacetime. Note that the chiral flows on a spacetime are not unique.

\begin{lemma} \label{lem:chiral_curves_integral}
	Let $M$ be a 2D oriented spacetime with right-chiral vector field $n_+$.
	Fix interval $I \subseteq \mathbb{R}$ and $t_0 \in I$.
	A smooth curve $\gamma : I \to M$ is right-chiral if and only if there is an interval $J \subseteq \mathbb{R}$ containing $0 \in J$ and a diffeomorphism $\phi : J \to I$ with $\phi(0) = t_0$ such that the reparameterised curve $\gamma \circ \phi|_{\operatorname{int}J}$ restricted to the interior $\operatorname{int}J$ of $J$ is an integral curve of $n_+$.
	The reparameterisation $\phi$ is increasing if $\gamma$ is future-directed and decreasing if $\gamma$ is past-directed.
	Moreover, $\gamma$ is inextendable if and only if $\gamma \circ \phi|_{\operatorname{int}J}$ is a maximal integral curve.
\end{lemma}
\begin{proof}
	By definition, any smooth $\gamma : I \to M$ is right-chiral if and only if there is some smooth, nowhere-zero $\alpha : \operatorname{int}I \to \mathbb{R}$ such that for all $t \in \operatorname{int}I$,
	\begin{equation} \label{eqn:chiral_curves_integral}
		\dot{\gamma}_t = \alpha(t) \cdot n_+|_{\gamma(t)}.
	\end{equation}
	$\gamma$ is future-directed and right-moving exactly when $\alpha$ is positive everywhere; dually, $\gamma$ is past-directed and left-moving exactly when $\alpha$ is negative everywhere.
	Say $\gamma \circ \phi|_{\operatorname{int}J}$ is an integral curve of $n_+$; then \cref{eqn:chiral_curves_integral} holds with
	\begin{equation*}
		\alpha(t) = \frac{1}{\dot{\phi}\left(\phi^{-1}(t)\right) }.
	\end{equation*}
	$\gamma$ is future-directed and right-moving when $\phi$ is increasing, and past-directed and left-moving when $\phi$ is decreasing.
	
	Conversely, say $\gamma$ is right-chiral so that \cref{eqn:chiral_curves_integral} holds.
	Define $\psi : I \to \mathbb{R}$ by
	\begin{equation*}
		\psi(t) = \int_{t_0}^t \alpha(t') \dd t'.
	\end{equation*}
	Since $\alpha$ is smooth and nowhere-zero, $\psi$ is smooth and strictly monotonic (increasing if $\alpha$ is everywhere positive, decreasing if $\alpha$ is everywhere negative)
	and so a diffeomorphism onto its image in $\mathbb{R}$.
	Take reparameterisation $\phi : J \to I$ to be the inverse of the codomain-restriction of $\psi$.
	Then for any $s \in \operatorname{int}J$,
	\begin{equation*}
		\dot{\left(\gamma \circ \phi\right)}_s
		=
		\dot{\phi}(s) \cdot \dot{\gamma}_{\phi(s)}
		=
		\frac{1}{\dot{\psi}(\phi(s))} \cdot \alpha(\phi(s)) \cdot n_+|_{\gamma \circ \phi (s)}
		=
		n_+|_{\gamma \circ \phi (s)},
	\end{equation*}
	so $\gamma \circ \phi|_{\operatorname{int}J}$ is an integral curve of $n_+$.

	It remains to show that the right-chiral curve $\gamma$ is inextendable if and only if the integral curve $\gamma \circ \phi|_{\operatorname{int}J}$ of $n_+$ is maximal.
	Denote $\operatorname{int}I = (i_1, i_2)$ and $\operatorname{int}J = (j_1, j_2)$, with some of these values possibly infinite.
	Then it holds that
	\begin{equation*}
		\lim_{t\to i_1^+} \gamma(t) = \lim_{\tau \to j_1^+} \gamma \circ \phi (\tau)
		\qquad \text{and} \qquad
		\lim_{t\to i_2^-} \gamma(t) = \lim_{\tau \to j_2^-} \gamma \circ \phi (\tau)
	\end{equation*}
	for $\phi$ increasing, and vice versa for $\phi$ decreasing; in particular, the limits of $\gamma$ exist exactly when the corresponding limits of $\gamma \circ \phi$ exist.
	So $\gamma$ is inextendable if and only if the limits $\lim_{\tau \to j_1^+} \gamma \circ \phi (\tau)$ and $\lim_{\tau \to j_2^-} \gamma \circ \phi (\tau)$ do not exist.

	Say $q := \lim_{\tau \to j_2^-} \gamma \circ \phi (\tau)$ exists,
	i.e. for any neighbourhood $U \subseteq M$ of $q$ there is $\tau_0 \in \operatorname{int}J$ such that $\gamma \circ \phi \left( [\tau_0, j_2) \right) \subseteq U$.
	Since manifold $M$ is locally compact, there is a compact neighbourhood $U$ of $q$.
	Then the Escape Lemma \parencite[Lemma 9.19]{Lee2013} gives that either the integral curve $\gamma \circ \phi|_{\operatorname{int}J}$ is not maximal or $j_2 = \infty$.
	Assume for the sake of contradiction that $j_2 = \infty$.
	Then for any $\tau' \in \mathbb{R}$ such that $(\tau' , q)$ is in the domain of the maximal flow $\theta^+$ of vector field $n_+$, it follows that
	\begin{equation*}
		\theta^+_{\tau'} (q)
		=
		\theta^+_{\tau'} \left(\lim_{\tau \to \infty} \gamma \circ \phi (\tau)\right)
		=
		\lim_{\tau \to \infty} \theta^+_{\tau'} \left(\gamma \circ \phi (\tau)\right)
		=
		\lim_{\tau \to \infty} \gamma \circ \phi (\tau + \tau')
		=
		q,
	\end{equation*}
	so $q$ is a fixed point of the flow.
	But then $n_+|_q = 0$, which contradicts that $n_+$ is right-chiral.
	We conclude that $\gamma \circ \phi|_{\operatorname{int}J}$ is not maximal.
	A similar argument shows that if $\lim_{\tau \to j_1^+} \gamma \circ \phi (\tau)$ exists then $\gamma \circ \phi|_{\operatorname{int}J}$ is not maximal.
	
	Conversely, say $\gamma \circ \phi|_{\operatorname{int}J}$ is not maximal.
	Then there exists  an open interval $K = (k_1,k_2)$ with $\operatorname{int}J \subset K$ strictly and an integral curve $\delta : K \to M$ of $n_+$ with $\delta|_{\operatorname{int}J} = \gamma \circ \phi|_{\operatorname{int}J}$.
	Because $\operatorname{int}J \subset K$ strictly, either $k_1 < j_1$ or $k_2 > j_2$.
	In the former case, $\lim_{\tau \to j_1^+} \gamma \circ \phi (\tau) = \lim_{\tau \to j_1^+} \delta(\tau) = \delta(j_1)$ exists, while in the latter $\lim_{\tau \to j_2^-} \gamma \circ \phi (\tau) = \lim_{\tau \to j_2^-} \delta(\tau) = \delta(j_2)$ exists.
\end{proof}

The lemma above gives directly that any right-chiral curve $\gamma : I \to M$ defined on an open interval $I$ is, up to reparameterisation, an integral curve of a right-chiral vector field $n_+$.
In case $\gamma$ is defined on a non-open interval, it gives that after reparameterisation, $\gamma$ extends to an integral curve of $n_+$.
For if $I$ is not open, then $\gamma : I \to M$ is not inextendable; consequently, the reparameterisation  $\gamma \circ \phi|_{\operatorname{int}J} : \operatorname{int} J \to M$ given by the lemma is not maximal and so has extension to an integral curve $\delta : K \to M$ of $n_+$ with $(j_1, j_2) := \operatorname{int} J \subset K$.
Say $j_2 \in J$; then so long as $j_2 \not \in K$, the extension $\delta$ constitutes a non-inextendable right-chiral curve in its own right since $\gamma \circ \phi(j_2) = \lim_{\tau \to j_2^-} \delta(\tau)$ exists.
Similar if $j_1 \in J$.
Consequently, we may always find an extension $\delta : K \to M$ of the integral curve $\gamma \circ \phi|_{\operatorname{int}J}$ such that $J \subset K$ entirely, rather than merely $\operatorname{int} J \subset K$.
It follows that $\delta|_J = \gamma \circ \phi$, so $\delta$ extends not merely the restriction $\gamma \circ \phi|_{\operatorname{int}J}$ but the whole reparameterisation $\gamma \circ \phi$ of $\gamma$.

For instance, the lemma gives that any right-chiral curve $\gamma : [0,1] \to M$ has reparameterisation diffeomorphism $\phi : [0,\tau] \to [0,1]$ with $\phi(0) = 0$ such that $\gamma \circ \phi|_{(0,\tau)}$ is an integral curve of $n_+$.
By the above argument, there is some open interval $K$ with $[0,\tau] \subseteq K$ and integral curve $\delta : K \to M$ of $n_+$ with $\delta |_{[0,\tau]} = \gamma \circ \phi$.

From this we immediately obtain characterisations of the chiral relation $\chi^+$ and its symmetric closure in terms of the integral curves of a chiral vector field $n_+$:

\begin{proposition} \label{prop:characteristion_chiral_relation}
	Let $M$ be a 2D oriented spacetime with right-chiral vector field $n_+$, and denote by $\theta^+$ the maximal flow of $n_+$.
	Let $p,q$ be points in $M$.
	The following are equivalent:
	\begin{enumerate}[label=\textup{(\roman*)}]
		\item
		$(p,q) \in \chi^+_M$,
		
		\item
		There is an integral curve $\gamma : I \to M$ of $n_+$ from $\gamma(0) = p$, with $\gamma(\tau) = q$ for some $\tau \in I$ with $\tau \geq 0$,
		
		\item
		There is $\tau \geq 0$ such that $q = \theta^+_\tau(p)$.
	\end{enumerate}
	
	Likewise, the following are equivalent:
	\begin{enumerate}[label=\textup{(\roman*)}]
		\item
		$(p,q) \in s\chi^+_M$,
		
		\item
		There is an integral curve $\gamma : I \to M$ of $n_+$ from $\gamma(0) = p$, with $\gamma(\tau) = q$ for some $\tau \in I$,
		
		\item
		There is $\tau \in \mathbb{R}$ such that $q = \theta^+_\tau(p)$.
	\end{enumerate}
\end{proposition}

Using the group law of flows, this gives:
\begin{corollary} \label{cor:chi_rel_transitive}
	Let $M$ be a 2D oriented spacetime.
	The right-chiral relation $\chi^+_M$ and its symmetric closure $s\chi^+_M$ are both transitive.
\end{corollary}
\begin{proof}
	Choose right-chiral vector field $n_+$ on $M$, as is always possible by \cref{prop:chiral_frames_exist}.
	Denote its maximal flow by $\theta^+$.
	If $(p,q) \in s\chi^+_M$ and $(q,r) \in s\chi^+_M$ then by the preceding proposition
% \cref{prop:characteristion_chiral_relation}
	there exist $\tau, \tau' \in \mathbb{R}$ such that $q = \theta^+_{\tau}(p)$ and $r = \theta^+_{\tau'}(q)$.
	Then by the group law of the flow,
	\begin{equation*}
		r = \theta^+_{\tau'}( \theta^+_{\tau}(p)) = \theta^+_{\tau + \tau'}(p),
	\end{equation*}
	so that $(p,r) \in s\chi^+_M$.
	Similar holds if $(p,q) \in \chi^+_M$ and $(q,r) \in \chi^+_M$, but with the restriction that $\tau, \tau' \geq 0$.
	Then $(p,r) \in \chi^+_M$ since $\tau + \tau' \geq 0$.
\end{proof}

By definition, the symmetric closure $s\chi^+$ of a right-chiral relation is symmetric and reflexive; the preceding proposition and corollary lead to:
\begin{corollary}\label{cor:sym_chi_equivalence_rel}
	Let $M$ be a 2D oriented spacetime.
	The relation $s\chi^+_M$ is an equivalence relation on $M$.
	
% 	A subset $S \subseteq M$ is an equivalence class of $s\chi^+_M$ if and only if $S$ is an orbit of a right-chiral flow $\theta^+ : \mathbb{R}\times M \to M$, if and only if $S = \gamma(I)$ is the image of an inextendable right-chiral curve $\gamma : I \to M$.
	
	For subset $S \subseteq M$, the following are equivalent:
	\begin{enumerate} [label=\textup{(\roman*)}]
		\item $S$ is an equivalence class of $s\chi^+_M$,
		
		\item $S$ is an orbit of a right-chiral flow $\theta^+ : \mathbb{R}\times M \to M$,
		
		\item $S$ is the image $\gamma(I)$ of an inextendable right-chiral curve $\gamma : I \to M$.
	\end{enumerate}
\end{corollary}

In addition to their use in describing chiral relations $\chi^+$, chiral (co)frames may be used to characterise conformal maps which preserve orientation and time-orientation.

\begin{lemma} \label{lem:frames_preserved_by_map}
	Let $X, Y$ be parallelisable smooth manifolds with $\dim X = \dim Y = n$ and smooth global frames $(E_i^X)$ and $(E_i^Y)$ on $X$ and $Y$ respectively.
	Let $(\varepsilon^i_X)$ and $(\varepsilon^i_Y)$ be the respective dual coframes.
	
	If $f : X \to Y$ is any smooth map and $f_i \in C^\infty(X)$ for $i \in \{1, \ldots , n\}$, the following are equivalent:
	\begin{enumerate}[label=\textup{(\roman*)}]
		\item \label{lem:frames_preserved_by_map:coframe}
		$f^* \varepsilon^i_Y = e^{f_i} \varepsilon^i_X$ for all $i \in \{1, \ldots , n\}$,
		
		\item \label{lem:frames_preserved_by_map:frame}
		$ df \circ E_i^X = e^{f_i} \cdot (E_i^Y \circ f)$ as sections of the pullback bundle $f^* TY \to X$ for all $i \in \{1, \ldots , n\}$.
	\end{enumerate}
\end{lemma}
\begin{proof}
	Fix some $i \in \{1,\ldots, n\}$.
	By definition $\varepsilon^i_X(E_j^X) = \delta^i_j = \varepsilon^i_Y(E_j^Y)$,
	from which it is straightforward to verify that \labelcref{lem:frames_preserved_by_map:coframe} implies that
	\begin{equation*}
		\left(\varepsilon^k_Y\right)_{f(p)} \left(df_p E_i^X|_p \right)
		=
		\left(f^* \varepsilon^k_Y\right)_p \left(E_i^X|_p \right)
		=
		\left(\varepsilon^k_Y\right)_{f(p)} \left( e^{f_i(p)} E_i^Y|_{f(p)} \right)
	\end{equation*}
	for all $k \in \{1, \ldots , n\}$ and $p \in X$, which gives \labelcref{lem:frames_preserved_by_map:frame}.
	Similarly, \labelcref{lem:frames_preserved_by_map:frame} implies that
	\begin{equation*}
		\left(f^* \varepsilon^i_Y\right)_p \left(E_j^X|_p \right)
		=
		\left(\varepsilon^i_Y\right)_{f(p)} \left(df_p E_j^X|_p \right)
		=
		\left(e^{f_i} \varepsilon^i_X\right)_p \left(E_j^X |_p\right)
	\end{equation*}
	for all $j \in \{1, \ldots , n\}$ and $p \in X$, which gives \labelcref{lem:frames_preserved_by_map:coframe}.
\end{proof}

\begin{proposition} \label{prop:chiral_frame_characterisation_conf_map}
	Let $M, N$ be 2D oriented spacetimes, and $(n_-^M, n_+^M)$, $(n_-^N, n_+^N)$ chiral frames on $M$ and $N$ respectively.
	Denote by $(-\eta^+_M,-\eta^-_M)$, $(-\eta^+_N,-\eta^-_N)$ the respective dual chiral coframes on $M$ and $N$.
	
	A smooth map $f : M \to N$ is conformal and orientation- and time-orientation-preserving if and only if there exist $F_+, F_- \in C^\infty(M)$ such that the following equivalent conditions hold:
	\begin{enumerate}[label=\textup{(\roman*)}]
		\item \label{prop:chiral_frame_characterisation_conf_map:coframe}
		$f^* \eta^\pm_N = e^{F_\mp} \eta^\pm_M$,
		
		\item \label{prop:chiral_frame_characterisation_conf_map:frame}
		$ df \circ n_\pm^M = e^{F_\pm} \cdot (n_\pm^N \circ f)$ as sections of the pullback bundle $f^* TN \to M$.
	\end{enumerate}
\end{proposition}
\begin{proof}
	Equivalence of \labelcref{prop:chiral_frame_characterisation_conf_map:coframe,prop:chiral_frame_characterisation_conf_map:frame} follows from the preceding lemma.
% 	\cref{lem:frames_preserved_by_map}.

$(\Rightarrow)$:
	Since $f$ is conformal and orientation- and time-orientation-preserving, it sends future-directed left/right-pointing null vectors to future-directed left/right-pointing null vectors; hence at each $p \in M$ there are unique $F_-(p), F_+(p) \in \mathbb{R}$ such that
	\begin{equation*}
		df_p n_\pm^M|_p = e^{F_\pm(p)} n_\pm^N |_{f(p)}.
	\end{equation*}
	Smoothness of the functions $F_\pm : M \to \mathbb{R}$ so defined follows by recognising $e^{F_\pm}$ as components of the smooth map
	\begin{equation*}
		M \xrightarrow{n_\pm^M} TM \xrightarrow{df} TN \xrightarrow{\psi_N} N \times \mathbb{R}^2,
	\end{equation*}
	where $\psi_N$ is the global trivialisation of the tangent bundle $TN \to N$ given by the frame $(n_-^N, n_+^N)$.
	It follows that \labelcref{prop:chiral_frame_characterisation_conf_map:frame} holds.

$(\Leftarrow)$:
	We show that \labelcref{prop:chiral_frame_characterisation_conf_map:coframe} implies $f$ is conformal and orientation- and time-orientation-preserving.
	Because $n_\pm^M$ are null,
	the metric $g_M$ on $M$ may be expressed in terms of the chiral coframe as
	\begin{equation*}
		g_M = - e^{G_M} \left(\eta^+_M \otimes \eta^-_M + \eta^-_M \otimes \eta^+_M\right),
	\end{equation*}
	where $G_M \in C^\infty(M)$ is defined by $e^{G_M} = - g(n_-^M, n_+^M)$ which is strictly positive since $n_+^M$ and $n_-^M$ lie in the same half of the lightcone.
	Similar is true for the metric $g_N$ on $N$.
	Then from \labelcref{prop:chiral_frame_characterisation_conf_map:coframe} it follows that
	\begin{align*}
		f^* g_N
		&=
		- e^{G_N \circ f} \left(f^*\eta^+_N \otimes f^*\eta^-_N + f^*\eta^-_N \otimes f^*\eta^+_N\right)
		\\ &=
		- e^{F_+ + F_- + G_N \circ f} \left(\eta^+_M \otimes \eta^-_M + \eta^-_M \otimes \eta^+_M\right)
		\\ &=
		e^{F_+ + F_- - G_M + G_N \circ f} \cdot g_M,
	\end{align*}
	so $f$ is conformal.
	
	As in \cref{prop:chiral_frame_orientation}, the orientation and time-orientation on $M$ may be represented by 2-form $\Omega_M  = \eta^+_M \wedge \eta^-_M$ and future-directed timelike vector field $T^M = n_-^M + n_+^M$ on $M$ respectively; similar for $N$.
	Then from \labelcref{prop:chiral_frame_characterisation_conf_map:coframe} it follows that $f$ preserves orientation since
	\begin{equation*}
		f^* \Omega_N
		=
		f^*( \eta^+_N \wedge \eta^-_N)
		=
		e^{F_+ + F_-} \cdot \eta^+_M \wedge \eta^-_M
		=
		e^{F_+ + F_-} \cdot \Omega_M.
	\end{equation*}
	Because $(-\eta^+_M, -\eta^-_M)$ is dual to $(n_-^M, n_+^M)$, we have $\eta^\pm_M(n_\pm^M) = 0$ and $\eta^\pm_M(n_\mp^M) = -1$, and hence $\eta^\pm_M(T^M) = \eta^\pm_M(n_-^M + n_+^M) = -1$; similar for $N$.
	Then at every $p \in M$, it follows from \labelcref{prop:chiral_frame_characterisation_conf_map:coframe} that
	\begin{align*}
		g_N|_{f(p)} \left(df_p(T_p^M) , T_{f(p)}^N\right)
		&=
		- e^{G_N \circ f}(p) \left[\eta^+_N \otimes \eta^-_N + \eta^-_N \otimes \eta^+_N\right]_{f(p)}  \left(df_p(T_p^M) , T_{f(p)}^N\right)
		\\ &=
		- e^{G_N \circ f}(p) \left[-(f^*\eta^+_N)_p(T_p^M) - (f^*\eta^-_N)_p(T_p^M) \right]
		\\ &=
		- e^{G_N \circ f}(p)\left[e^{F_-}(p) + e^{F_+}(p)\right]
		< 0,
	\end{align*}
	so $df_p(T_p^M)$ and $T_{f(p)}^N$ lie in the same half of the lightcone.
	Hence $f$ preserves time-orientation.
\end{proof}

\begin{example}
	Consider open subsets $U, V \subseteq \mathbb{R}^{1,1}$ of two-dimensional Minkowski spacetime and smooth map $f : U \to V$.
	It can be shown by direct computation that $f$ is conformal and orientation- and time-orientation-preserving if and only if, in lightcone coordinates $(x^-, x^+)$ as described in \cref{ex:Minkowski_lightcone_coordinates}, we have
	\begin{equation*}
		f(x^-, x^+) = (f_-(x^-), f_+(x^+))
	\end{equation*}
	for some smooth functions $f_\pm : \operatorname{pr}_\pm (U) \to \operatorname{pr}_\pm (V)$ which are strictly increasing, i.e. $f_\pm' > 0$ everywhere.
	Here $\operatorname{pr}_\pm : \mathbb{R}^{1,1} \to \mathbb{R}$ are projections onto the lightcone coordinates.
	
	Then with respect to chiral frames $(\partial_-, \partial_+)$ on $U$ and $V$, the smooth functions $F_\pm \in C^\infty(U)$ of \cref{prop:chiral_frame_characterisation_conf_map} above are given by $e^{F_\pm} = f_\pm' \circ \operatorname{pr}_\pm$ since
	\begin{equation*}
		df_{(x^-,x^+)} \partial_-|_{(x^-,x^+)}
		= \frac{\partial f_-}{\partial x^-} \partial_-|_{f(x^-,x^+)}  +  \frac{\partial f_-}{\partial x^+} \partial_+|_{f(x^-,x^+)}
		= (f_-' \circ \operatorname{pr}_-) \partial_-|_{f(x^-,x^+)},
	\end{equation*}
	and similarly 
		$df_{(x^-,x^+)} \partial_+|_{(x^-,x^+)}
		= (f_+' \circ \operatorname{pr}_+) \partial_+|_{f(x^-,x^+)}$.
\end{example}

\subsection{Spacetimes with good chiral properties}
\label{sec:spacetimes-cats-chiral:chiral-properties}

In analogy with the hierarchy of causal properties discussed in \cref{sec:spacetimes-cats-relativistic:causal_properties}, we now introduce a hierarchy of \emph{chiral properties} that a two-dimensional oriented spacetime may possess.
Each chiral property has a left-chiral and a right-chiral version.
We present the right-chiral properties below; their left-chiral duals are obtained by changing all right-chiral curves to left-chiral curves, i.e. swapping $\chi^+$ with $\chi^-$.
Upon restriction to spacetimes with particular chiral properties, the characterisation in \cref{thm:overlap-monics-CSpTm} of $\overlap$-monics will simplify.

Let $M$ be a 2D oriented spacetime.
We call the spacetime \emph{$\chi^+$-causal} if it contains no closed right-chiral curves, or equivalently if the right-chiral relation $\chi^+_M$ is antisymmetric.
This is a chiral analogue of causal spacetimes, wherein the causal relation $J$ is antisymmetric\footnote{It may be clarifying to call causal spacetimes $J$-causal, to compare with $\chi^+$-causal spacetimes.}.
\todo{should I give all chiral properties in definition environments?}

Trivially, if $M$ contains no closed causal curves then in particular it contains no closed right-chiral curves; hence if 2D oriented spacetime $M$ is causal then it is also $\chi^+$-causal.
The converse is not true; for instance, consider the spacetime given by the quotient of Minkowski $\mathbb{R}^{1,1}$ under identification $(x^-,x^+) \sim (x^- +1, x^+)$ in lightcone coordinates.
This spacetime contains no closed right-chiral curves (i.e. curves of constant $x^-$) and so is $\chi^+$-causal; however, it clearly contains closed left-chiral curves and so is neither $\chi^-$-causal nor causal.

Let $M$ and $N$ be 2D oriented spacetimes, and let $f : M \to N$ be a conformal map which preserves orientation and time-orientation.
We obtain several useful facts about $f$ when either its domain $M$ or its codomain $N$ is $\chi^+$-causal, analogous to \cref{prop:conf_map_to_causal_strictly_preserves_J,cor:codomain_caus_reflect_J_iff_reflect_sJ,prop:domain_caus_reflect_J_injective} for causal spacetimes.
The proofs carry over verbatim, replacing causal curves with right-chiral curves and $J$ with $\chi^+$:

\begin{proposition}
	Say codomain $N$ of $f$ is $\chi^+$-causal.
	Then $f$ strictly preserves $\chi^+$.
\end{proposition}

\begin{corollary} \label{cor:codomain_chi_caus_reflect_chi_iff_reflect_schi}
	Say codomain $N$ of $f$ is $\chi^+$-causal.
	Then $f$ reflects $\chi^+$ if and only if $f$ reflects $s\chi^+$.
\end{corollary}

\begin{proposition} \label{prop:domain_chi_causal_reflect_chi_injective}
	Say domain $M$ of $f$ is $\chi^+$-causal.
	If $f$ reflects $\chi^+$ then $f$ is injective.
\end{proposition}

A subset $U \subseteq M$ of 2D oriented spacetime $M$ is called \emph{$\chi^+$-convex} if any right-chiral curve $\gamma : [0,1] \to M$ with endpoints $\gamma(0), \gamma(1)$ lying in $U$ has $\gamma([0,1]) \subseteq U$, i.e. $\gamma$ remains always in $U$.
Clearly any causally convex $U \subseteq M$ is also $\chi^+$-convex.
There is a chiral analogue of the observation on $J$-reflecting maps in \cref{prop:causally_convex_conf_emb_reflect_J}:

\begin{proposition} \label{prop:chirally_convex_conf_inj_reflect_chi}
	Let $M$ and $N$ be 2D oriented spacetimes, and let $f : M \to N$ be a conformal map which preserves orientation and time-orientation.
	If $f$ is injective with image $f(M)$ $\chi^+$-convex in $N$, then $f$ reflects $\chi^+$.
\end{proposition}
\begin{proof}
	Similar to \cref{prop:causally_convex_conf_emb_reflect_J}, noting that $f$ is a local diffeomorphism (it is a smooth immersion whose domain and codomain have the same dimension) so $f$ is injective if and only if it is a smooth embedding.
\end{proof}

When both the domain and codomain of $f$ possess the $\chi^+$-causal property, \cref{prop:domain_chi_causal_reflect_chi_injective} can be significantly strengthened into a converse of \cref{prop:chirally_convex_conf_inj_reflect_chi}:

\begin{proposition} \label{prop:chi_causal_reflect_chi_if_chirally_convex_conf_inj}
	Let $M$ and $N$ be $\chi^+$-causal 2D oriented spacetimes, and let $f : M \to N$ be a conformal map which preserves orientation and time-orientation.
	If $f$ reflects $\chi^+$ then $f$ is injective with image $f(M)$ $\chi^+$-convex in $N$.
\end{proposition}
\begin{proof}
	Injectivity is given by \cref{prop:domain_chi_causal_reflect_chi_injective}.
% 	It remains to show that for $f$ $\chi^+$-reflecting, the image $f(M)$ is $\chi^+$-convex in $N$.
	Let $\gamma : [0,1] \to N$ be a right-chiral curve with $p,q \in M$ such that $\gamma(0) = f(p)$ and $\gamma(1) = f(q)$.
	Without loss of generality, take $\gamma$ future-directed so $(f(p),f(q)) \in \chi^+_N$; then $(p,q) \in \chi^+_M$ since $f$ reflects $\chi^+$.
	We have $f(p) \neq f(q)$ since $N$ is $\chi^+$-causal; then also $p \neq q$ since $f$ is injective.
	So there must exist right-moving null curve $\delta : [0,1] \to M$ with $\delta(0) = p$ and $\delta(1) = q$.
	
	Now both $\gamma$ and $f \circ \delta$ are right-moving null curves in $N$ starting at point $p$.
	Choose any right-chiral vector field $n_+$ on $M$.
	By \cref{lem:chiral_curves_integral}, we have reparameterisation diffeomorphisms $\phi : [0, \tau] \to [0,1]$ and $\phi' : [0, \tau'] \to [0,1]$ such that $\gamma \circ \phi
	: [0, \tau] \to N$ and $f \circ \delta \circ \phi' : [0, \tau'] \to N$ may be extended to integral curves of $n_+$ starting at $f(p) = \gamma \circ \phi (0) = f \circ \delta \circ \phi'(0)$ and containing the point $f(q) = \gamma \circ \phi (\tau) = f \circ \delta \circ \phi'(\tau')$.
% 	\todo{at \cref{lem:chiral_curves_integral} I limited my attention to curves defined on open intervals... In paragraphs below that lemma, I mention this extension business. Is this clear/formal enough?}
	Assume for the sake of contradiction that $\tau < \tau'$.
	Then since integral curves starting at the same point coincide on shared domain, $f \circ \delta \circ \phi' |_{[0,\tau]} = \gamma \circ \phi$ and in particular $f \circ \delta \circ \phi' (\tau) = \gamma \circ \phi (\tau) = f(q)$.
	But then $f \circ \delta \circ \phi'$ is not injective, in contradiction with the $\chi^+$-causal property of $N$.
	A similar contradiction arises if $\tau' < \tau$.
	Hence we must have $\tau = \tau'$, so that $\gamma \circ \phi = f \circ \delta \circ \phi' : [0, \tau] \to N$ by uniqueness of integral curves.
	Then $\gamma = f \circ \delta \circ ( \phi' \circ \phi^{-1})$, showing in particular that $\gamma ([0,1]) \subseteq f(M)$.
\end{proof}

The condition that $M$ is $\chi^+$-causal may be expressed in terms of chiral flows (\cref{def:chiral_flow}):
\begin{lemma} \label{lem:chiral_flow_free}
	Let $M$ be a 2D oriented spacetime, and let $\theta^+ : \mathbb{R}\times M \to M$ be a right-chiral flow on $M$.
	Then $M$ is $\chi^+$-causal if and only if $\theta^+$ is a free $\mathbb{R}$-action.
\end{lemma}
\begin{proof}
	The $\mathbb{R}$-action $\theta^+$ is not free if and only if there exists $p \in M$ and $\tau \in \mathbb{R}$ with $\tau \neq 0$, such that $\theta^+_\tau(p) = p$.
	This is true if and only if there is non-injective integral curve $\gamma : \mathbb{R} \to M$ of the complete right-chiral vector field $n_+$ generating $\theta^+$.
	If such $\gamma$ exists, then $\gamma$ is a closed right-chiral curve.
	Conversely, any closed right-chiral curve $\gamma' : I \to M$ may be reparameterised to an integral curve of $n_+$ by \cref{lem:chiral_curves_integral}; this integral curve is necessarily not injective since $\gamma'$ is closed.
\end{proof}

A 2D oriented spacetime $M$ is \emph{$\chi^+$-simple} if its right-chiral relation $\chi^+$ is antisymmetric and topologically closed.
This is a chiral analogue of causally simple spacetimes\footnote{Just as we may call causal spacetimes $J$-causal to compare with the $\chi^+$-causal property, we may also call causally simple spacetimes $J$-simple to compare with the $\chi^+$-simple property.}.

Clearly if $M$ is $\chi^+$-simple then it is $\chi^+$-causal.
An example of a spacetime that is $\chi^+$-causal but not $\chi^+$-simple is given by $\mathbb{R}^{1,1} \setminus \{r\}$, for arbitrary point $r$ in Minkowski $\mathbb{R}^{1,1}$; see also \cref{ex:minkowski_minus_point}.

Global hyperbolicity, which is the strongest condition in the hierarchy of causal properties in \cref{sec:spacetimes-cats-relativistic:causal_properties}, can be phrased in terms of the existence of Cauchy surfaces.
To find an appropriate chiral analogue of global hyperbolicity, we begin with a chiral analogue of Cauchy surfaces:

\begin{definition} \label{def:chiral_cauchy_surface}
	Let $M$ be a 2D oriented spacetime.
	A \emph{$\chi^+$-Cauchy surface}
	in $M$ is a smoothly embedded one-dimensional submanifold $S \subseteq M$ such that for every inextendable right-chiral curve $\gamma : I \to M$, there is a unique $t \in I$ such that $\gamma(t) \in S$.
\end{definition}

This is analogous to (smooth, spacelike) Cauchy surfaces, which are intersected exactly once by any inextendable causal curve.
Where a Cauchy surface is an appropriate surface on which to define initial conditions for relativistic equations of motion, a $\chi^+$-Cauchy surface is similarly an appropriate surface on which to define initial conditions for right-chiral equations of motion.

% ---------
% CHIRAL_EXAMPLE_EDITS
% ---------
% % inserted with example

\begin{example} \label{ex:free_fermion:initial_data}
	Recall the right-moving component $\widetilde{\psi}$ of the free fermion field theory in \cref{ex:free_fermion}.
% 	In lightcone coordinates $(x^-, x^+)$ on $\mathbb{R}^{1,1}$, the classical equation of motion for $\widetilde{\psi}$ simplifies to $\partial_+ \widetilde{\psi} = 0$; the general solution is that $\widetilde{\psi}$ is an arbitrary function of $x^-$ only:
% 	\begin{equation*}
% 		\widetilde{\psi}(x^-, x^+) = f(x^-).
% 	\end{equation*}
% 	
	In standard coordinates $(t,x)$ on $\mathbb{R}^{1,1}$, the equation of motion for $\widetilde{\psi}$ is $(\partial_t + \partial_x) \widetilde{\psi} = 0$; its general solution is that $\widetilde{\psi}$ is an arbitrary function of $x^- = t - x$ only:
	\begin{equation*}
		\widetilde{\psi}(t,x) = f(t-x).
	\end{equation*}
	Consider the 1-dimensional submanifold given by the line at constant $x=0$.
	This is a $\chi^+$-Cauchy surface of $\mathbb{R}^{1,1}$: as noted in \cref{ex:Minkowski_lightcone_coordinates}, right-chiral curves in $\mathbb{R}^{1,1}$ are lines of constant $c = t-x$, and every such line intersects the line $x=0$, specifically at $t=c$.
	
	It suffices to specify initial conditions for $\widetilde{\psi}$ on the line $x=0$: for, say initial condition $\widetilde{\psi}(t,0) = \widetilde{\psi}_0(t)$ is given for some function $\widetilde{\psi}_0$.
	Then for arbitrary $(t',x')$, we have
	\begin{equation*}
		\widetilde{\psi}(t',x') = f(t'-x') = \widetilde{\psi}(t'-x',0) = \widetilde{\psi}_0(t'-x').
	\end{equation*}
\end{example}

% ---------

Any smooth, spacelike Cauchy surface in a 2D oriented spacetime is also a $\chi^+$-Cauchy since any inextendable causal curve intersects it exactly once and right-chiral curves are causal.
The converse is not true; take for example the timelike surface $x=0$ in Minkowski $\mathbb{R}^{1,1}$.
% ---------
% CHIRAL_EXAMPLE_EDITS
% ---------
% % original
% This is a $\chi^+$-Cauchy surface;
% however, it is clearly not a Cauchy surface since it is not achronal.
% ---------
% % with example
This is a $\chi^+$-Cauchy surface as observed in the preceding example;
however, it is clearly not a Cauchy surface since it is not achronal.
% ---------

Since images of inextendable right-chiral curves coincide with equivalence classes of $s\chi^+$ (\cref{cor:sym_chi_equivalence_rel}), $\chi^+$-Cauchy surfaces in $M$ are related to the quotient space $M/s\chi^+$:

\begin{lemma} \label{lem:chiral_cauchy_sections}
	Let $M$ be a $\chi^+$-causal 2D oriented spacetime, with quotient $ M / s\chi^+$ a smooth manifold.
	The image of any smooth section $\sigma : M / s\chi^+ \to M$ of the quotient map $\pi_M: M \to M/s\chi^+$ is a $\chi^+$-Cauchy surface of $M$.
\end{lemma}
\begin{proof}
	A section $\sigma : M/s\chi^+ \to M$ of $\pi_M$ intersects each equivalence class of $s\chi^+$ in $M$ exactly once by definition.
	The image $\gamma(I)$ of any inextendable right-chiral curve $\gamma : I \to M$ is an equivalence class of $s\chi^+$, by \cref{cor:sym_chi_equivalence_rel}; thus $\sigma(M / s\chi^+)$ intersects $\gamma(I)$ at exactly one point.
	Since $M$ is $\chi^+$-causal, $\gamma$ is injective; it follows that there is a unique $t \in I$ such that $\gamma(t) \in \sigma(M / s\chi^+)$.
	
	It remains to show that if section $\sigma$ is smooth then it is an embedding.	
	That $\sigma$ is an injective immersion follows immediately from $\pi_M \circ \sigma = \mathrm{id}_{M/ s\chi^+}$.
	Any section $\sigma$ is also a proper map.
	For, if $K \subseteq M$ is compact, then by continuity so is $\pi_M (K) \subseteq M / s\chi^+$.
	Compact $K$ in Hausdorff $M$ is closed so $\sigma^{-1}(K)$ is also closed.
	But $\sigma^{-1}(K) \subseteq \pi_M(K)$: if $[p] \in M / s\chi^+$ has $\sigma[p] \in K$, then $[p] = \pi_M \circ \sigma[p] \in \pi_M (K)$.
	Since $\sigma^{-1}(K)$ is a closed subset contained in $\pi_M(K)$ a compact subset, it follows that $\sigma^{-1}(K)$ is compact in $M / s\chi^+$.
	Thus $\sigma$ is a proper injective immersion, and so a smooth embedding \cite[Proposition 4.22]{Lee2013}.
\end{proof}

We say that a 2D oriented spacetime $M$ is \emph{$\chi^+$-initial}
\todo{name? maybe $\chi^+$-globally-hyperbolic instead? `Initial' has categorical meaning...}
if it contains a $\chi^+$-Cauchy surface.
There are at least two other useful characterisations of this chiral property, one in terms of the quotient space $M /s\chi^+$ and the other in terms of chiral flows:

\begin{lemma} \label{lem:chiral_initial_R-equivariant_RxS}
	Let $M$ be any $\chi^+$-initial 2D oriented spacetime.
	Take any $S \subseteq M$ a $\chi^+$-Cauchy surface and $n_+$ a complete right-chiral vector field on $M$.
	
	The flowout $\phi^+ : \mathbb{R} \times S \to M$ from $S$ along $n_+$ is an $\mathbb{R}$-equivariant diffeomorphism with respect to the $\mathbb{R}$-actions given on $\mathbb{R} \times S$ by translation in the first factor, and on $M$ by the flow of $n_+$.
\end{lemma}
\begin{proof}
	Denote by $i : S \hookrightarrow M$ the inclusion, and $\theta^+ : \mathbb{R} \times M \to M$ the flow of $n_+$.
	The flowout $\phi^+ : \mathbb{R} \times S \to M$ from $S$ along $n_+$ is given by the composite
	\begin{equation*}
		\mathbb{R} \times S \xhookrightarrow{\mathrm{id}_\mathbb{R} \times i} \mathbb{R} \times M \xrightarrow{\theta^+} M.
	\end{equation*}
	By definition of $\chi^+$-Cauchy surfaces and \cref{lem:chiral_curves_integral}, for every $p \in M$ there is a unique $\tau_p \in \mathbb{R}$ such that $\theta^+_{-\tau_p} (p) \in S$.
	Then the map
	\begin{align*}
		M & \to \mathbb{R} \times S,
		\\
		p & \mapsto (\tau_p, \theta^+_{-\tau_p} (p)),
	\end{align*}
	is the inverse of flowout map $\phi^+$ since for $p \in M$,
	\begin{align*}
		\phi^+(\tau_p, \theta^+_{-\tau_p}(p)) = \theta^+_{\tau_p} \circ \theta^+_{-\tau_p} (p) = p,
	\end{align*}
	and for $(\tau, q) \in \mathbb{R} \times S$, it is clear that $\tau_{\theta^+_{\tau}(q)} = \tau$ so
	\begin{equation*}
		\left(\tau_{\phi^+(\tau,q)}, \theta^+_{-\tau_{\phi^+(\tau,q)}}\left({\phi^+(\tau,q)}\right)\right)
		= \left(\tau, \theta^+_{-\tau} \circ \theta^+_{\tau}(q)\right)
		= (\tau, q).
	\end{equation*}
	Since $\phi^+$ is bijective, it follows that $\phi^+$ is a diffeomorphism by the Flowout Theorem \cite[Theorem 9.20]{Lee2013}.
	
	By the group law of $\theta^+$, for any $(\tau,p) \in \mathbb{R}\times S$ and group element $t \in \mathbb{R}$,
	\begin{equation*}
		\theta^+_t(\phi^+(\tau,p))
		=
		\theta^+_t (\theta^+_\tau(p))
		=
		\theta^+_{t+\tau}(p)
		=
		\phi^+(t+\tau,p),
	\end{equation*}
	so that $\phi^+ : \mathbb{R} \times S \to M$ is $\mathbb{R}$-equivariant.
\end{proof}

\begin{theorem} \label{thm:chiral_initial_characterisations}
	Let $M$ be a 2D oriented spacetime. The following are equivalent:
	\begin{enumerate}[label=\textup{(\roman*)}]
		\item \label{thm:chiral_initial_characterisations:cauchy}
		$M$ contains a $\chi^+$-Cauchy surface,
		
		\item \label{thm:chiral_initial_characterisations:flow}
		Any right-chiral flow $\theta^+ : \mathbb{R} \times M \to M$ on $M$ is a free and proper $\mathbb{R}$-action.
		
		\item \label{thm:chiral_initial_characterisations:quotient}
		$M$ is $\chi^+$-causal and the quotient space $M / s\chi^+$ is a smooth 1-manifold such that the quotient map $\pi_M : M \to M / s\chi^+$ is a smooth submersion,
	\end{enumerate}
\end{theorem}
\begin{proof}
\labelcref{thm:chiral_initial_characterisations:cauchy} $\Rightarrow$ \labelcref{thm:chiral_initial_characterisations:flow}:
	Pick $\chi^+$-Cauchy surface $S \subseteq M$; the flowout $\phi^+ : \mathbb{R} \times S \to M$ from $S$ along any complete right-chiral vector field $n_+$ is an $\mathbb{R}$-equivariant diffeomorphism by the preceding lemma.
% 	\cref{lem:chiral_initial_R-equivariant_RxS}.
	The action of $\mathbb{R}$ on itself by translation is free and proper.
	Hence so is the $\mathbb{R}$-action on $\mathbb{R} \times S$ by translation in the first factor.
	Since $\phi^+$ is an $\mathbb{R}$-equivariant diffeomorphism intertwining this $\mathbb{R}$-action with the flow $\theta^+ : \mathbb{R} \times M \to M$ of $n_+$, it follows that $\theta^+$ is also free and proper.
	
\labelcref{thm:chiral_initial_characterisations:flow} $\Rightarrow$ \labelcref{thm:chiral_initial_characterisations:quotient}:
	Say flow $\theta^+ : \mathbb{R} \times M \to M$ of a complete right-chiral vector field $n_+$ is a free and proper $\mathbb{R}$-action.
	Then $M$ is $\chi^+$-causal by \cref{lem:chiral_flow_free}.
	By the Quotient Manifold Theorem \cite[Theorem 21.10]{Lee2013} the quotient $M / \mathbb{R}$ of $M$ by $\theta^+$ is a manifold of dimension $\dim M - \dim \mathbb{R} = 1$, with unique smooth structure such that the quotient map $\pi_M : M \to M / \mathbb{R}$ is a smooth submersion.
	But by 
% 	\cref{prop:characteristion_chiral_relation} 
	\cref{cor:sym_chi_equivalence_rel}
	the orbits of $\theta^+$ are exactly the equivalence classes of $s\chi^+$, so $M / \mathbb{R} \equiv M / s\chi^+$.
	
\labelcref{thm:chiral_initial_characterisations:quotient} $\Rightarrow$ \labelcref{thm:chiral_initial_characterisations:cauchy}:
	Pick any right-chiral flow $\theta^+ : \mathbb{R} \times M \to M$.
	By \cref{lem:chiral_flow_free}, $\theta^+$ is a free $\mathbb{R}$-action since $M$ is $\chi^+$-causal.
% 	By \cref{prop:characteristion_chiral_relation}, 
	By \cref{cor:sym_chi_equivalence_rel},
	the orbits of $\theta^+$ are the equivalence classes of $s\chi^+$, and so the fibres of the quotient map $\pi_M : M \to M/s \chi^+$.
	Thus $\pi_M : M \to M / s\chi^+$ together with $\mathbb{R}$-action $\theta^+$ constitutes a principal $\mathbb{R}$-bundle \cite[Lemma 10.3]{KolarMichorSlovak1993}.
	
	For any contractible Lie group $G$, any principal $G$-bundle $p : E \to B$ is trivial:
	for then the principal $G$-bundle $G \to \{*\}$ is a universal principal $G$-bundle, and there is only one homotopy class of maps $B \to \{*\}$ and hence only one isomorphism class of principal $G$-bundles over base $B$.
	
	So since $\mathbb{R}$ is contractible, the principal $\mathbb{R}$-bundle $\pi_M : M \to M / s\chi^+$ is trivial and thus has a smooth global section $\sigma : M / s\chi^+ \to M$.
	By \cref{lem:chiral_cauchy_sections}, $\sigma$ is a $\chi^+$-Cauchy surface of $M$.
\end{proof}

The proof above also demonstrates the following:
\begin{corollary} \label{cor:chiral_initial_has_section}
	Let $M$ be a $\chi^+$-initial 2D oriented spacetime.
	Then the quotient map $\pi_M : M \to M /s\chi^+$ has a smooth global section whose image is a $\chi^+$-Cauchy surface in $M$.
\end{corollary}

\begin{remark}
	Not all $\chi^+$-Cauchy surfaces $S \subseteq M$ may be realised as images of smooth sections of $\pi_M : M \to M /s\chi^+$.
	For example, take Minkowski $\mathbb{R}^{1,1}$ and surface $S$ defined in lightcone coordinates by $x^- = (x^+)^3$.
	Since right-chiral curves in $\mathbb{R}^{1,1}$ are simply curves of constant $x^-$, this surface is $\chi^+$-Cauchy -- every curve at constant $x^- = c$  intersects the surface exactly once, at $(c, c^{1/3})$.
	Also, $S$ is the image of map $\mathbb{R}^{1,1}/s\chi^+ \to \mathbb{R}^{1,1}$, $[(x^-,x^+)] \mapsto \left(x^-, (x^-)^{1/3}\right)$.
	While this map is indeed a section of the quotient map $\pi_{\mathbb{R}^{1,1}} : \mathbb{R}^{1,1} \to \mathbb{R}^{1,1}/s\chi^+$, it is not smooth at $[(0,x^+)]$.
	
	We may therefore compare \cref{cor:chiral_initial_has_section} to a causal analogue: not all Cauchy surfaces are smoothly embedded and spacelike, but every globally hyperbolic spacetime nonetheless contains a smoothly embedded spacelike Cauchy surface.
	See also \cref{rm:cauchy_surfaces_smooth}.
\end{remark}

We have related antisymmetry of the relation $\chi^+$ to chiral flows in \cref{lem:chiral_flow_free}.
We may do similar for topological closedness of $\chi^+$:

\begin{lemma} \label{lem:chiral_flow_proper_implies_relation_closed}
	Let $M$ be a 2D oriented spacetime, and take $\theta^+ : \mathbb{R} \times M \to M$ any right-chiral flow.
	If the $\mathbb{R}$-action $\theta^+$ is proper, then the right-chiral relation $\chi^+_M$ on $M$ is topologically closed.
\end{lemma}
\begin{proof}
	If $\theta^+ : \mathbb{R} \times M \to M$ is a proper action, then its shear map
	\begin{align*}
		\Theta^+ : \mathbb{R} \times M & \to M \times M, \\
		(\tau, p) & \mapsto (p, \theta^+_\tau (p)),
	\end{align*}
	is a proper map.
	Since $\Theta^+$ is continuous and its codomain is locally compact, it follows that $\Theta^+$ is a closed map \cite[Theorem A.57]{Lee2013}.
	Since $[0, \infty) \subseteq \mathbb{R}$ is a closed subset, so is $[0, \infty) \times M \subseteq \mathbb{R} \times M$; hence using the characterisation in \cref{prop:characteristion_chiral_relation},
	\begin{equation*}
		\chi^+_M 
		= \Set{ (p,q) \in M \times M | \exists \tau \geq 0 \text{ such that } q = \theta^+_\tau(p) }
		= \Theta^+\left([0, \infty) \times M\right)
	\end{equation*}
	is a closed subset of $M \times M$.
\end{proof}

We conclude that any $\chi^+$-initial spacetime $M$ is $\chi^+$-simple: for by \cref{thm:chiral_initial_characterisations}, any right-chiral flow $\theta^+$ on $M$ is free and proper, so $\chi^+_M$ is antisymmetric by \cref{lem:chiral_flow_free} and topologically closed by \cref{lem:chiral_flow_proper_implies_relation_closed}.

The author's lack of examples of $\chi^+$-simple spacetimes which are not $\chi^+$-initial leads to:

\begin{openquestion} \label{ques:chi_simple_implies_initial}
	If $M$ is a $\chi^+$-simple 2D oriented spacetime, does it follow that $M$ is $\chi^+$-initial?
	In light of \cref{lem:chiral_flow_free,thm:chiral_initial_characterisations}, this may be rephrased as follows: if $M$ is $\chi^+$-simple, does it follow that any right-chiral flow $\theta^+ : \mathbb{R}\times M \to M$ on $M$ is a proper $\mathbb{R}$-action?
	
	In case the answer is affirmative, then \emph{$\chi^+$-initial} and \emph{$\chi^+$-simple} are two formulations of the same chiral property.
\end{openquestion}

It is straightforward to observe that topological closedness of the relation $\chi^+$ does not by itself imply that a right-chiral flow $\theta^+ : \mathbb{R} \times M \to M$ is proper.
Recall that by one characterisation, $\mathbb{R}$-action $\theta^+$ is proper if and only if for any compact $K \subseteq M$ the set $\mathbb{R}_K := \Set{ \tau \in \mathbb{R} | \theta^+_\tau(K) \cap K \neq \varnothing}$ is compact.
Any 2D oriented spacetime $M$ with topologically closed but not antisymmetric $\chi^+_M$ has any right-chiral flow $\theta^+$ not proper: for then $M$ contains a closed right-chiral curve and so a periodic integral curve of complete right-chiral vector field $n_+$ generating $\theta^+$.
Say this periodic integral curve starts at $p \in M$ and has period $\Delta \tau$; then
$\mathbb{R}_{\{p\}} = \Set{ \tau \in \mathbb{R} | \theta^+_\tau(p) = p } = \Set{ n \Delta\tau | n \in \mathbb{Z} }$ is not compact in $\mathbb{R}$.
An example of such a spacetime is the quotient of $\mathbb{R}^{1,1}$ under identification $(x^-, x^+ + 1) \sim (x^-, x^+) $ in lightcone coordinates.

On the other hand, if right-chiral flow $\theta^+$ is a free and proper $\mathbb{R}$-action then the orbits of $\theta^+$ are closed embedded submanifolds \cite[Proposition 21.7]{Lee2013}.
It is suggestive that $\chi^+$-simple spacetimes satisfy the following necessary but not sufficient condition for the $\chi^+$-initial property\footnote{Recall \cref{cor:sym_chi_equivalence_rel}: inextendable right-chiral curves coincide with orbits of $\theta^+$.}:

\begin{proposition}
	Let $M$ be a $\chi^+$-simple 2D oriented spacetime.
	Then inextendable right-chiral curves in $M$ are closed embedded submanifolds.
\end{proposition}
\begin{proof}
	Pick a right-chiral flow $\theta^+: \mathbb{R} \times M \to M$.
	By \cref{lem:chiral_curves_integral}, inextendable right-chiral curves are diffeomorphic reparameterisations of the orbit maps $\theta^{(p)} : \mathbb{R} \to M$ defined as $\theta^{(p)}(\tau) := \theta^+_\tau(p)$.
	We show that for each $p \in M$, $\theta^{(p)}$ is a closed embedding.
	
	By definition of the flow $\theta^+$, $\theta^{(p)}$ is a smooth immersion.
	$M$ is $\chi^+$-causal since it is $\chi^+$-simple;
	then $\theta^+$ is free by \cref{lem:chiral_flow_free} so that $\theta^{(p)}$ is injective.
	It remains to show that $\theta^{(p)} : \mathbb{R} \to M$ is a closed map.
	Recall that $\mathbb{R}$ has a basis of closed sets $C(a,b) := (-\infty, a] \cup [b, \infty)$ for $a,b \in \mathbb{R}$.
	(The complements $\mathbb{R} \setminus C(a,b) = (a,b)$ form a basis of open sets.)
	Then any closed subset $A \subseteq \mathbb{R}$ is an arbitrary intersection of closed basis elements, $A = \cap_\alpha C(a_\alpha, b_\alpha)$.
	Since $\theta^{(p)}$ is injective, $\theta^{(p)}(A) = \theta^{(p)} \left(\cap_\alpha C \left(a_\alpha, b_\alpha\right)\right) = \cap_\alpha \theta^{(p)}\left(C(a_\alpha, b_\alpha)\right)$.
	So to show that $\theta^{(p)}$ is a closed map, it suffices to show that
	\begin{equation*}
		\theta^{(p)}\left(C(a, b)\right) = \theta^{(p)}\left((-\infty,a]\right) \cup \theta^{(p)}\left([b, \infty)\right)
	\end{equation*}
	is a closed set in $M$ for any $a,b \in \mathbb{R}$.
	But since the relation $\chi^+ \subseteq M \times M$ is topologically closed, so is
	\begin{align*}
		\chi^+ \cap \left( \{\theta^{(p)}(b) \} \times M \right)
		&=
% 		\Set{ (p',q') \in M \times M | p' = \theta^{(p)}(b) \text{ and } (p',q') \in \chi^+}
% 		\\&=
		\Set{ \left(\theta^{(p)}(b),q \right) \in M \times M | q = \theta^+_\tau(p) \text{ for } \tau \geq b}
		\\&=
		\{ \theta^{(p)}(b) \} \times \theta^{(p)}\left([b, \infty)\right),
	\end{align*}
	where we have used the characterisation of $\chi^+$ given in \cref{prop:characteristion_chiral_relation}.
	It follows that $\theta^{(p)}\left([b, \infty)\right) \subseteq M$ is closed.
	Similarly, $\theta^{(p)}\left((-\infty,a]\right) \subseteq M$ is closed since $ \chi^+ \cap \left( M \times \{\theta^{(p)}(a) \}\right) = \theta^{(p)}\left( (-\infty, a] \right) \times \{ \theta^{(p)}(a) \}$ is closed.
\end{proof}

Though its obvious use is as a causal property, \emph{global hyperbolicity} may also serve as a chiral property of 2D oriented spacetimes.
As observed under \cref{def:chiral_cauchy_surface}, any smooth, spacelike Cauchy surface is also a $\chi^+$-Cauchy surface.
Consequently, if 2D oriented spacetime $M$ is globally hyperbolic then it is also $\chi^+$-initial.
The converse is not true: for instance, take the open strip $U := \Set{(t,x) \in \mathbb{R}^{1,1} | x \in (-1,1)}$ in Minkowski $\mathbb{R}^{1,1}$ defined using its standard coordinates.
Then $U$ has $\chi^+$-Cauchy surface $x=0$, but is not globally hyperbolic.

Similar to the hierarchy of causal properties, we may realise the hierarchy of chiral properties as a nested sequence of full $\disjoint$-subcategories
% of $(\mathsf{CSpTm}, \disjoint_{\chi^+})$:
of $\mathsf{CSpTm}$:

\begin{definition}
	Let $\kappa^+$ denote any right-chiral property in the hierarchy above; we define $\kappa^+ \mathsf{CSpTm}$ to be the full $\disjoint$-subcategory of $(\mathsf{CSpTm}, \disjoint_{\chi^+})$ consisting of those spacetimes which satisfy $\kappa^+$.
\end{definition}
In particular, we have full $\disjoint$-subcategory inclusions:
\begin{equation*}
	\mathsf{GlobHypCSpTm} \hookrightarrow \chi^+\mathsf{InitCSpTm} \hookrightarrow \chi^+\mathsf{SimCSpTm} \hookrightarrow \chi^+\mathsf{CausCSpTm} \hookrightarrow \mathsf{CSpTm}.
\end{equation*}
If \cref{ques:chi_simple_implies_initial} has affirmative answer,
then $\chi^+\mathsf{InitCSpTm} \hookrightarrow \chi^+\mathsf{SimCSpTm}$ is the identity functor.

\begin{remark} \label{rm:chiral_hierarchy_left-right}
	There is a dual hierarchy of full $\disjoint$-subcategories of $(\mathsf{CSpTm}, \disjoint_{\chi^-})$ defined with left-chiral properties $\kappa^-$ dual to right-chiral properties $\kappa^+$.
	Any left-chiral property $\kappa^-$ is obtained by swapping right-chiral curves for left-chiral curves everywhere in the definition of right-chiral property $\kappa^+$.
	For any such pair $\kappa^+$ and $\kappa^-$,
	there is an isomorphism of $\disjoint$-categories $(\kappa^+\mathsf{CSpTm}, \disjoint_{\chi^+}) \to (\kappa^-\mathsf{CSpTm}, \disjoint_{\chi^-})$ which reverses orientation (but not time-orientation) on spacetimes.
	That this functor is an isomorphism of $\disjoint$-categories means the functor is invertible and preserves and reflects $\disjoint$-relations.
\end{remark}

\subsection{Overlap-monics in categories of spacetimes with good chiral properties}
\label{sec:spacetimes-cats-chiral:overlap-monics}

Analogous to the causal case of \cref{sec:spacetimes-cats-relativistic:overlap-monics}, we use the general characterisation of $\overlap$-monics in $(\mathsf{CSpTm}, \disjoint_{\chi^+})$ given by \cref{thm:overlap-monics-CSpTm} to produce simplified such characterisations in $\disjoint$-subcategories $\kappa^+\mathsf{CSpTm}$.

% Before doing so, we verify that restricting to such a $\disjoint$-subcategory does not introduce new $\overlap$-monicity of morphisms.
The $\disjoint$-subcategory inclusions $\kappa^+ \mathsf{CSpTm} \hookrightarrow \lambda^+ \mathsf{CSpTm}$ reflect $\overlap$-monics by \cref{prop:functors-respect-overlap-monics}.
They also preserve $\overlap$-monics; the proof is verbatim the same as the causal analogue \cref{prop:spacetime_inclusions_preserve_overlap-monics}, up to replacing $\kappa\mathsf{SpTm}$ and $\disjoint_J$ with $\kappa^+\mathsf{CSpTm}$ and $\disjoint_{\chi^+}$ respectively:

\begin{proposition} \label{prop:chiral_spacetime_inclusions_preserve_overlap-monics}
	The inclusion functors $(\kappa^+ \mathsf{CSpTm}, \disjoint_{\chi^+}) \hookrightarrow (\lambda^+ \mathsf{CSpTm}, \disjoint_{\chi^+})$ preserve $\overlap$-monics.
\end{proposition}

\begin{remark}
	Note that \cref{prop:spacetime_inclusions_preserve_overlap-monics,prop:chiral_spacetime_inclusions_preserve_overlap-monics} are proven using that any spacetime admits a covering by open globally hyperbolic neighbourhoods.
	It is relevant that global hyperbolicity is the strongest condition in both the causal and chiral hierarchies.
\end{remark}

\begin{theorem} \label{thm:overlap-monics-chiSimCSpTm}
	A morphism $f : M \to N$ in $(\chi^+\mathsf{SimCSpTm}, \disjoint_{\chi^+})$ is $\overlap$-monic if and only if the following equivalent conditions hold:
	\begin{enumerate}[label=\textup{(\roman*)}]
		\item \label{thm:overlap-monics-chiSimCSpTm:reflect_chi}
		$f$ reflects $\chi^+$,
		
		\item \label{thm:overlap-monics-chiSimCSpTm:inj_chi_convex}
		$f$ is injective and has $\chi^+$-convex image $f(M)$ in $N$.
	\end{enumerate}
\end{theorem}
\begin{proof}
	The inclusion $\chi^+\mathsf{SimCSpTm} \hookrightarrow \mathsf{CSpTm}$ preserves and reflects $\overlap$-monics with respect to $\disjoint_{\chi^+}$; hence $f$ is $\overlap$-monic in $\chi^+\mathsf{SimCSpTm}$ if and only if it is $\overlap$-monic in $\mathsf{CSpTm}$, so if and only if it reflects $\overline{s\chi^+}$ by
	\cref{thm:overlap-monics-CSpTm}.
	All spacetimes in $\chi^+\mathsf{SimCSpTm}$ have $\overline{s\chi^+} = s(\overline{\chi^+}) = s\chi^+$, where the first equality uses \cref{lem:sym_top_closures_commute}.
	Thus $f$ is $\overlap$-monic if and only if it reflects $s\chi^+$.
	
	Since $M$ and $N$ are necessarily $\chi^+$-causal, this is equivalent to \labelcref{thm:overlap-monics-chiSimCSpTm:reflect_chi} by \cref{cor:codomain_chi_caus_reflect_chi_iff_reflect_schi}, which in turn is equivalent to \labelcref{thm:overlap-monics-chiSimCSpTm:inj_chi_convex} by \cref{prop:chirally_convex_conf_inj_reflect_chi,prop:chi_causal_reflect_chi_if_chirally_convex_conf_inj}.
\end{proof}

Restriction to $\chi^+$-initial or globally hyperbolic spacetimes is trivial, and the resulting characterisation parallels the causal case of \cref{thm:overlap-monics-GlobHypSpTm}:
\begin{corollary} \label{cor:overlap-monics-chiInitCSpTm}
	A morphism $f : M \to N$ in either $(\chi^+\mathsf{InitCSpTm}, \disjoint_{\chi^+})$ or $(\mathsf{GlobHypCSpTm}, \disjoint_{\chi^+})$ is $\overlap$-monic in that $\disjoint$-category if and only if the following equivalent conditions hold:
	\begin{enumerate}[label=\textup{(\roman*)}]
		\item
		$f$ reflects $\chi^+$,
		
		\item
		$f$ is injective and has $\chi^+$-convex image $f(M)$ in $N$.
	\end{enumerate}
\end{corollary}

In analogy with the nomenclature for the standard category $\mathsf{Loc}$ of spacetimes for relativistic AQFTs:
\begin{definition} \label{def:chi-loc}
	Denote by $\chi \mathsf{Loc}$ the wide $\disjoint$-subcategory in $\chi^+\mathsf{InitCSpTm}$ of $\overlap$-monics with respect to $\disjoint_{\chi^+}$:
	\begin{equation*}
		\chi \mathsf{Loc} := \OverlapMonics(\chi^+\mathsf{InitCSpTm}, \disjoint_{\chi^+}).
	\end{equation*}
\end{definition}
Then $\chi \mathsf{Loc}$ has as objects all two-dimensional oriented spacetimes $M$ containing $\chi^+$-Cauchy surfaces, and as morphisms all orientation- and time-orientation-preserving injective conformal maps $f : M \to N$ with $\chi^+$-convex image $f(M)$ in $N$.
It is equipped with $\disjoint$-relation $\disjoint_{\chi^+}$, which is in particular an orthogonality relation in the sense of \cite{BeniniSchenkelWoike2021}.

As per our proposal described in \cref{sec:introduction}, we therefore expect that the right-moving half of a chiral CFT may be formulated as a functor
\begin{equation*}
	\mathcal{A}:\chi\mathsf{Loc} \to \mathsf{Obs},
\end{equation*}
satisfying the causality axiom with respect to $\disjoint_{\chi^+}$.
It remains to describe a time slice axiom on such functors.

\begin{definition}
	For 2D oriented spacetimes $M$ and $N$, an orientation- and time-orientation-preserving conformal map $f : M \to N$ is called a \emph{$\chi^+$-Cauchy map} if its image $f(M)$ contains a $\chi^+$-Cauchy surface of $N$.
\end{definition}
Evidently, the codomain of any $\chi^+$-Cauchy map is necessarily $\chi^+$-initial.

A functor $\mathcal{A} : \chi\mathsf{Loc} \to \mathsf{Obs}$ obeys the \emph{chiral time-slice axiom}\footnote{Since $\chi^+$-Cauchy surfaces need not be achronal, `time-slice' axiom is a misnomer; we use it for compatibility with standard terminology.} if $\mathcal{A}f$ is an isomorphism in $\mathsf{Obs}$ for every morphism $f$ in $\chi\mathsf{Loc}$ which is a $\chi^+$-Cauchy map.
Equivalently, in the terminology of \cite{BeniniSchenkelWoike2021}, $\mathcal{A} : \chi\mathsf{Loc} \to \mathsf{Obs}$ obeys the chiral time-slice axiom if it is $W$-constant, for $W$ the collection of morphisms in $\chi\mathsf{Loc}$ which are $\chi^+$-Cauchy maps.

\begin{remark}
	Following from \cref{rm:chiral_hierarchy_left-right}, there is an isomorphism of orthogonal categories $\OverlapMonics(\chi^-\mathsf{InitCSpTm}, \disjoint_{\chi^-}) \to \OverlapMonics(\chi^+\mathsf{InitCSpTm}, \disjoint_{\chi^+})$ which reverses orientations on spacetimes.
	Given this isomorphism, we do not retain any notational distinction between the category $\chi\mathsf{Loc} = \OverlapMonics(\chi^+\mathsf{InitCSpTm}, \disjoint_{\chi^+})$ proposed for right-moving halves of chiral CFTs, and its counterpart for left-moving halves.
	
	It follows that the left-moving half of a chiral CFT is also a functor $\chi\mathsf{Loc} \to \mathsf{Obs}$ satisfying causality and time-slice axioms just as the right-moving half.
	Only upon combining the two halves into a full CFT does a distinction need to be made between left- and right-moving halves.
\end{remark}

\begin{remark}
	We define $\chi\mathsf{Loc}$ using $\chi^+$-initial spacetimes rather than globally hyperbolic spacetimes because, in their own right, (right-)chiral theories require well-posedness of (right-)chiral initial value problems and hence the existence of $\chi^+$-Cauchy surfaces rather than standard Cauchy surfaces.
	But for chiral theories arising as chiral halves of full CFTs, it may be more appropriate to restrict to globally hyperbolic spacetimes in the definition of $\chi\mathsf{Loc}$.
	For the discussion in the next section at least, this distinction is unimportant: the weaker $\chi^+$-initial property is enough.
\end{remark}

\begin{remark}
	Non-chiral 2D CFTs may be formulated as functors $\mathsf{CLoc}_{1+1}^\mathrm{o,to} \to \mathsf{Obs}$.
	The domain $\mathsf{CLoc}_{1+1}^\mathrm{o,to}$ consists of globally hyperbolic 2D oriented spacetimes and orientation- and time-orientation-preserving conformal maps which are injective with causally convex image \cite{Pinamonti2009}.
	It can be equipped with a causal $\disjoint$-relation $\disjoint_J$ defined similarly to \cref{def:causal_discon}.
	Compare to $\mathsf{CLoc}_{d+1}$ discussed in \cref{rm:CLoc}: in $\mathsf{CLoc}_{1+1}^\mathrm{o,to}$ we have restricted to oriented spacetimes and orientation- and time-orientation preserving maps, but the settings are otherwise similar.

	There is an evident inclusion functor
	\begin{equation*}
		\iota : (\mathsf{CLoc}_{1+1}^\mathrm{o,to}, \disjoint_J) \hookrightarrow (\chi\mathsf{Loc}, \disjoint_{\chi^+}),
	\end{equation*}
	using that any globally hyperbolic 2D oriented spacetime is necessarily $\chi^+$-initial, and any causally convex subset is necessarily $\chi^+$-convex.
	This functor $\iota$ preserves $\disjoint$-relations: if there is no causal curve in codomain $N$ connecting images $f_i(M_i)$ of conterminous pair $f_1 : M_1 \rightarrow N \leftarrow M_2 : f_2$ in $\mathsf{CLoc}_{1+1}^\mathrm{o,to}$, then necessarily also there is no right-chiral curve in $N$ connecting them.
	
	Because any (smooth, spacelike) Cauchy surface in a 2D oriented spacetime is also a $\chi^+$-Cauchy surface, it follows that Cauchy maps between 2D oriented spacetimes are also $\chi^+$-Cauchy maps.
	Consequently, the inclusion $\iota$ sends Cauchy maps in $\mathsf{CLoc}_{1+1}^\mathrm{o,to}$ to $\chi^+$-Cauchy maps in $\chi\mathsf{Loc}$.

	As a result, any chiral CFT formulated as a functor $\mathcal{A} : \chi\mathsf{Loc} \to \mathsf{Obs}$ may be pre-composed with $\iota$ to describe a general CFT $\mathcal{A}\circ \iota : \mathsf{CLoc}_{1+1}^\mathrm{o,to} \to \mathsf{Obs}$.
	If $\mathcal{A}$ satisfies the causality axiom with respect to $\disjoint_{\chi^+}$ then $\mathcal{A} \circ \iota$ satisfies the causality axiom with respect to $\disjoint_J$ because $\iota$ preserves $\disjoint$-relations.
	If $\mathcal{A}$ satisfies the chiral time-slice axiom (with respect to $\chi^+$-Cauchy maps) then $\mathcal{A} \circ \iota$ satisfies the time-slice axiom (with respect to Cauchy maps) because $\iota$ sends Cauchy maps to $\chi^+$-Cauchy maps.
\end{remark}

\subsection{Comparison with established approaches to chiral CFT}
\label{sec:spacetimes-cats-chiral:comparison}

Our proposed formulation of chiral CFTs as functors $\chi\mathsf{Loc} \to \mathsf{Obs}$ differs from established AQFT formulations of chiral CFTs.
Commonly, in the older net-based formulations of AQFT, a chiral CFT is defined on a net of open subsets of the circle $\mathbb{S}^1$.
An analogous categorical formulation, for instance as in \cite{BeniniGiorgettiSchenkel2021}, uses a category $\mathsf{Emb}_1$ of one-dimensional oriented manifolds and smooth orientation-preserving embeddings.
(The net of open subsets of $\mathbb{S}^1$ may be recovered from the categorical formulation as the skeleton of the slice category $\mathsf{Emb}_1 / \mathbb{S}^1$.)

The key difference between these established formulations and our proposal is that the established formulations use one-dimensional manifolds as representatives of spacetimes.
This builds a time-slice axiom into the formulation preemptively; a functor $\mathsf{Emb}_1 \to \mathsf{Obs}$ may be understood as a chiral CFT without imposing that it obey any time-slice axiom (so long as it obeys an appropriate causality property).

We construct below a functor $\chi\mathsf{Loc} \to \mathsf{Emb}_1$ which sends any morphisms $f$ in $\chi\mathsf{Loc}$ which are $\chi^+$-Cauchy maps to isomorphisms in $\mathsf{Emb}_1$.
It then follows that any functor $\mathsf{Emb}_1 \to \mathsf{Obs}$ representing a chiral CFT in the established formulation may be pre-composed with the functor $\chi\mathsf{Loc} \to \mathsf{Emb}_1$ to give a chiral CFT obeying the chiral time-slice axiom in our proposed formulation.

To begin, let us consider in more detail the established formulation.
Denote by $\mathsf{Man}_1$ the category whose objects are smooth, oriented, one-dimensional manifolds, and whose morphisms are orientation-preserving local diffeomorphisms.
Since this is a concrete category, we may equip it with the $\disjoint$-relation $\disjoint_\mathrm{set}$ of setwise-disjointness, as per \cref{ex:setwise_disjointness:concrete}.

\begin{proposition} \label{prop:overlap-monics-Man}
	A morphism $f : X \to Y$ of $(\mathsf{Man}_1, \disjoint_\mathrm{set})$ is $\overlap$-monic if and only if it is injective.
\end{proposition}
\begin{proof}
	From \cref{ex:setwise_disjointness:concrete,ex:setwise_disjointness:diagonal}, $\disjoint_\mathrm{set}$ on $\mathsf{Man}_1$ is the pullback of $\disjoint_\mathrm{bin}$ along
	\begin{equation*}
		\mathsf{Man}_1 \xrightarrow{U} \mathsf{Set} \xrightarrow{\Delta} \mathsf{sBin},
	\end{equation*}
	where $U$ is the forgetful functor to $\mathsf{Set}$ and $\Delta$ equips sets with the diagonal relation.
	So morphisms $g_i : W_i \to X$ in $\mathsf{Man}_1$ have $g_1 \disjoint_\mathrm{set} g_2$ if and only if $g_1(W_1) \times g_2(W_2)$ does not intersect the diagonal relation $\Delta_X := \Set{(x,x) \in X \times X}$ on $X$.
	
	Since all morphisms in $\mathsf{Man}_1$ are local diffeomorphisms and hence open maps, $f : X \to Y$ is $\overlap$-monic with respect to $\disjoint_\mathrm{set}$ if and only if $f$ reflects $\overline{s\Delta}$;
	the proof is verbatim the same as \cref{thm:overlap-monics-SpTm}, with $\mathsf{SpTm}_{d+1}$, spacetimes and relations $J$ replaced by everywhere by $\mathsf{Man}_1$, oriented 1-manifolds and relations $\Delta$.
	
	But diagonal relations $\Delta$ here are symmetric by definition and topologically closed since manifolds are Hausdorff, so $\overline{s\Delta} = \Delta$.
	Hence $f$ is $\overlap$-monic with respect to $\disjoint_\mathrm{set}$ if and only if $f$ reflects $\Delta$, i.e. is injective.
\end{proof}

\begin{remark}
	Compare $(\mathsf{Man}_1, \disjoint_\mathrm{set})$ to $(\mathsf{Set}, \disjoint_\mathrm{set})$: in both cases, morphisms are $\overlap$-monic if and only if they are injective.
	However, the methods of proof of this fact differ.
	For $(\mathsf{Set}, \disjoint_\mathrm{set})$, the argument in \cref{ex:setwise_disjointness:overlap_monics} uses the pointlike generalised elements $\{*\} \to X$ of any object $X$.
	Since $(\mathsf{Man}_1, \disjoint_\mathrm{set})$ lacks such pointlike generalised elements, $\overlap$-monics must be characterised via the topological closures $\overline{s\Delta}$ of the underlying binary relations.
\end{remark}

Then the orthogonal category $\mathsf{Emb}_1$ used in the established formulation of chiral CFTs is the $\disjoint$-subcategory of $\overlap$-monics, $\mathsf{Emb}_1 := \OverlapMonics(\mathsf{Man}_1, \disjoint_\mathrm{set})$.
Since its morphisms are injective local diffeomorphisms, they are smooth embeddings.

On any two-dimensional oriented spacetime $M$, we have $s\chi^+$ an equivalence relation by \cref{cor:sym_chi_equivalence_rel}.
We claim that the assignment of the quotient $M / s\chi^+$ to spacetime $M$ gives a functor from $\chi^+\mathsf{InitCSpTm}$ to $\mathsf{Man}_1$:

\begin{theorem}
	There is a functor $Q : \chi^+\mathsf{InitCSpTm} \to \mathsf{Man}_1$ acting on objects $M$ by $QM = M / s\chi^+$, and on morphisms $f : M \to N$ by
	\begin{align}
		Qf : \faktor{M}{s\chi^+}  &\to \faktor{N}{s\chi^+}, \label{eqn:Q:morphisms}\\
			[p] & \mapsto [f(p)]. \nonumber
	\end{align}
% 	$Qf : M / s\chi^+ \to N / s\chi^+$ defined as $(Qf)[p] = [f(p)]$ for $p \in M$.
\end{theorem}
\begin{proof}
	We show that the spaces and maps so prescribed are well-defined objects and morphisms of $\mathsf{Man}_1$; it is trivial after this to see that $Q$ is functorial.
	
	\Cref{thm:chiral_initial_characterisations} shows that the quotient space $M / s\chi^+$ of $\chi^+$-initial $M$ is a smooth 1-manifold such that the quotient map $\pi_M : M \to M/s\chi^+$ is a smooth submersion.
	To define an orientation on $M / s\chi^+$,
	observe that there exists on $M / s\chi^+$ a nowhere-vanishing 1-form $\omega$ such that $\pi_M^* \omega$ is a right-chiral 1-form on $M$.
	Choosing a chiral frame $(n_-,n_+)$ on $M$ with dual coframe $(-\eta^+, -\eta^-)$, this means there is some $\kappa \in C^\infty(M)$ such that $\pi_M^* \omega = e^\kappa \cdot \eta^+$.
	
	Specifically, since $M / s\chi^+$ is a smooth 1-manifold, it is orientable; pick an arbitrary nowhere-vanishing 1-form $\widetilde{\omega}$ on it.
	By expanding in coframe $(-\eta^+, -\eta^-)$, we have
	\begin{equation*}
		\pi_M^* \widetilde{\omega} = \lambda_- \eta^- + \lambda_+ \eta^+,
	\end{equation*}
	for some coefficient functions $\lambda_\pm \in C^\infty(M)$.
	Because $n_+$ is tangent to the fibres of $\pi_M$ (\cref{cor:sym_chi_equivalence_rel}),
% 	Because $n_+$ is vertical with respect to $\pi_M$,
	so $d\pi_M|_p (n_+) = 0$ at all $p \in M$, it follows that
	\begin{align*}
		-\lambda_-(p) = \lambda_-(p) \eta^-_p(n_+) + \lambda_+(p) \eta^+_p(n_+)
		= (\pi_M^* \widetilde{\omega})_p (n_+)
		= \widetilde{\omega}_{[p]} ( d\pi_M |_p (n_+))
		= 0,
	\end{align*}
	for all $p \in M$; hence $\pi_M^* \widetilde{\omega} = \lambda_+ \eta^+$.
	
	Since $\pi_M$ is a submersion and $\widetilde{\omega}$ is nowhere-vanishing, it follows that $\lambda_+$ is also nowhere-vanishing.
	So on each path-component of $M$, $\lambda_+$ has fixed sign, i.e. $\operatorname{sgn} \lambda_+ : M \to \{-1,1\} \subseteq \mathbb{R}$ is locally constant.
	Since the orbits of $\theta^+$ are path-connected, this means $\operatorname{sgn} \lambda_+$ is constant on those orbits and so descends to the quotient, i.e. there is $\rho : M / s\chi^+ \to \mathbb{R}$ such that the following diagram is commutative:
	\begin{equation*}
	% https://q.uiver.app/?q=WzAsMyxbMCwwLCJNIl0sWzIsMCwiXFxtYXRoYmJ7Un0iXSxbMSwxLCJNIC8gc1xcY2hpXisiXSxbMCwxLCJcXG9wZXJhdG9ybmFtZXtzZ259XFxsYW1iZGFfKyJdLFswLDIsIlxccGlfTSIsMl0sWzIsMSwiXFxyaG8iLDJdXQ==
	\begin{tikzcd}
		M && {\mathbb{R}} \\
		& {M / s\chi^+}
		\arrow["{\operatorname{sgn}\lambda_+}", from=1-1, to=1-3]
		\arrow["{\pi_M}"', from=1-1, to=2-2]
		\arrow["\rho"', from=2-2, to=1-3]
	\end{tikzcd}
	\end{equation*}
	In particular, $\rho$ is also locally constant; for if $\operatorname{sgn} \lambda_+$ is constant on a neighbourhood $U$ of $p$, then $\rho$ is constant on $\pi_M(U)$ which is a neighbourhood of $[p] = \pi_M(p)$ since the submersion $\pi_M$ is an open map.
	
	With this, set $\omega := \rho \cdot \widetilde{\omega}$.
	Then $\pi_M^* \omega = (\rho \circ \pi_M) \cdot \pi_M^* \widetilde{\omega} = (\operatorname{sgn} \lambda_+) \cdot \lambda_+ \eta^-$,
	and the coefficient function $(\operatorname{sgn} \lambda_+) \cdot \lambda_+ =: e^\kappa$ is everywhere-strictly-positive by construction.

	Equip 1-manifold $M / s\chi^+$ with the orientation represented by this 1-form $\omega$.
	This orientation is well-defined, i.e. any nowhere-vanishing 1-forms $\omega$ and $\omega'$ on $M/s\chi^+$ which both have $\pi_M^* \omega$ and $\pi_M^* \omega'$ right-chiral on $M$ define the same orientation.
	For, say $\omega$ and $\omega'$ represent distinct orientations of $M / s\chi^+$, so that there exists some path-component $X$ of $M / s\chi^+$ on which the restrictions have $\omega'|_X = - e^\mu \omega|_X$ for some $\mu \in C^\infty(X)$.
	Then $\pi_M^* (\omega'|_X) = - e^{\mu \circ \pi_M} \pi_M^* (\omega|_X)$, so that $\pi_M^* \omega'|_{\pi_M^{-1}(X)}$ and $\pi_M^* \omega|_{\pi_M^{-1}(X)}$ differ by a sign; in particular, they cannot both be right-chiral.
	
	This defines smooth, oriented 1-manifold $QM = M / s\chi^+$ an object of $\mathsf{Man}_1$ for any object $M$ of $\chi^+\mathsf{InitCSpTm}$.
	Next we show that, for $f : M \to N$ a smooth conformal map which preserves orientation and time-orientation, the induced map \cref{eqn:Q:morphisms} on quotients is a well-defined local diffeomorphism which preserves orientation.
	
	Well-definedness of \cref{eqn:Q:morphisms} follows because morphisms $f : M \to N$ of $\chi^+\mathsf{InitCSpTm}$ preserve the relation $\chi^+$.
	By definition of $Qf$, the following diagram is commutative:
	\begin{equation*}
		% https://q.uiver.app/?q=WzAsNCxbMCwwLCJNIl0sWzEsMCwiTiJdLFsxLDEsIk4gXFxiaWcvIHMoXFxjaGleKykiXSxbMCwxLCJNIFxcYmlnLyBzKFxcY2hpXispIl0sWzAsMSwiZiJdLFsxLDIsIlxccGlfTiJdLFswLDMsIlxccGlfTSIsMl0sWzMsMiwiUWYiLDJdXQ==
	\begin{tikzcd}
		M & N \\
		{M / s\chi^+} & {N / s\chi^+}
		\arrow["f", from=1-1, to=1-2]
		\arrow["{\pi_N}", from=1-2, to=2-2]
		\arrow["{\pi_M}"', from=1-1, to=2-1]
		\arrow["Qf"', from=2-1, to=2-2]
	\end{tikzcd}
	\end{equation*}
	Then $Qf$ is smooth since $ Qf \circ \pi_M = \pi_N \circ f$ is smooth and $\pi_M$ is a smooth submersion \cite[Theorem 4.29]{Lee2013}.
	
	Take chiral frames $(n_-^M, n_+^M)$ on $M$ and $(n_-^N, n_+^N)$ on $N$, with dual coframes $(-\eta^+_M , -\eta^-_M)$ and $(-\eta^+_N, \eta^-_N)$ respectively.
	Since $f$ is conformal, orientation- and time-orientation-preserving, there exist $F_\pm \in C^\infty(M)$ such that $df \circ n_\pm^M = e^{F_\pm} \cdot n_\pm^N \circ f$ or equivalently $f^* \eta^\pm_N = e^{F_\mp} \eta^\pm_M$ by \cref{prop:chiral_frame_characterisation_conf_map}.

	To show that $Qf$ is a local diffeomorphism, it suffices to show that its tangent map $d(Qf)_{[p]}$ is not the zero map for any $[p]\in M / s\chi^+$ since both the domain and codomain of $Qf$ are one-dimensional.
	For $q \in N$, the left-chiral vector $n_-^N(q)$ on $N$ is not in $\ker d\pi_N|_q = \operatorname{span}\{n_+^N(q)\}$.
	Then for any $p \in M$,
	\begin{equation*}
		d(Qf)_{[p]} \circ d\pi_M |_p (n_-^M|_p)
		= d \pi_N |_{f(p)} \circ df_p(n_-^M|_p)
		= e^{F_-(p)} \cdot d \pi_N|_{f(p)} (n_-^N|_{f(p)})
		\neq 0,
	\end{equation*}
	so $d(Qf)_{[p]}$ is not the zero map.
	
	Denote by $\omega_N$ a nowhere-vanishing 1-form representing the orientation on $QN = N / s\chi^+$.
	This means $\pi_N^* \omega_N$ is right-chiral and so there exists $\kappa \in C^\infty(N)$ such that $\pi_N^* \omega_N = e^\kappa \eta^+_N$.
	Then
	\begin{align*}
		\pi_M^* (Qf^*\omega_N)
		&=
		(Qf \circ \pi_M)^* \omega_N
		=
		(\pi_N \circ f)^* \omega_N
		=
		f^* (e^\kappa \eta^+_N)
		=
		e^{\kappa \circ f + F_-} \eta^+_M,
	\end{align*}
	which demonstrates that the nowhere-vanishing 1-form $Qf^* \omega_N$ on $M / s\chi^+$ has $\pi_M^* (Qf^*\omega_N)$ right-chiral on $M$.
	Hence $Qf^* \omega_N$ agrees with the orientation on $M / s\chi^+$, and consequently $Qf$ preserves orientation.
\end{proof}

Taking quotients $M /s\chi^+$ of spacetimes $M$ moreover relates the right-chiral disjointness relation on spacetimes to the setwise-disjointness relation on their quotients:

\begin{proposition} \label{prop:Q_reflects_preserves_disjointness_relations}
	The functor $Q : (\chi^+\mathsf{InitCSpTm}, \disjoint_{\chi^+}) \to (\mathsf{Man}_1, \disjoint_\mathrm{set})$ both preserves and reflects $\disjoint$-relations.
\end{proposition}
\begin{proof}
	Take any conterminous pair $M_1 \xrightarrow{f_1} N \xleftarrow{f_2} M_2$ in $\chi^+\mathsf{InitCSpTm}$.
	Then $Qf_1 \overlap_\mathrm{set} Qf_2$ holds if and only if there exist arbitrary representatives $p_i \in M_i$ of $[p_i] \in M_i/ s\chi^+$ with $Qf_1 [p_1] = Qf_2 [p_2]$, which holds if and only if there exist $p_i \in M_i$ with $(f_1(p_1), f_2(p_2)) \in s\chi^+_N$, i. e. $f_1 \overlap_{\chi^+} f_2$.
\end{proof}
Consequently by \cref{prop:functors-respect-overlap-monics}, $Q$ reflects $\overlap$-monics.
\todo[color=cyan]{Is $Q$ also full and essentially surjective on objects?}
Further, using the characterisations of $\overlap$-monics established above in its domain and codomain, $Q$ also preserves $\overlap$-monics:

\begin{proposition}
	The functor $Q : (\chi^+\mathsf{InitCSpTm}, \disjoint_{\chi^+}) \to (\mathsf{Man}_1, \disjoint_\mathrm{set})$ preserves $\overlap$-monics.
\end{proposition}
\begin{proof}
	Let morphism $f : M \to N$ of $(\chi^+\mathsf{InitCSpTm}, \disjoint_{\chi^+})$ be $\overlap$-monic, so $f$ reflects the chiral relation $\chi^+$ by \cref{cor:overlap-monics-chiInitCSpTm}.
	The induced map $Qf$ on quotients is injective: if $[p], [p'] \in M / s\chi^+$ have $Qf[p] = [f(p)] = [f(p')] = Qf[p']$, then $(f(p), f(p')) \in s\chi^+_N$ so $(p,p') \in s\chi^+_M$ since $f$ reflects $s\chi^+$; then $[p] = [p']$.
	Hence by \cref{prop:overlap-monics-Man}, $Qf$ is $\overlap$-monic in $(\mathsf{Man}_1, \disjoint_\mathrm{set})$.
\end{proof}

Then $Q$ restricts to a functor $\chi\mathsf{Loc} \to \mathsf{Emb}_1$ on the subcategories of $\overlap$-monics, which we denote by the same name.
The morphisms made invertible by this functor are exactly the $\chi^+$-Cauchy maps contained in $\chi\mathsf{Loc}$:

\begin{proposition} \label{prop:chi_cauchy_iso}
	A morphism $f : M \to N$ in $\chi\mathsf{Loc}$ is a $\chi^+$-Cauchy map if and only if $Qf : M/s\chi^+ \to N /s\chi^+$ is an isomorphism in $\mathsf{Emb}_1$.
\end{proposition}
\begin{proof}
$(\Rightarrow)$:
	Let $S \subseteq f(M) \subseteq N$ be a $\chi^+$-Cauchy surface of $N$.
% 	Since $Q$ preserves $\overlap$-monics, $Qf$ is a morphism of $\mathsf{Emb}_1$ and so a smooth embedding; 
	To show that $Qf$ is a diffeomorphism, and so an isomorphism in $\mathsf{Emb}_1$, it suffices to show that it is surjective.
	Take any $[q] \in N/s\chi^+$ with arbitrarily chosen representative $q \in N$, and let $\gamma : I \to N$ be an inextendable right-chiral curve through $q = \gamma(0)$.
	By definition of $\chi^+$-Cauchy surfaces, there exists a unique $t \in I$ such that $\gamma(t) \in S$.
	Since $S \subseteq f(M)$ there exists $p \in M$ with $f(p) = \gamma(t)$, and hence
	\begin{equation*}
		Qf[p] = [f(p)] = [\gamma(t)] = [\gamma(0)] = [q].
	\end{equation*}
	
$(\Leftarrow)$:
	Say $Qf$ has inverse $Qf^{-1}$.
	By \cref{cor:chiral_initial_has_section}, there is a smooth section $\sigma$ of the quotient map $\pi_M : M \to M/s\chi^+$.
	Then the composite
	\begin{equation*}
		\faktor{N}{s\chi^+} \xrightarrow{Qf^{-1}} \faktor{M}{s\chi^+} \xrightarrow{\sigma} M \xrightarrow{f} N
	\end{equation*}
	of smooth maps satisfies
	\begin{equation*}
		\pi_N \circ (f \circ \sigma \circ Qf^{-1}) = Qf \circ (\pi_M \circ \sigma) \circ Qf^{-1} = Qf \circ \mathrm{id}_{M/s\chi^+} \circ Qf^{-1} = \mathrm{id}_{N/s\chi^+}.
	\end{equation*}
	So $f \circ \sigma \circ Qf^{-1}$ is a smooth section of $\pi_N : N \to N/s\chi^+$.
	Then by \cref{lem:chiral_cauchy_sections}, the image of $f \circ \sigma \circ Qf^{-1}$ is a $\chi^+$-Cauchy surface in $N$ and is evidently contained in $f(M)$.
\end{proof}

As an aside, from this it follows that compositions of $\chi^+$-Cauchy maps in $\chi\mathsf{Loc}$ are also $\chi^+$-Cauchy maps: for if $f: M \to N$ and $g: N \to P$ are $\chi^+$-Cauchy in $\chi\mathsf{Loc}$, then $Qf$, $Qg$ are isomorphisms in $\mathsf{Emb}_1$; then so too is $Q(gf)$, so that $gf$ is $\chi^+$-Cauchy.
As a consequence, the collection $W$ of morphisms in $\chi\mathsf{Loc}$ which are $\chi^+$-Cauchy constitute a wide subcategory $\mathsf{W}$.
Moreover, this subcategory possesses the 2-of-6 property:

\begin{corollary} \label{cor:Cauchy_weak_equivalences}
	The wide subcategory $\mathsf{W}$ of $\chi^+$-Cauchy maps in $\chi\mathsf{Loc}$ satisfies the following 2-of-6 property: for composable triple $f,g,h$ of morphisms in $\chi\mathsf{Loc}$, if $gf$ and $hg$ are in $\mathsf{W}$ then so are $f$, $g$, $h$, and $hgf$.
\end{corollary}
\begin{proof}
	This follows from the 2-of-6 property of isomorphisms via the preceding proposition: say composable triple $f,g,h$ in $\chi\mathsf{Loc}$ has $Q(gf)$ and $Q(hg)$ isomorphisms in $\mathsf{Emb}_1$.
	Then $Qg$ has left-inverse $Qf \circ Q(gf)^{-1}$; also since $Q(hg) = Qh \circ Qg$ is an isomorphism, $Qg$ is monic so left-inverse $Qf \circ Q(g f)^{-1}$ of $Qg$ is also a right-inverse:
	\begin{equation*}
		Qg \circ \left\{\left[Qf \circ Q(gf)^{-1}\right] \circ Qg\right\}
		=
		\mathrm{id} \circ Qg
		=
		Qg \circ \mathrm{id}.
	\end{equation*}
	It then follows since $Qg$ is invertible that $Qf$ has inverse $Q(gf)^{-1} \circ Qg$, $Qh$ has inverse $Qg \circ Q(hg)^{-1}$, and $Q(hgf)$ has inverse $Q(gf)^{-1} \circ Qg \circ Q(hg)^{-1}$.
\end{proof}
Consequently, the data $(\chi\mathsf{Loc}, \mathsf{W})$ may be understood as a homotopical category, with $\chi^+$-Cauchy morphisms $\mathsf{W}$ its weak equivalences; see for instance \cite[Definition 2.1.1]{Riehl2014}.

The functor $Q : \chi\mathsf{Loc} \to \mathsf{Emb}_1$ allows us to produce from any chiral CFT in the established formalism a chiral CFT in our proposed formalism:

\begin{proposition}
	Consider any functor $\mathcal{A} : \mathsf{Emb}_1 \to \mathsf{Obs}$.
	Then the composite $\mathcal{A} \circ Q : \chi\mathsf{Loc} \to \mathsf{Obs}$ obeys the chiral time-slice axiom.
	
	Moreover, if $\mathcal{A}$ obeys the causality axiom with respect to orthogonality relation $\disjoint_\mathrm{set}$ on domain $\mathsf{Emb}_1$, then $\mathcal{A} \circ Q$ obeys the causality axiom with respect to orthogonality relation $\disjoint_{\chi^+}$ on $\chi\mathsf{Loc}$.
\end{proposition}
\begin{proof}
	The former statement follows because any $\chi^+$-Cauchy map $f$ in $\chi\mathsf{Loc}$ has $Qf$ and hence necessarily $\mathcal{A} Qf$ an isomorphism (\cref{prop:chi_cauchy_iso}).
	
	The latter statement follows since $Q$ preserves $\disjoint$-relations (\cref{prop:Q_reflects_preserves_disjointness_relations}): if conterminous pair $f_1 : M_1 \rightarrow N \leftarrow M_2 : f_2$ in $\chi\mathsf{Loc}$ has $f_1 \disjoint_{\chi^+} f_2$, then necessarily $Qf_1 \disjoint_\mathrm{set} Qf_2$ in $\mathsf{Emb}_1$.
	Then if $\mathcal{A}: \mathsf{Emb}_1 \to \mathsf{Obs}$ obeys the causality axiom, it follows that
	\begin{equation*}
		\left[\mathcal{A} Q f_1 (a_1), \mathcal{A} Q f_2 (a_2)\right]_{\mathcal{A} Q N} = 0,
	\end{equation*}
	for any observables $a_1 \in \mathcal{A} QM_1$ and $a_2 \in \mathcal{A} QM_2$.
\end{proof}

It is not clear to us yet whether a converse to the above holds:

\begin{openquestion}
	Say a functor $\mathcal{B} : \chi\mathsf{Loc} \to \mathsf{Obs}$ satisfies both the causality axiom with respect to $\disjoint_{\chi^+}$ and the chiral time-slice axiom.
	Does there exist a functor $\widetilde{\mathcal{B}} : \mathsf{Emb}_1 \to \mathsf{Obs}$ satisfying the causality axiom with respect to $\disjoint_\mathrm{set}$ such that $\mathcal{B} \cong \widetilde{\mathcal{B}} \circ Q$?
\end{openquestion}

One possibility is that the functor $Q : \chi\mathsf{Loc} \to \mathsf{Emb}_1$ is a localisation of $\chi\mathsf{Loc}$ on the $\chi^+$-Cauchy morphisms $\mathsf{W}$, so $\mathsf{Emb}_1 \simeq \chi\mathsf{Loc}[\mathsf{W}^{-1}]$; see \cite[Chapter I]{GabrielZisman1967} and \cite[Chapter 7]{KashiwaraSchapira2006} for general information about localisations of categories, and \cite[Lemma 3.20]{BeniniSchenkelWoike2021} for their interaction with orthogonality relations.
In this case, our formulation of chiral CFTs in terms of $\chi\mathsf{Loc}$ would be equivalent to the established formulation in terms of $\mathsf{Emb}_1$.

If $Q$ is not such a localisation, then our proposed formulation of chiral CFTs may be meaningfully new; it would remain to determine whether it is of any greater physical interest than the established formulation.

% END

\paragraph*{Acknowledgments}
We thank Matthew Headrick and tslil clingman for useful discussions and comments.
This work has been supported by grants DE-SC0020194 and DE-SC0009986 from the U.S. Department of Energy.

\paragraph*{Declarations}
\begin{itemize}
	\item Data sharing not applicable to this article as no datasets were generated or analysed during the current study.
	
	\item The author has no relevant financial or non-financial interests to disclose.
\end{itemize}

\printbibliography

\end{document}